\documentclass[11pt]{article}

\usepackage{amsmath,amsfonts,amssymb,amscd,mathrsfs,amsthm} 
\usepackage{array}
\usepackage{bbm}
\usepackage{paralist}
\usepackage{xcolor}
\usepackage{multicol}
\usepackage{multirow}
\usepackage{bbm}
\usepackage{graphicx}
\usepackage{booktabs}
\usepackage{subfigure}
\usepackage{grffile}
\usepackage{float}   
\usepackage{mathrsfs}
\usepackage{cases}
\usepackage{lscape}
\usepackage[normalem]{ulem}
\usepackage{soul,xcolor}

\usepackage{clipboard}

\newcommand{\ba}{\begin{align}}
\newcommand{\ea}{\end{align}}

\newtheorem{proposition}{Proposition}
\newtheorem{assumption}{Assumption}
\newtheorem{lemma}{Lemma}
\newtheorem{example}{Example}

\newtheorem{theorem}{Theorem}[section]
\newtheorem{definition}{Definition}

\usepackage{mathtools}
\usepackage{natbib}
\usepackage[font=small,labelfont=bf,labelsep=quad]{caption}
\usepackage[colorlinks,citecolor=blue,urlcolor=blue,linkcolor=blue]{hyperref}
\usepackage{setspace}
\usepackage[margin=1in]{geometry}

\renewenvironment{thebibliography}[1]{%
\begin{oldthebibliography}{#1}%
\setlength{\baselineskip}{.8em}
\linespread{1.5}
\small
\setlength{\parskip}{0ex}%
\setlength{\itemsep}{.05em}%
}%
{%
\end{oldthebibliography}%
}

\mathtoolsset{
  showonlyrefs,
  mathic 
}

\begin{document}
\setstcolor{red}
\onehalfspacing

\title{{Non}-Concave Utility Maximization with Transaction Costs}
\author{Shuaijie Qian\thanks{Department of Mathematics, The Hong Kong University of Science and Technology, Hong Kong. Email: sjqian@ust.hk. Qian acknowledges the support from the start-up grant of the Hong Kong University of Science and Technology (B000-0172-R9863). }
\and
Chen Yang\thanks{Department of Systems Engineering and Engineering Management, The Chinese University of Hong Kong, Hong Kong. Email: cyang@se.cuhk.edu.hk. Yang acknowledges the supports from the Hong Kong Research Grants Council ECS (24207621) and GRF (14207723).}
}

\date{}

\maketitle

\begin{abstract}
This paper studies a finite-horizon portfolio selection problem with non-concave terminal utility and proportional transaction costs, in which the commonly used concavification principle for terminal value is no longer applicable. We establish a proper theoretical characterization of this problem via a two-step procedure. First, we examine the asymptotic terminal behavior of the value function, which implies that any transaction close to maturity only provides a marginal contribution to the utility. 
Second, we establish the theoretical foundation in terms of the discontinuous viscosity solution, incorporating the proper characterization of the terminal condition. 
Via extensive numerical analyses involving several types of utility functions, we find that the introduction of transaction costs into non-concave utility maximization problems can make it optimal for investors to hold on to a larger long position in the risky asset compared to the frictionless case, or hold on to a large short position in the risky asset despite a positive risk premium.
\end{abstract}

{\flushleft\textbf{Keywords:} utility maximization, portfolio selection, transaction costs, concavification principle}

\section{Introduction}

The utility maximization framework is widely used to study individuals' decisions in problems such as portfolio selection theory and consumer theory. For example, the classic Merton's problem (c.f. \cite{merton1975optimum}) studies the optimal portfolio selection, where an investor aims at maximizing the expected utility over terminal wealth and intertemporal consumption. 
In the classic utility maximization literature, the utility function is typically chosen as a concave function (e.g., CRRA or CARA utilities), which represents the individual's risk aversion. However, in many practical problems, the individual's utility has a non-concave dependence on the terminal wealth level. For example, the investor can have an investment objective and gain a sudden boost in her utility level if the wealth breaks through such an objective. This creates a jump discontinuity in the utility and makes it non-concave (see, e.g., the goal-reaching problem in Example \ref{exam 1m goal} and aspiration utility in Example \ref{exam 4 aspiration}). Another example is from the S-shaped utility in behavioral economics (see Example \ref{exam Sshape}). More examples can be found in delegated portfolio choice problems where the non-concavity can arise from the non-concave dependence of the payoff on the terminal portfolio value (see, e.g., \cite{carpenter2000does}, \cite{basak2007optimal}, \cite{he2018profit}).

The non-concave utility maximization is commonly tackled in the literature using the \textit{concavification principle}. Using this principle, the optimal investment strategy can be equivalently obtained by solving a ``concavified'' problem with the utility $U$ replaced by its concave envelope $\widehat{U}$. The basic idea behind this principle is that when the end of the investment horizon is imminent, it is optimal for the investor to avoid reaching a terminal wealth level $Z_T$ where $\widehat{U}(Z_T)>U(Z_T)$ by pushing for an unbounded (instantaneous) variance of portfolio value (c.f. \cite{browne1999reaching}). With one-sided portfolio bounds, \cite{bian2019utility} show that this principle still remains valid, since the investor can still establish unbounded portfolio variance in the convex direction. However, in the case with two-sided portfolio bounds that directly prohibit unbounded variance, \cite{dai2022nonconcave} prove that the concavification principle no longer holds, and they show that the non-concavity of the terminal utility has significant impacts on the investor's strategy, both theoretically and practically. For example, the investor may choose to gamble by short-selling stocks of positive risk-premium, or take extreme positions that attain the portfolio bounds and deviate significantly from the frictionless optimum.

The transaction costs are widely used to capture the trading fees or general asset illiquidity; see the large body of literature studying portfolio selection under transaction costs (e.g. \cite{magill1976portfolio,davis1990portfolio,shreve1994optimal}).
With transaction costs, the investor trades less actively since they need to pay the cost proportional to the trading amount. Transaction costs can impact non-concave utility maximization problems from both theoretical and financial perspectives, which motivates this paper. From the theoretical perspective, transaction costs make it prohibitively costly to achieve unbounded portfolio variance, since buying or short-selling a large amount of stocks drags down the portfolio value significantly due to transaction costs. As a result, this intuition suggests that the concavification principle is inapplicable in the presence of transaction costs.
From the financial perspective, the existing literature studying delegated portfolio selection under non-concave objectives (e.g. \cite{carpenter2000does,grinblatt1989adverse}) hints at the possible impact of transaction costs on the portfolio strategy without explicitly taking them into consideration in their models. While recent papers (e.g. \cite{dai2022convex,luo4191968mutual}) study such impact from an empirical perspective and also the implication of transaction costs on funds' liquidity premia, our interest focuses on the quantitative impact of transaction costs on the structure of the optimal portfolio strategy.

In this paper, we confirm the inapplicability of the concavification principle in the non-concave utility maximization problems with proportional transaction costs via a rigorous theoretical characterization.
To our best knowledge, this is the first paper studying continuous-time non-concave utility maximization problems with transaction costs. In stark contrast to the compulsory two-sided portfolio bounds in  \cite{dai2022nonconcave}, transaction costs impose ``soft'' bounds, making 
it diminishingly worthwhile for investors to transact and increase portfolio variance, 
as the end of the investment horizon approaches. Therefore, while the transaction costs in our setting and the two-sided portfolio bounds in \cite{dai2022nonconcave} both result in the inapplicability of the concavification principle, our work offers novel insights based on different fundamental reasoning and different underlying economic intuitions.

Our main contribution is two-fold. 
First, from the theoretical perspective, given the intrinsic discontinuity of the value function near the end of the horizon, we provide a rigorous treatment of the asymptotic terminal value in the singular control framework. Using it as the proper characterization of the terminal condition, we characterize the investor's optimal value as the unique viscosity solution of the corresponding Hamilton–Jacobi–Bellman (HJB) quasi-variational inequalities (see Definition \ref{def visco} and Theorem \ref{thm:viscosity solution}). The key step lies in the verification of the proposed terminal condition {(see Theorem \ref{prop:term})}. 
Although \cite{dai2022nonconcave} also propose a definition of discontinuous viscosity solution under the regular stochastic control framework, their terminal condition relies heavily on the homotheticity of the value function, which reduces the problem dimension to one. In comparison, our verification of the terminal condition is based on a delicate mathematical analysis of the impact of transactions on the terminal utility, which does not rely on homotheticity. 
Therefore, our approach has the potential to be extended to more general market frictions.
Our theoretical terminal condition is also motivated by and confirms the fundamental role of transaction costs in the inapplicability of the concavification principle.

Second, we perform numerical analyses based on several types of utility functions widely used in the literature, leading to many intriguing financial insights for the optimal portfolio strategy. 
When the remaining investment horizon is short, it can be optimal for the investor to hold on to a larger position in the risky asset compared to the frictionless case. Intuitively, since it is no longer feasible to push for an infinite portfolio variance immediately before the end of the horizon due to transaction costs, the investor maintains a larger position in advance as preparation. Interestingly, the corresponding leverage reported in Section \ref{sect:numerical} can be much higher than \cite{dai2022nonconcave}, where the leverage is both caused and limited by the portfolio bounds. 
Also, holding on to a large short position of the risky asset can be optimal despite its positive risk premium, since the positive effect of the large portfolio variance dominates the negative effect of the drift.

\textbf{Related Literature. } With concave utilities, there is a large body of literature studying the continuous-time utility maximization problems with transaction costs, starting from the seminal papers \cite{magill1976portfolio,davis1990portfolio, shreve1994optimal}. They found that the transaction costs, however small, virtually prohibit continuous portfolio rebalancing for optimal diversification. Instead, the investor should strike a balance between achieving the optimal risk exposure and diversification and minimizing the transaction costs. The transaction costs have since been widely used to model the bid-ask spread in a limit order book (e.g. \cite{kallsen2010using,gerhold2014transaction}), or to model the general liquidity cost when trading in an illiquid market (e.g. \cite{dai2019does,dai2011illiquidity}). Transaction costs also have been widely studied in portfolio selection (e.g. \cite{liu2002optimal,liu2004optimal}), in the explanation of liquidity premium (e.g. \cite{constantinides1986capital,vayanos1999equilibrium,gerhold2014transaction}), and in derivative pricing (e.g. \cite{davis1993european,kallsen2015option}).

Our proposed model is related to the literature studying non-concave utility functions. In addition to the goal-reaching utility (e.g. \cite{browne1999reaching}), aspiration utility (e.g. \cite{diecidue2008aspiration,aristidou2021rolling}), and  S-shaped utility (\cite{kah1979prospect,jin2008behavioral,he2018profit}) that will be discussed in details later, the non-concave utility functions are also widely used for modeling general objectives related to the distribution of wealth (e.g. \cite{he2011portfolio,he2020portfolio}).
 
Our result is also linked to the notion of viscosity solutions (see \cite{crandall1992user,fleming2006controlled,pham2009continuous}). 
Our definition of viscosity solution admits discontinuity that arises naturally as a joint result of the nonconcave utility function and the transaction costs.
\cite{altarovici2017optimal} consider the portfolio selection problem with both fixed and proportional transaction costs and smooth, concave utilities. Their derived value function may be discontinuous, and it is the unique viscosity solution up to a semicontinuous envelope.  Our definition is mostly close to \cite{dai2022nonconcave}, but their techniques for verifying the terminal boundary condition fail here, and we derive our condition by delicate analysis. 

The remainder of this paper is organized as follows. Section \ref{sect:model} describes the basic model setup and the assumptions. Section \ref{sect:theoretical} carries out the theoretical studies of the model, by characterizing the value function as the unique viscosity solution of the HJB equation, as well as identifying and proving the suitable terminal condition. Section \ref{sect:numerical} presents several numerical examples of our model and discusses their financial implications. Section \ref{sect:conclusion} concludes the paper.

\section{Model Setup}\label{sect:model}
We consider a finite investment horizon $T>0$ and assume that there are a risk-free asset (cash) and a risky asset (stock) in the market.
The cash position {grows at the} constant risk-free interest rate $r$ and the stock price follows
\begin{align}
\frac{dS_{t}}{S_t} = \mu(\nu_{t}) dt + \sigma (\nu_{t}) d \mathcal{B}_t,
\end{align}
where the expected stock return rate $\mu (\nu_t)\in \mathbb{R}$ and the stock volatility $\sigma(\nu_t)>0$ are assumed to be deterministic functions of the stochastic state variable $\nu_t$. 
We assume that $\nu_t$ follows 
\begin{align}\label{equ nu dyn}
d \nu_t = m(\nu_t) dt + \zeta(\nu_t) d \mathcal{B}^x_t. 
\end{align}
Here, $\{ (\mathcal{B}_t, \mathcal{B}^x_t)\}_{0\leq t \leq T}$ is a standard two-dimensional Brownian motion on a filtered probability space $(\Omega, \mathcal{F}, \{\mathcal{F}_t\}_{{0\le t \le T}}, {\mathbb{P}})$ with $\mathcal{B}_0 = \mathcal{B}^x_0 = 0$ and cross-variation $d[\mathcal{B}_t,\mathcal{B}^x_t] = \rho dt$. The filtration $\{\mathcal{F}_t\}_{{0\le t \le T}}$ is generated by this two-dimensional Brownian motion, and $\mathcal{F}_t$ contains all the {$\mathbb{P}$}-null sets of $\mathcal{F}$ for any $t\in [0, T]$.

Trading the stock incurs proportional transaction costs. We denote $\overline{X}_t$ and $\overline{Y}_t$ as the amount of wealth in the cash and stock, respectively. Let $\theta_1 \in (0,1)$ and $\theta_2 \in (0, +\infty)$ be the rates of the proportional
costs incurred on the stock sale and purchase, respectively. The dynamics of $\overline{X}_t$ and $\overline{Y}_t$, $0\leq t \leq T$, are 
\begin{align}
	\begin{cases}
	d \overline{X}_{t}  =  r \overline{X}_{t} dt -  (1+\theta_2)d \overline{L}_t + (1-\theta_1) d\overline{M}_t, \notag\\
	d \overline{Y}_{t} = \mu(\nu_t) \overline{Y}_t dt + \sigma(\nu_t) \overline{Y}_t d \mathcal{B}_t +d \overline{L}_t - d\overline{M}_t, 
	\end{cases}
\end{align}
where $\overline{L}_t$ and $\overline{M}_t$ represent the cumulative dollar amounts of stock purchase and sale, respectively.  They are both right-continuous with left limits, non-negative,non-decreasing, and adapted to {$\{\mathcal{F}_t\}_{0\le t\le T}$}. 

As \cite{dai2022nonconcave}, we consider the forward wealth in cash and stock, which are defined as 
\begin{align}
{X}_t = e^{-r(t-T)} \overline{X}_t,\quad  {Y}_t = e^{-r(t-T)} \overline{Y}_t.	
\end{align}
Then  
\begin{align}
\begin{cases}
	d {X}_{t}  =  - (1+\theta_2)d {L}_t + (1-\theta_1) d{M}_t, \\
	d {Y}_{t} = \eta(\nu_t) {Y}_t dt + \sigma(\nu_t) {Y}_t d \mathcal{B}_t + d {L}_t - d{M}_t, \label{equ:dyna XY}  
\end{cases}
\end{align}
where $\eta(\nu_t): = \mu(\nu_t) -r$ is the excess rate of return, and ${L}_t = \int_0^t e^{-r(u-T)} d \bar{L}_u$,  ${M}_t = \int_0^t e^{- r(u-T)} d \bar{M}_u$. 

\subsection{The Investor's Problem}
Denote by {$\{Z_t\}_{0\le t\le T}$} the forward wealth process, i.e., 
\begin{align}
Z_t = X_t+(1-\theta_1){Y}_t^+ - (1+\theta_2){Y}_t^- , 
\end{align}
where $Y_t^+ := \max \{0, Y_t\}$ and $Y_t^- := \max\{0, -Y_t\}$ are the positive and negative parts of $Y_t$, respectively. 
Furthermore, we impose an exogenous liquidation boundary $K\geq 0$. If the forward wealth $Z_t$ is no greater than $K$ at some time point $t$, we require the stock position to be immediately liquidated and the account to be closed, then the investor can only hold cash in $[t, T]$. For example, the liquidation boundary $K$ can be interpreted as a benchmark value for the wealth over the investment period $[0,T]$, adjusted by the interest growth (e.g., a percentage of the past high water-mark \cite{hodder2007incentive}). Setting $K=0$ leads to the usual bankruptcy constraint as in \cite{shreve1994optimal}.

The solvency region is 
\begin{align}
\mathscr{S} = \{(x, y, \nu)\in \mathbb{R}^3| x+(1-\theta_1){y}^+ - (1+\theta_2){y}^- \geq K  \}.
\end{align}  
Given an initial time $t\in [0, T]$ and position {$(X_{t-},Y_{t-}, \nu_{t-})=(x, y, \nu) \in \mathscr{S}$}, an investment strategy $(L_s, M_s)_{t\leq s \leq T}$ is admissible if $(X_s, Y_s, \nu_s)$ given by \eqref{equ nu dyn} and \eqref{equ:dyna XY} is in $\mathscr{S}$ for all $s\in[t, T]$. Denote by $\mathcal{A}_t(x, y, \nu)$ the set of all admissible strategies with initial time $t$ and initial position $(x, y, \nu)$. 

The investor's objective is choosing an admissible strategy to maximize the expected terminal utility over $Z_T$, i.e., 
\begin{align}
\sup \limits_{(L_s, M_s)_{0\leq s \leq T} \in \mathcal{A}_0(x, y, \nu)} {\mathbb{E}}_0^{x, y, \nu}[U(Z_T)],
\end{align} 
subject to \eqref{equ nu dyn} and \eqref{equ:dyna XY}, where {$\mathbb{E}_t^{x, y, \nu}$ denotes} the conditional expectation given {$X_{t-} = x$, $Y_{t-} = y$, and $\nu_{t-} = \nu$}. {Finally,  $U(\cdot)$ is the utility function, which satisfies the following assumption throughout this paper.} 

\begin{assumption}\label{assmpt:1}
The utility function $U(w), w \geq K$ is nondecreasing, right-continuous, and bounded from above by a function $C_1+C_2 z^p$ for some constants $C_1>0, C_2>0$ and $0<p<1$. Moreover, this utility function satisfies either of the following two conditions at $w=K$: \\
(1) $U(K)>-\infty$;\\
(2) $K=0$ and there are constants $\epsilon_w>0$, $\widetilde{p} \leq 0$, $A_1>0$, and $A_2 \in \mathbb{R}$ such that $U(w)=A_1 \frac{w^{\widetilde{p}}-1}{\widetilde{p}}+A_2$ for $0<w<\epsilon_w$ and $U(w) \geq A_1 \frac{w^{\widetilde{p}}-1}{\widetilde{p}}+A_2$ for all $w>0$.\footnote{Note that if $\widetilde{p}=0$, we set $U(w)=A_1 \ln (w)+A_2$ for $0<w<\epsilon_w$. In this case, the proofs of Theorems \ref{thm compa principle} and \ref{thm:viscosity solution} are not affected since the $\log$ utility also permits a closed-form solution when there are no transaction costs.   
}
\end{assumption}  
This assumption follows \cite{dai2022nonconcave}. It is broad enough to cover all CRRA utilities, and $U(z)$ is quite flexible in any compact subset of $[K,+\infty)$.

The following are some examples of non-concave utility functions satisfying this assumption. 
\begin{example}\label{exam 1m goal}
{\bf Goal-Reaching Utility. }  \cite{browne1999reaching} considers a fund manager whose objective is {maximizing the probability that the portfolio value $z$ beats some benchmark of $\bar{z}$  in a given finite time horizon}. Then the corresponding utility function is 
\begin{align}
	U(z) = \mathbf{1}_{z\geq \bar{z}},
\end{align}
where $\bar{z}$ is the benchmark. This utility function is discontinuous at $z=\bar{z}$ and hence non-concave.
\end{example}

\begin{example}\label{exam 4 aspiration}
{\bf The Aspiration Utility.} \cite{diecidue2008aspiration,aristidou2021rolling} study the type of discontinuous utility functions
\begin{align}\label{eqn:aspUtility}
U(z)  = \begin{cases}
		\frac{z^p}{p} & \text{if } z< \bar{z},\\
		c_1+c_2 \frac{z^p}{p} & \text{if } z\geq \bar{z}.
	\end{cases}
\end{align}
	Here, 
	 $p<1$ indicates the risk aversion level,
	$c_1>0$ and $c_2$ are constants such that $U(\bar{z}-)<U(\bar{z})$, and $\bar{z}$ denotes the aspiration level. 
	As a result, the utility function $U$ has an upward jump at $\bar{z}$, meaning that the investor achieves a boost in her utility once the wealth reaches $\bar{z}$. For example, this can be due to a change in the investor's social status. More theoretical and empirical evidence can be found in the above two papers.
\end{example}
\begin{example}\label{exam Sshape}
{\bf The S-shaped Utility of Prospect Theory. } \cite{kah1979prospect} consider the following S-shaped utility function: 
\begin{equation}
	{U(z)} = 
	\begin{cases}
		(z-z_0)^p  & \text{if}\ z > z_0	\\
		-\lambda (z_0-z)^p & \text{if}\ z \leq z_0,
	\end{cases}\label{eqn:Sutil}
\end{equation}
where $z_0$ is the wealth at time $0$ to distinguish gains from losses, $p \in (0, 1)$ since the investor is risk-averse over gains, and $\lambda>1$ because the pain from loss is higher than the pleasure from the same amount of gain.  
\end{example}
{In the above examples, the liquidation boundary $K$ can be any nonnegative real number, since the utility function is well-defined on $[K, \infty)$ for any $K\geq 0$.}

The following assumption is on the dynamics of stock prices. 

\begin{assumption}\label{ass:2}
1. The deterministic functions $\mu(\nu)$, $\sigma(\nu)>0$, $m(\nu)$, and $\zeta(\nu)>0$ are globally Lipschitz continuous in $\nu$ with Lipschitz constant $L$. 

2. Given $\nu \in \mathbb{R}$, $\overline{N}>0$  there are constants $\delta_{\nu, \bar{N}}$, $\widetilde{C}_{\nu, \bar{N}}$, $t_{\nu, \bar{N}}$, which are  locally uniformly positive and bounded with respect to $\nu$ and $\bar{N}$, 
\footnote{
That is, for any constant $M>0$, there is a constant $\epsilon_M$, such that for any $|\nu|, \overline{N} \leq M$, we have $\delta_{\nu, \bar{N}}$, $\widetilde{C}_{\nu, \bar{N}}$, $t_{\nu, \bar{N}}\in [\epsilon_M, 1/\epsilon_M]$.  
}
such that
for any $N>\overline{N}$ and $t\in [t_{\nu, \bar{N}}, T]$, $$\mathbb{P} \left( \max \limits_{t\leq s \leq T} |\nu_s -\nu| \geq N \bigg| \nu_t = \nu \right) \leq {\widetilde{C}_{\nu, \bar{N}}} e^{- \delta_{\nu, \bar{N}} \frac{N^2}{T-t}}.$$ 

3. For the case without transaction costs, there exists a classical solution for $U(w)=\frac{w^{\hat{p}}}{\hat{p}}, 0<\hat{p}<1$, in a time interval $\left(t_{\hat{p}}, T\right]$, where $t_{\hat{p}}\in [-\infty, T)$ is a continuous and increasing function w.r.t. $\hat{p}$ for $0<\hat{p}<1$. We assume the value function has the form 
$$
V(t, z, \nu)=\frac{z^{\hat{p}}}{\hat{p}} F^{(\hat{p})}(t, \nu),
$$
where $z$ is the initial wealth at time $t$. Moreover, the function $F^{(\hat{p})}(t, \nu)$ has the form
\begin{align*}
F^{(\hat{p})}(t, \nu):=e^{A_{\hat{p}}(t) \nu^2 +B_{\hat{p}}(t) \nu +C_{\hat{p}}(t)},
\end{align*}
where $F^{(\hat{p})}(t, \nu)$ is finite for $t_{\hat{p}}<t \leq T$,  $A_{\hat{p}}(t)>0, F^{(\hat{p})}(t, \nu)$ are strictly decreasing in $t$, and
\begin{align*}
\left\{
\begin{array}{>{\displaystyle}l}
\lim _{q \rightarrow \hat{p}+} A_q(t)=A_{\hat{p}}(t) \\
\lim _{q \rightarrow \hat{p}+} B_q(t)=B_{\hat{p}}(t) \\
\lim _{q \rightarrow \hat{p}+} C_q(t)=C_{\hat{p}}(t), \\
A_q(t)>A_{\hat{p}}(t), \quad \text { when } q>\hat{p},\ t_q<t \leq T .
\end{array}
\right.
\end{align*}
\end{assumption}
Assumption \ref{ass:2} covers the following Gaussian mean return model studied in \cite{kim1996dynamic}\footnote{ We will verify that it satisfies Assumption \ref{ass:2}.2 in \ref{sec GMR ass2.2}; the verification for Assumption \ref{ass:2}\textcolor{blue}{.3} follows from \cite{liu2007portfolio}. Depending on the parameters, the Gaussian mean return model may either explode in finite time ($t_{\hat{p}} > -\infty$) or never explode in finite time ($t_{\hat{p}} = -\infty$); see \cite{kim1996dynamic}. The geometric Brownian motion never explodes in finite time ($t_{\hat{p}}=-\infty$).}
\begin{align}\label{param:GMR}
\mu(\nu) = r+ \sigma \nu, ~\sigma(\nu) \equiv \sigma, ~m(\nu) = \kappa (\bar{\nu} - \nu), ~\zeta(\nu) \equiv \zeta,
\end{align}
which captures the effect of time-varying risk premium.
Assumption \ref{ass:2} also covers the geometric Brownian motion model with
\begin{align}\label{param:GBM}
	\mu(\nu)\equiv r+\eta,~\sigma(\nu)\equiv \sigma,
\end{align}
where $\eta$ is the risk premium. 

Note that Assumption \textcolor{blue}{\ref{ass:2}.3} allows for a wider class of models beyond the Gaussian mean return model (see, for example, the Heston model and the CIR model in  \cite{liu2007portfolio}). Our main restriction on the model setup is from Assumptions \textcolor{blue}{\ref{ass:2}.1} and \textcolor{blue}{\ref{ass:2}.2}.

\section{Theoretical Analysis}\label{sect:theoretical}
In the following, we denote 
\begin{align}\label{eqn:z}
z: = x+(1-\theta_1){y}^+ - (1+\theta_2){y}^-.
\end{align}
We define the value function by 
\begin{align}\label{eqn:vf}
V(t, x, y, \nu) = \sup \limits_{(L_s, M_s)_{t\leq s \leq T} \in \mathcal{A}_t(x, y, \nu)} {\mathbb{E}}_t^{x, y, \nu}[U(Z_T)]
\end{align} 
for $(x, y, \nu) \in \mathscr{S}$, $0\leq t \leq T$. Formally, in the interior of $[0, T]\times \mathscr{S}$, $V(t, x, y, \nu)$ satisfies the following HJB equation
\begin{align}\label{equ PDE main}
\begin{split}
&\mathcal{L} V: = \min\bigg\{-V_t -\frac{1}{2}\bigg(\sigma^2(\nu) y^2 V_{yy} + \zeta^2(\nu) V_{\nu\nu} + 2 \rho \sigma(\nu) y \zeta(\nu) V_{y\nu} \bigg) - \eta(\nu) y V_y-m(\nu) V_\nu ,\\
&\qquad \qquad \qquad  V_y - (1-\theta_1)V_x, \quad (1+\theta_2)  V_x - V_y   \bigg\} = 0,
\end{split}
\end{align}
with the following boundary condition  for $x_0+ (1-\theta_1) y_0^+ - (1+\theta_2) y_0^- = K$,  
\begin{align}\label{equ bound condition}
\begin{cases}
\lim \limits_{(s, x, y, \nu) \rightarrow(t, x_0, y_0, \nu_0)} V(s, x, y, \nu)=U(K) & \text { if } U(K)>-\infty \\ 
\lim \limits_{(s, x, y, \nu) \rightarrow(t, x_0, y_0, \nu_0)}  V(s, x, y, \nu)-V_{CRRA}(s, x, y, \nu)=0 & \text { if } K=0 \text { and } U(0)=-\infty,
\end{cases}
\end{align}
where $V_{CRRA}(s, x, y, \nu)$ is the value function when the terminal utility is $U(w) = A_1 \frac{w^{\widetilde{p}}-1}{\widetilde{p}}+A_2$, $\forall w\geq 0$, and initial state $(x_s, y_s, \nu_s) = (x, y, \nu)$.
In the special case where $S$ follows a geometric Brownian motion, all derivatives of $V$ with respect to $\nu$ vanish from the equation \eqref{equ PDE main}, which then reduces to 
\begin{align*}
	\min\left\{-V_t -\frac{1}{2}\sigma^2 y^2 V_{yy} - \eta y V_y ,~~ V_y - (1-\theta_1)V_x,~~ (1+\theta_2)  V_x - V_y   \right\} = 0. 
\end{align*}

\subsection{Terminal Condition}
Without transaction costs, the investor can push for infinite instantaneous portfolio variance near the terminal time, and the concavification principle holds. Therefore, the terminal utility can be replaced by its concave envelope, which is continuous. With portfolio bounds that put a hard constraint on the leverage and limit the portfolio variance, the concavification principle is proved to be invalid by \cite{dai2022nonconcave}, but the intuition of taking the maximum allowed leverage around terminal time still holds.

By introducing transaction costs, the behavior of the value function and the strategy near terminal time become more intriguing. While the investor has the incentive to take a large portfolio variance allowed as mentioned above, transaction costs virtually prohibit the investor from taking infinite variance. Consequently, the concavification principle becomes no longer applicable in the presence of transaction costs. Intuitively, compared to the hard constraint on leverage imposed by the portfolio bound, transaction costs impose a ``soft'' constraint. Indeed, the following theorem characterizes the asymptotic behavior of the value function as the time approaches maturity $T$, which confirms that the value function can be discontinuous if it is both close to the terminal in time and close to the jump points of the utility function.

\begin{theorem}\label{prop:term}
{The value function $V$ defined in \eqref{eqn:vf} satisfies}
\begin{align} \label{equ: tercon}
	\lim  \limits_{(t, x, y, \nu) \to ( T-, \hat{x}, \hat{y}, \hat{\nu})} V(t, x, y, \nu) - U(\hat{z}-)-2 \Phi \bigg(\frac{\min\{z-\hat{z}, 0\}}{ |\hat{z}-x| \sigma(\hat{\nu})\sqrt{T-t} } \bigg) \left(U(\hat{z}) - U(\hat{z}-)\right)= 0,
\end{align}
where $U(\hat{z}-)$ is the left limit of $U$ at $\hat{z}$, $U(K-)$ = $U(K)$, $\Phi$ is the standard normal cumulative distribution function, and $\hat{z}$ is defined by \eqref{eqn:z} with $(x,y,z)$ replaced by $(\hat{x},\hat{y},\hat{z})$.
In the case $|\hat{z}-x| = 0$, we set 
\begin{numcases}{\Phi\bigg(\frac{\min\{z-\hat{z}, 0\}}{ |\hat{z}-x| \sigma(\hat{\nu})\sqrt{T-t} }\bigg) = }
	0 & \text{when }$z<\hat{z}$, \notag\\
	1 & \text{when }$z\geq \hat{z}$. \notag
\end{numcases} 
\end{theorem}

The proof of Theorem \ref{prop:term} will be relegated to \ref{sect:propterm}. The proof is significantly different from \cite{dai2022nonconcave} from the technical perspective. Indeed, they verify the discontinuous terminal condition by reducing the original problem to a one-dimensional problem. However, such a kind of homotheticity does not apply here. 
Instead, we directly estimate the contribution of transactions in the total value function by delicate mathematical analysis, and we show that when the time is close to maturity, this contribution is marginal, if not negative. The technical difficulty lies in the arbitrariness of trading strategy, and we are able to build a uniform estimation over all trading strategies. Our approach has the potential to be extended to more general market frictions.

Intuitively, to make the wealth increase to the threshold $\hat{z}$, either the investor needs to hold a large (positive or negative) amount of stock, or the stock price needs to be sufficiently fluctuating. But the value of holding a large amount of stock is eroded by the transaction costs, while the contribution of stock price fluctuation is smaller and smaller as time approaches maturity.  In this way, we can show that the value of the transaction is marginal. 

In the following, we provide the intuitions on the terminal condition \eqref{equ: tercon}.
In the special case $U(\hat{z}-) = U(\hat{z})$, i.e., $U$ is continuous around $\hat{z}$, \eqref{equ: tercon} degenerates to 
\begin{align}
	\lim  \limits_{(t, x, y, \nu) \to ( T-, \hat{x}, \hat{y}, \hat{\nu})} V(t, x, y, \nu) = U(\hat{z}), 
\end{align}
which implies a continuous terminal condition consistent with the classical definition. To elaborate on the more interesting case when $U(\hat{z}-)< U(\hat{z})$, we consider the goal-reaching problem by letting $U(z) = \mathbf{1}_{z\geq \hat{z}}$ with $\hat{z} = 1$. 
For the ease of exposition, we take as an example the case where $\mu$ and $\sigma$ are specified as \eqref{param:GBM}; similar intuition holds for the general case.
Consequently, \eqref{equ: tercon} degenerates into
\begin{align}
	\lim  \limits_{(t, x, y) \to ( T-, \hat{x}, \hat{y})} V(t, x, y) -2 \Phi \left(\frac{\min\{z-1, 0\}}{ |1-x| \sigma\sqrt{T-t} } \right) = 0. \label{equ exam goal term}
\end{align}
The equation \eqref{equ exam goal term} indicates the failure of the concavification principle, since this principle implies the boundary condition
\begin{align}
	\lim  \limits_{(t, x, y) \to ( T-, \hat{x}, \hat{y})} V(t, x, y)  = \hat{z}. 
\end{align}
We discuss \eqref{equ exam goal term} in two cases.
{\flushleft\textbf{When $z\geq 1$ with limit $\hat{z} = 1$:}} this equation becomes 
\begin{align}
	\lim  \limits_{(t, x, y) \to ( T-, \hat{x}, \hat{y})} V(t, x, y) - 1  = 0. 
\end{align}
The interpretation is that, when wealth is higher than the target, the investor can always liquidate the entire stock position and reach the goal. 

{\flushleft\textbf{When $z\leq 1$ with limit $\hat{z} = 1$:}} this equation becomes 
\begin{align}\label{equ terminal goal-reaching}
	\lim  \limits_{(t, x, y) \to ( T-, \hat{x}, \hat{y})} V(t, x, y) -2 \Phi \left(\frac{z-1}{ |1-x| \sigma\sqrt{T-t} } \right) = 0. 
\end{align}
In the limiting process, the second term has a singularity around $(T, \hat{x}, \hat{y})$, which cancels out the singularity of $V$ around this point. 
The second term is nothing but the leading term around $(T, \hat{x}, \hat{y})$ of the value function under the following strategy: the investor makes no transaction before reaching the goal and liquidates the entire stock position immediately after reaching the goal.   
Because when $T-t$ is short, the stock account wealth dynamic is approximately
\begin{align}
& d \widetilde{Y}_s =  \sigma y d \mathcal{B}_s,\ t\leq s \leq T, \quad \widetilde{Y}_t =  Y_t= y. \notag
\end{align}  
Therefore, $ \frac{\widetilde{Y}_s-y}{\sigma y }$ is a standard Brownian motion. Taking $x<1$ as an example, we have
\begin{align}
&\mathbb{P}( Z_T\geq1 ) = \mathbb{P} \bigg( \max\limits_{t\leq s \leq T} (1-\theta_1)Y_s \geq 1-x\bigg) \notag\\
 & \qquad \approx   \mathbb{P} \bigg( \max\limits_{t\leq s \leq T} (1-\theta_1)\widetilde{Y}_s  \geq 1-x \bigg) 
=  \mathbb{P} \bigg( \max\limits_{t\leq s \leq T}  \frac{\widetilde{Y}_s-y}{\sigma y } \geq \frac{1-z}{ (1-\theta_1) \sigma y  } \bigg). \notag
\end{align}
Since ${x}+(1-\theta_1)y \to 1$, 
\begin{align}
 \mathbb{P} \bigg( \max\limits_{t\leq s \leq T}  \frac{\widetilde{Y}_s-y}{\sigma y } \geq \frac{1-z}{ (1-\theta_1) \sigma y  } \bigg)\approx  \mathbb{P} \bigg( \max\limits_{t\leq s \leq T}  \frac{\widetilde{Y}_s-y}{\sigma y } \geq \frac{1-z}{ (1-x) \sigma  } \bigg) = 2 \Phi \bigg(\frac{z-1}{ (1-x) \sigma\sqrt{T-t} } \bigg). \notag
\end{align}
The case $x>1$ is analogous. 

It is worth clarifying that the condition \eqref{equ terminal goal-reaching} does not imply that the optimal strategy prohibits transactions before reaching the goal when $T-t$ is sufficiently small. On the contrary, numerical results in Section \ref{sect:numerical} illustrate that it can still be optimal to buy or sell in this case, which even leads to a very high leverage. 
The high leverage does not appear in the terminal condition \eqref{equ terminal goal-reaching} because when time to maturity $T-t$ is small, the contribution of transaction to the utility is marginal.

\subsection{Viscosity Solution and Comparison Principle}
In this subsection, we show that the value function $V(t, x, y)$ is the unique viscosity solution to the PDE problem \eqref{equ PDE main} with boundary condition \eqref{equ bound condition} and terminal condition  \eqref{equ: tercon}.     
{In the following, we simply refer to the PDE together with the boundary and terminal conditions as the HJB equation \eqref{equ PDE main} -- \eqref{equ: tercon}.}

We first introduce our notion of viscosity solution that incorporates the discontinuity implied by the terminal condition  \eqref{equ: tercon}.
Define the lower semicontinuous envelope and upper semicontinuous envelope of the value function $V$ as
\begin{align}
V_*(t, x, y, \nu)&=\liminf_{(t_1,x_1, y_1, \nu_1)\rightarrow(t, x, y, \nu)}V(t_1,x_1, y_1, \nu_1)\\
V^*(t, x, y, \nu)&=\limsup_{(t_1,x_1, y_1, \nu_1)\rightarrow(t, x, y, \nu)}V(t_1,x_1, y_1, \nu_1).
\end{align}

\begin{definition}[Viscosity Solution]\label{def visco}
(i).
We say that $V$ is a viscosity subsolution of the HJB equation \eqref{equ PDE main} -- \eqref{equ: tercon} if it satisfies the following conditions:
\\a) For all smooth $\psi$ such that $ V^* \leq\psi$ and
$V^*(\hat{t}, \hat{x}, \hat{y}, \hat{\nu})=\psi(\hat{t}, \hat{x}, \hat{y}, \hat{\nu})$ for some $(\hat{t}, \hat{x}, \hat{y}, \hat{\nu})\in [0, T)\times\mathscr{S}$, we have 
\begin{align}
\mathcal{L} \psi(\hat{t}, \hat{x}, \hat{y}, \hat{\nu}) \leq 0. \label{HJBsubt}
\end{align}
b) For all $0\leq \hat{t}< T$,
\begin{align}
\begin{cases}
\limsup \limits_{(t, x, y, \nu) \to (\hat{t}, \hat{x}, \hat{y}, \hat{\nu})}V^*(t, x, y, \nu) - U(K) \leq 0 & \text { if } U(K)>-\infty \\ 
\limsup \limits_{(t, x, y, \nu) \to (\hat{t}, \hat{x}, \hat{y}, \hat{\nu})}V^*(t, x, y, \nu)-V_{CRRA}(t, x, y, \nu)\leq0 & \text { if } K=0 \text { and } U(0)=-\infty.
\end{cases}
\end{align}
c) For all $w\geq K$,
\begin{align}
\limsup  \limits_{(t, x, y, \nu) \to ( T-, \hat{x}, \hat{y}, \hat{\nu})} V^*(t, x, y, \nu) - U(\hat{z}-)-2 \Phi \bigg(\frac{\min\{z-\hat{z}, 0\}}{ |\hat{z}-x| \sigma(\hat{\nu})\sqrt{T-t} } \bigg) \bigg(U(\hat{z}) - U(\hat{z}-)\bigg)\leq 0. \label{bTerValuecdsub}
\end{align}
(ii).
We say that $V$ is a viscosity supersolution of the HJB equation \eqref{equ PDE main} -- \eqref{equ: tercon} if it satisfies the following conditions:
\\a) For all smooth $\psi$ such that $ V_* \geq\psi$ and
$V_*(\hat{t}, \hat{x}, \hat{y}, \hat{\nu})=\psi(\hat{t}, \hat{x}, \hat{y}, \hat{\nu})$ for some $(\hat{t}, \hat{x}, \hat{y}, \hat{\nu})\in [0, T)\times\mathscr{S}$, we have 
\begin{align}
\mathcal{L} \psi(\hat{t}, \hat{x}, \hat{y}, \hat{\nu}) \geq 0. \label{HJBsupt}
\end{align}
b) For all $0\leq \hat{t}< T$, 
\begin{align}
\begin{cases}
\liminf \limits_{(t, x, y, \nu) \to (\hat{t}, \hat{x}, \hat{y}, \hat{\nu})}V_*(t, x, y, \nu) - U(K) \geq 0 & \text { if } U(K)>-\infty \\ 
\liminf \limits_{(t, x, y, \nu) \to (\hat{t}, \hat{x}, \hat{y}, \hat{\nu})}V_*(t, x, y, \nu)-V_{CRRA}(t, x, y, \nu)\geq0 & \text { if } K=0 \text { and } U(0)=-\infty.
\end{cases}
\end{align}
c) For all $w\geq K$,
\begin{align}
\liminf  \limits_{(t, x, y, \nu) \to ( T-, \hat{x}, \hat{y}, \hat{\nu})} V_*(t, x, y, \nu) - U(\hat{z}-)-2 \Phi \bigg(\frac{\min\{z-\hat{z}, 0\}}{ |\hat{z}-x| \sigma(\hat{\nu})\sqrt{T-t} } \bigg) \bigg(U(\hat{z}) - U(\hat{z}-)\bigg)\geq 0. \label{bTerValuecdsup}
\end{align}
 (iii).
  We say that $V$ is a viscosity solution if it is both a viscosity supersolution and subsolution.
\end{definition}

Define the set 
\begin{align}
\mathcal{C}: = &\bigg\{ v(t, x, y, \nu):(t_p, T]\times \mathscr{S} \to \mathbb{R} \bigg|  v(t, x, y, \nu) \leq C_2 (x+y)^p F^{(p)}(t, \nu) +C_1 \\
&\qquad  \text{and } 
v(t, x, y, \nu) \geq 
\begin{cases}
U(K), \qquad \qquad \ \ \text{ if } U(K)>-\infty,\\
V_{CRRA}(t, x, y, \nu),  \text{ if } K=0 \text { and } U(0)=-\infty
\end{cases} 
\bigg\}.
\end{align}
With the notion of viscosity subsolution (supersolution) in Definition \ref{def visco}, we have the following comparison principle. 
\begin{theorem}[Comparison Principle] \label{thm compa principle}
Assume that $u$ and $v$ are a viscosity subsolution and a supersolution to HJB equation 
\eqref{equ PDE main} -- \eqref{equ: tercon}, respectively. If $u$ and $v$ are both in $\mathcal{C}$, then $u\leq v$ in  $(t_p, T]\times\mathscr{S}$. 
\end{theorem}
The proof of this theorem will be given in \ref{sect:comparison}. The comparison principle is essential to guarantee that our definition is reasonable.  In the proof of this comparison principle, we pay special attention to the terminal condition, which differs from the classical proof.  

The following theorem summarizes our result. 
\begin{theorem}\label{thm:viscosity solution}
(i) There is at most one viscosity solution to \eqref{equ PDE main} -- \eqref{equ: tercon} in $\mathcal{C}$. 

(ii) The value function $V(t, x, y, \nu)$ is a viscosity solution to \eqref{equ PDE main} -- \eqref{equ: tercon} and $V\in\mathcal{C}$.

(iii) $V(t, x, y, \nu)$ is the unique viscosity solution to \eqref{equ PDE main} -- \eqref{equ: tercon} in $\mathcal{C}$.
\end{theorem}

The proof of this theorem will be given in \ref{sect:viscosity}.
Theorem \ref{thm:viscosity solution} (i) is from the comparison principle Theorem \ref{thm compa principle}. This is because any viscosity solution must be both a subsolution and a supersolution, then any two viscosity solutions must be equal. Also, Theorem \ref{thm:viscosity solution} (iii) is a direct corollary of (i) and (ii). The proof of Theorem \ref{thm:viscosity solution} (ii) includes two parts. The first part is to verify that $V$ satisfies Definition \ref{def visco}, i.e., it is both a subsolution and a supersolution.
The second part is to check the value function $V(t, x, y, \nu) \in \mathcal{C}$. Actually, this can be proved using Assumption \ref{assmpt:1}, which implies
\begin{align}
C_1+ C_2 (x+y)^p F^{(p)}(t, \nu) \geq V(t, x, y, \nu) \geq \begin{cases}
U(K), \qquad \qquad \ \ \text{ if } U(K)>-\infty,\\
V_{CRRA}(t, x, y, \nu),  \text{ if } K=0 \text { and } U(0)=-\infty.
\end{cases} 
\end{align}

\section{Numerical Examples}\label{sect:numerical}
Based on the above theoretical framework, in this section, we provide numerical examples that illustrate interesting financial insights. To this end, we numerically characterize the investor's optimal strategy through the action regions in terms of the dollar value in stock $y$ and the wealth $z$ (liquidation value of the portfolio) as defined in \eqref{eqn:z}.
From Section \ref{subsect:GR0} to Section \ref{subsect:asp}, we consider a stock $S$ following a geometric Brownian motion model with the risk premium $\eta=\mu-r=0$ or 0.04, and $\sigma=0.3$, and we assume $\bar{z}=1$, $K=0$, and $r=0$ for simplicity. In Section \ref{subset:GMR}, we explore a more general market environment where $S$ follows the Gaussian mean return model. In \ref{appendix:terminal}, we numerically verify the terminal condition \eqref{equ: tercon}  under various utility settings. 
For later use, we define the action regions, including the sell region (SR), the buy region (BR), and the no-trading region (NR), as follows:
\begin{align*}
SR =& \bigg\{(t, x, y, \nu)\in (t_p, T]\times \mathscr{S} \bigg| V_y = (1-\theta_1)V_x  \bigg\}, \\
BR =&  \bigg\{(t, x, y, \nu)\in (t_p, T]\times \mathscr{S} \bigg| V_y = (1+\theta_2)V_x \bigg\}, \\
NR=&  \bigg\{(t, x, y, \nu)\in (t_p, T]\times \mathscr{S} \bigg|   (1-\theta_1)V_x<  V_y < (1+\theta_2)V_x   \bigg\}.
\end{align*}
The optimal actions in these three regions are to sell the stock, to buy the stock, and to refrain from trading, respectively.

\subsection{Goal-Reaching Problem with Zero Risk Premium}\label{subsect:GR0}
To begin with, we consider the goal-reaching problem with the risk premium $\eta=0$. The action regions for $T-t=0.01,0.05,0.1$ are illustrated in the left panels in Figure \ref{fig goal reaching etas} from top to bottom, respectively. Note that due to the transaction costs, buying (resp. selling) decreases the wealth $z$, besides increasing (resp. decreasing) $y$.

These three figures show common patterns for the goal-reaching problem strategy with transaction costs. 
First, when the absolute dollar investment in stock $|y|$ is sufficiently small, it is optimal to buy to the dashed line when $y$ is positive and in the green region, i.e., buying an amount $\Delta y=y^\prime-y$ such that the point representing the post-purchase wealth and investment in stock, $(z^\prime,y^\prime)$, is located on the dashed line (e.g., Point A to B in the top left panel of Figure \ref{fig goal reaching etas}).
On the other hand, it is optimal to sell to the dash-dotted line when $y$ is negative and in the yellow region, i.e., selling an amount of $\Delta y=y-y^\prime$ such that the point representing the post-sale wealth and investment $(z^\prime,y^\prime)$ is located on the dash-dotted line (e.g., Point C to D in the top left panel of Figure \ref{fig goal reaching etas}) 
This strategy has two benefits: (a) it increases $|y|$, leading to a larger portfolio variance and hence a higher probability of reaching the goal; (b) the trading direction minimizes the transaction costs. For example, starting from the green region, buying to the dashed line results in a smaller trading amount (hence a smaller transaction cost) compared to selling to the dash-dotted line. 

Second, the investor will strike a balance between increasing the variance and reducing transaction costs. On the one hand, when $|y|$ is already sufficiently large, it is optimal to hold on to the current position without trading, as indicated by the no-trading region (the blue regions). This saves transaction costs since the current portfolio variance is sufficiently high. On the other hand, the investor will buy or sell fewer stocks when $z$ is very close to 0 or 1. Indeed, in these two cases, a high transaction cost drags down the portfolio value, increasing the liquidation risk (when $z$ is close to 0) or reducing the probability of reaching the already very close target (when $z$ is close to 1). Therefore, the investor will put more emphasis on reducing costs by buying or selling fewer stocks.

Third, as the remaining time horizon $T-t$ becomes shorter (from bottom to top), the trading region becomes larger. Indeed, the investor needs to achieve a larger portfolio variance so as to reach the goal when $T-t$ is smaller.

Finally, we compare the action regions to the optimal target position without transaction costs (frictionless target) in \cite{browne1999reaching} depicted as the grey solid line, i.e.,
\begin{align*}
  y = \frac{1}{\sigma\sqrt{T-t}}\phi(\Phi^{-1}(z)),~~z\in (0,1),  
\end{align*}
where $\phi$ is the standard normal probability density function.
\Copy{pushfor}{
We see that the upper boundary of the buy region (the dashed line) can be higher than the frictionless target, meaning that in the presence of transaction costs, the investor may establish a larger stock position compared to the case without transaction costs. 
To understand this, when there is no transaction cost, the frictionless target will go to infinity as $T-t$ gets very close to 0 for any given $z$; in other words, if the current position does not make it very close to the goal (or to liquidation), the investor can take on an even larger stock position in the remaining time before maturity. 
However, this is impossible in the presence of transaction costs since the dollar purchase amount should be bounded by $z/\theta_2$ to keep a positive portfolio value and avoid liquidation. 
As a result, the investor will establish a larger position before the maturity is imminent for compensation. 
Nevertheless, when $z$ is close to 1, the investor will still buy to a smaller position compared to the frictionless target in order to save cost.
}

\begin{figure}[hptb!]
\includegraphics[width=0.48\textwidth]{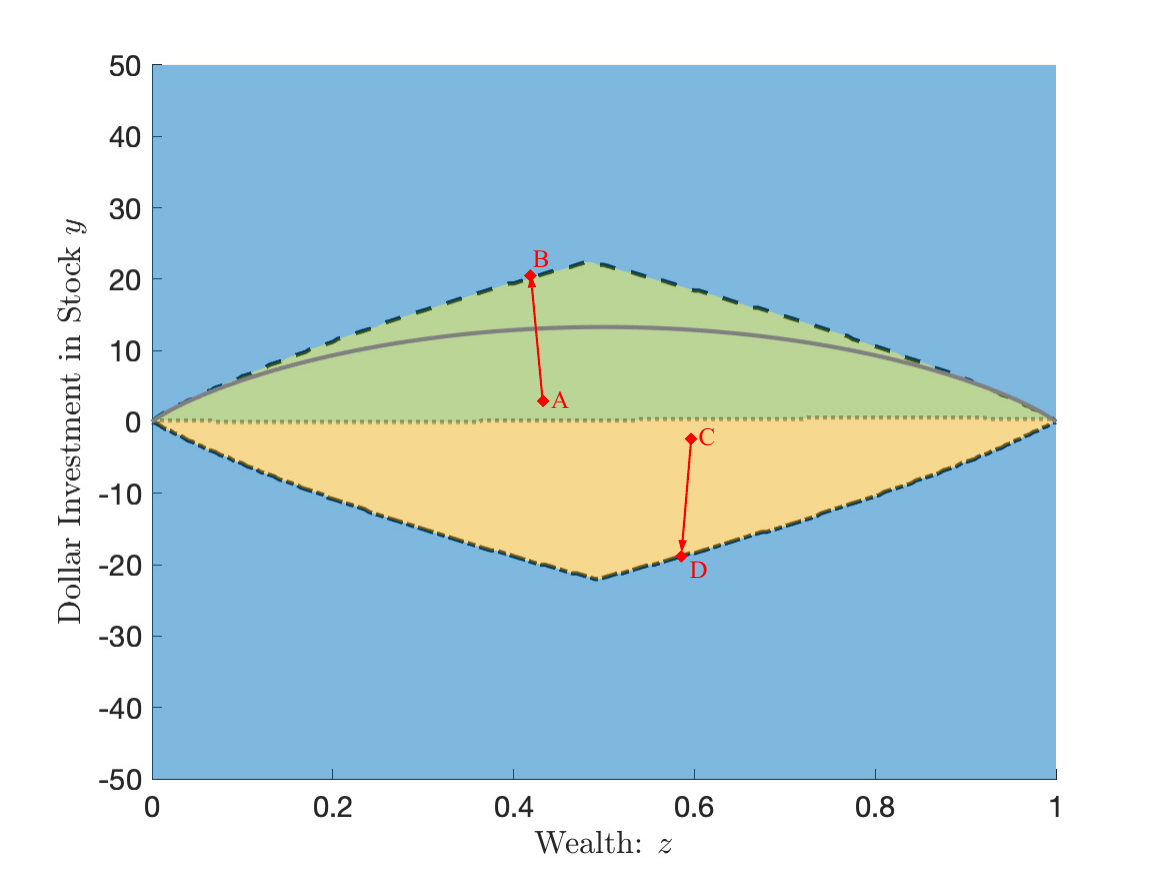}
\includegraphics[width=0.48\textwidth]{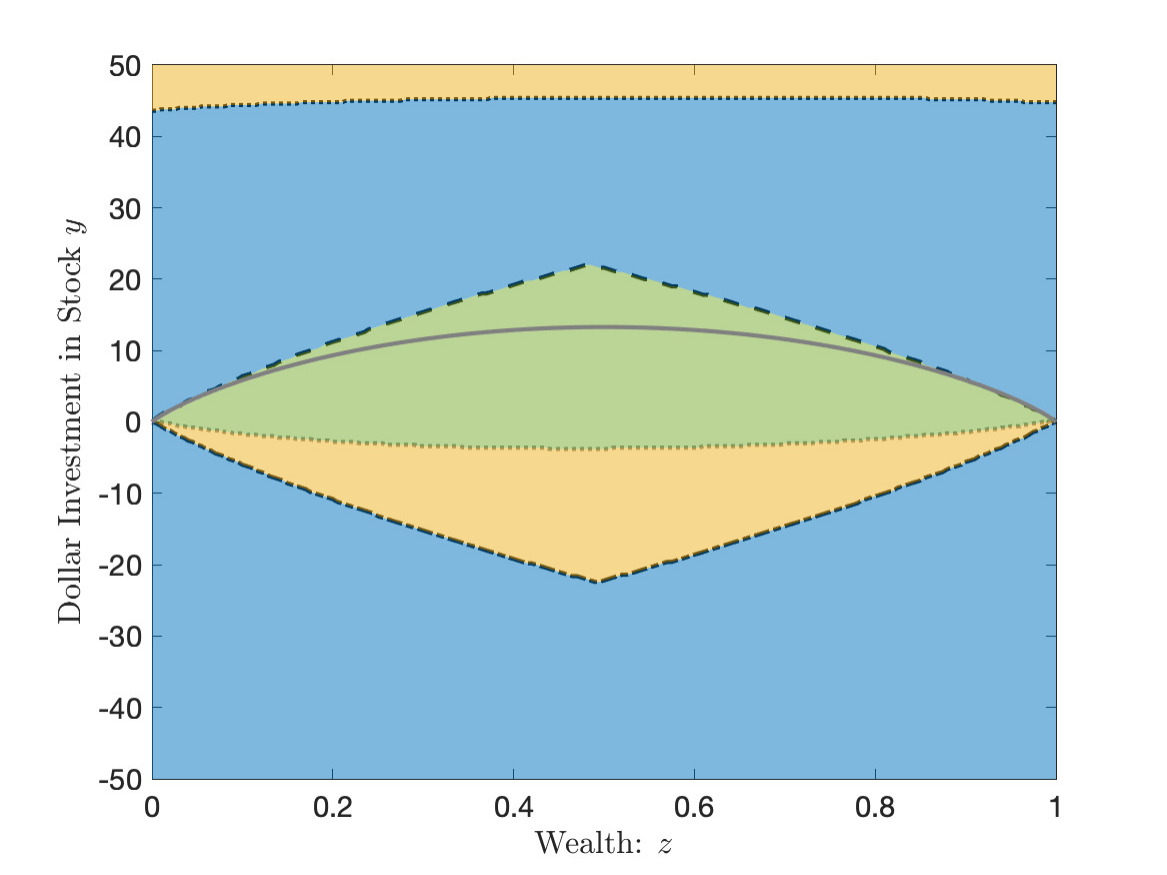}\\
\includegraphics[width=0.48\textwidth]{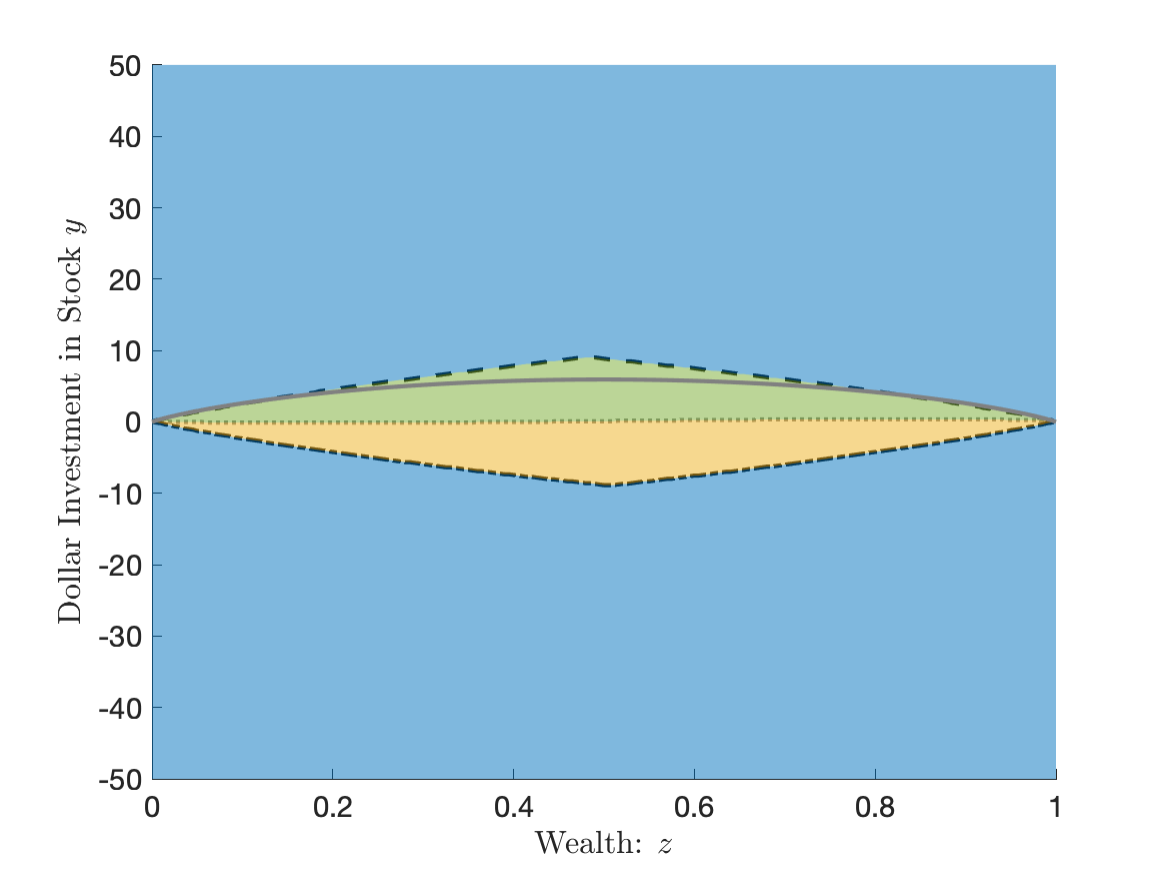}
\includegraphics[width=0.48\textwidth]{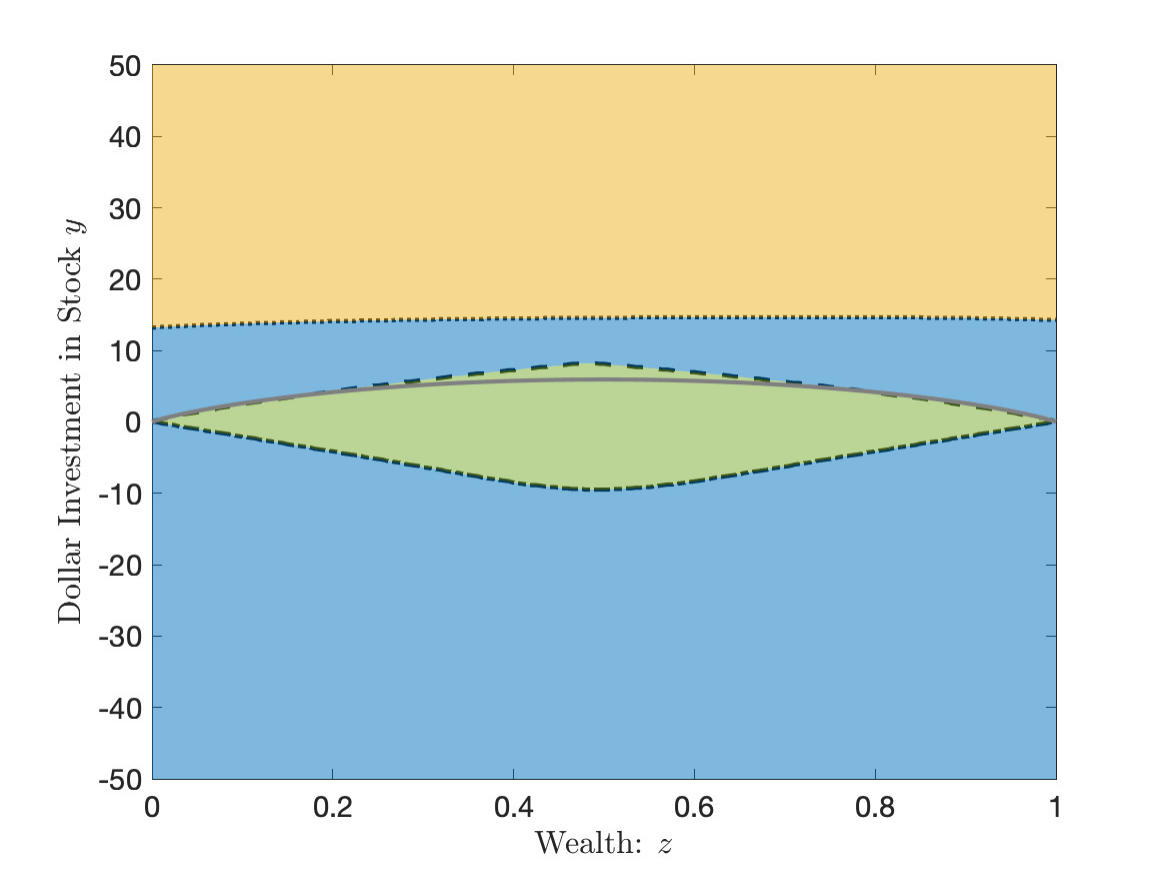}\\
\includegraphics[width=0.48\textwidth]{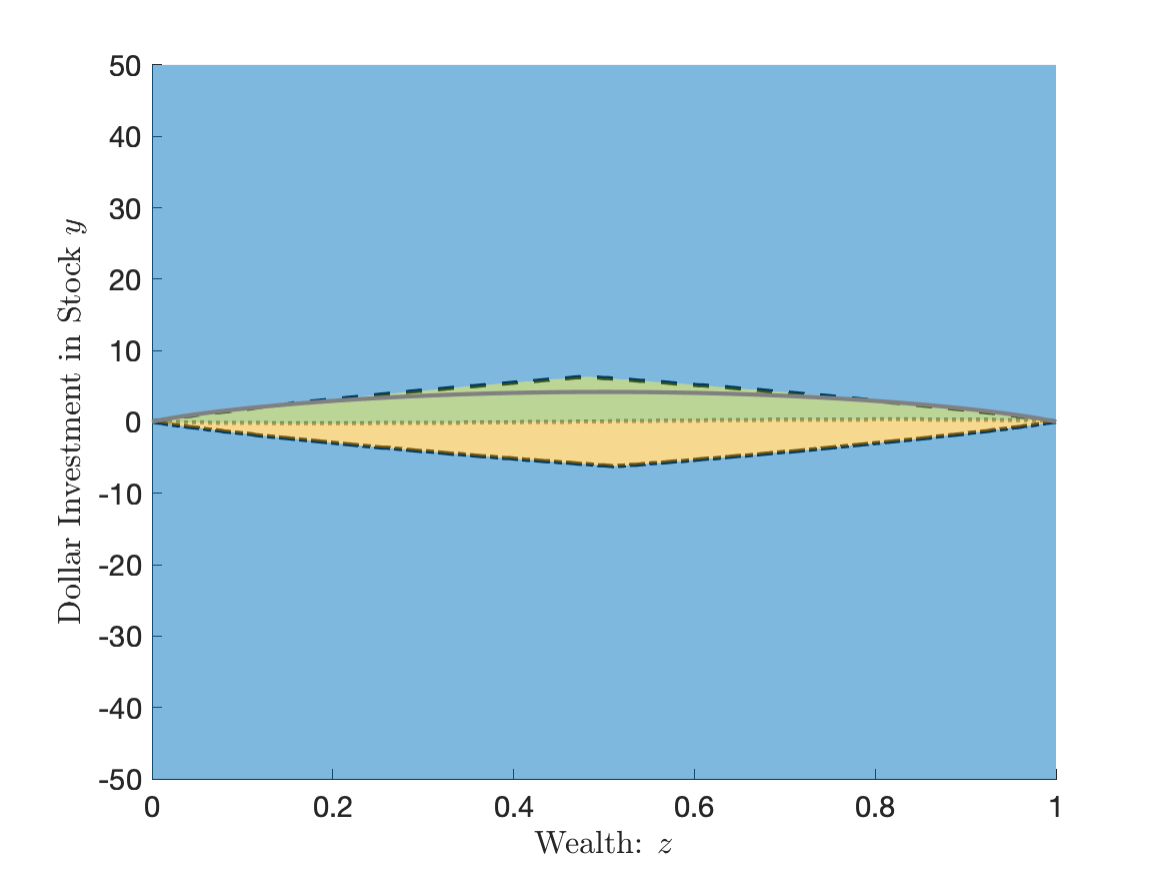}
\includegraphics[width=0.48\textwidth]{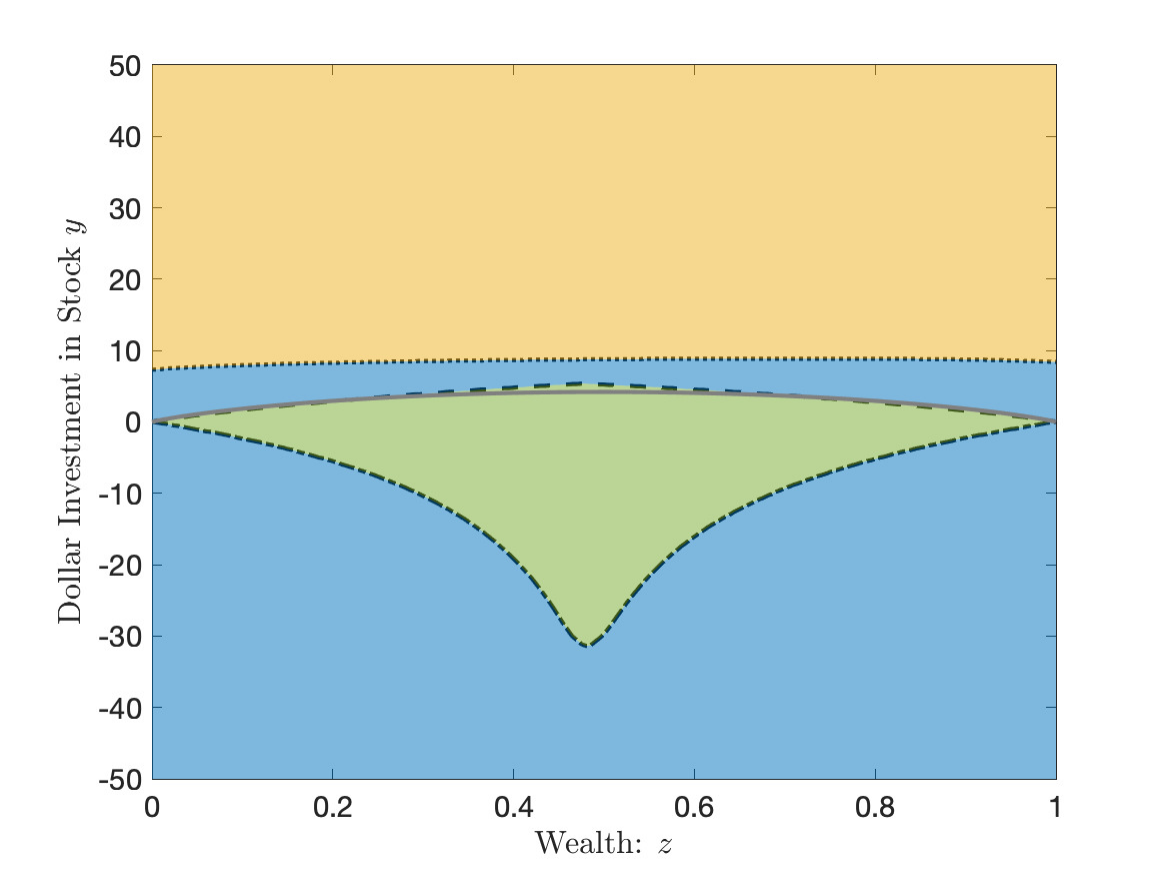}
\caption{Action regions for the goal-reaching problem with $\eta=0$ (left panels) and $\eta=0.04$ (right panels). The three rows from top to bottom correspond to $T-t=0.01,0.05,0.1$, respectively. Yellow: sell region; Blue: no-trading region; Green: buy region. Solid line: target position without transaction costs. Parameters: $\theta_1=\theta_2=10^{-3},\sigma=0.3.$}\label{fig goal reaching etas}
\end{figure}

\subsection{Goal-Reaching Problem with Positive Risk Premium}\label{subsect:GRp}
Next, we study the goal-reaching problem with a positive risk premium. To this end, we set $\eta=0.04$ and keep the other parameters the same as in Section \ref{subsect:GR0}. The action regions for $T-t=0.01,0.05,0.1$ are illustrated in the right panels in Figure \ref{fig goal reaching etas} from top to bottom, respectively. 

These figures show two effects of the positive risk premium. First, the investor will reduce the long position $y>0$ when it is sufficiently large, as indicated by the additional sell region.
This can be explained as follows. When $y>0$, the positive risk premium of the risky asset contributes towards achieving the goal.  Therefore, when $y>0$ and the leverage is too high, the investor has the incentive to reduce the position to stay in the market and take advantage of the risk premium. Such an incentive is stronger when $T-t$ is larger, as shown by the enlargement of this sell region in $y>0$. Nevertheless, this sell region has a strictly positive limit as $z$ increases to $1$, which is in stark contrast to \cite{browne1999reaching} and \cite{dai2022nonconcave}, where the limit is 0 (see the grey solid line). To understand this, in these two papers, the positive risk premium leads to the incentive to stay in the market for a longer time. Therefore, the leverage is significantly lowered as $z$ increases to reduce bankruptcy risk. In our case, the bankruptcy risk is negligible compared with the future purchase cost incurred.  Therefore, the investor would rather keep a strictly positive stock position and delay the liquidation to maturity. The intuition for the strictly positive sell boundary around $z = 0$ is similar.

Second, compared to the case with $\eta=0$ (left panels), when $|y|$ is sufficiently small, the investor is more inclined to establish a long position rather than a short position, when $T-t$ is larger. For example, in the top and middle panels, the buy region for $\eta=0.04$ is larger than those for $\eta=0$. Furthermore, in the bottom right panel with $T-t=0.1$, the buy region protrudes downwards when $z$ is away from 0 or 1, meaning that the investor is motivated to switch to a long position even from a large short position. This is because, given $y$, as $T-t$ becomes larger, the effect of the positive risk premium on the wealth becomes more significant compared to the portfolio variance and transaction costs. Nevertheless, when $z$ is close to 0 or 1, such motivation will be weaker, and the investor still defers trading and saves transaction costs so as to avoid getting close to liquidation or getting away from the goal. 
Furthermore, given $T-t$, the buy region in $y<0$ is bounded from below, suggesting that it is still optimal to hold onto a very large short position despite a positive risk premium. Intuitively, given $T-t$, with a sufficiently large short position (i.e., $-y>0$ is sufficiently large), the benefit from a large portfolio variance to reach the goal dominates the drawback from the negative drift of the portfolio value from the large short position. This can be seen from the following heuristic example.
Imagine that $Y_0<0$ is large, and the investor does not trade until reaching the goal or liquidation boundary. Then the expected stopping time $\mathbb{E}[\tau_0\wedge \tau_1] = O(1/|Y_0|^2)$, where $\tau_i$, $i = 0, 1$ is the first time that the wealth $z_t$ hits $i$. Since the time before reaching the goal or liquidation boundary is short, the negative drift impacts the wealth only at the level of $|Y_0|\mathbb{E}[\tau_0\wedge \tau_1] = O(1/|Y_0|) $, which is decreasing in $|Y_0|$. This implies that when the initial short position is large, the negative drift's relative influence decreases in $|Y_0|$ since the wealth level reaches either 0 or 1 faster. 

\begin{figure}[hptb!]
\includegraphics[width=0.48\textwidth]{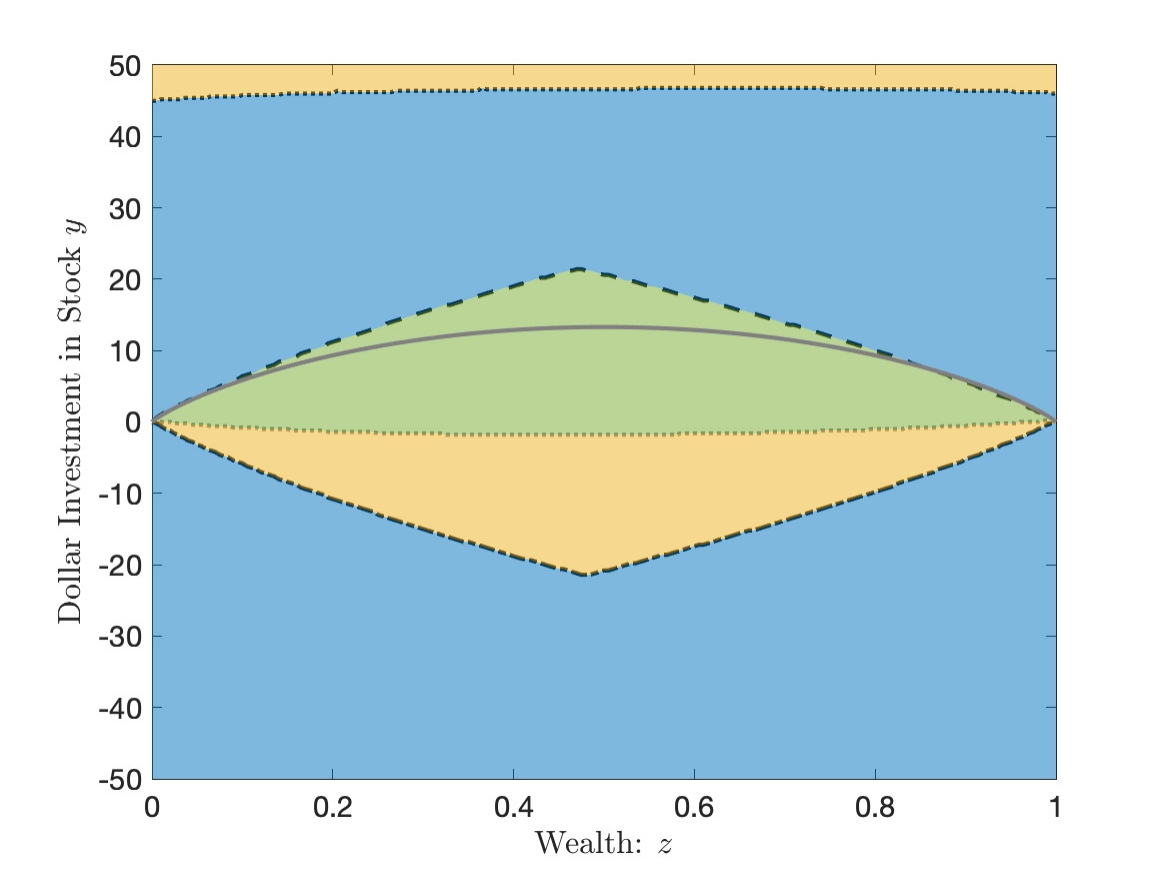}
\includegraphics[width=0.48\textwidth]{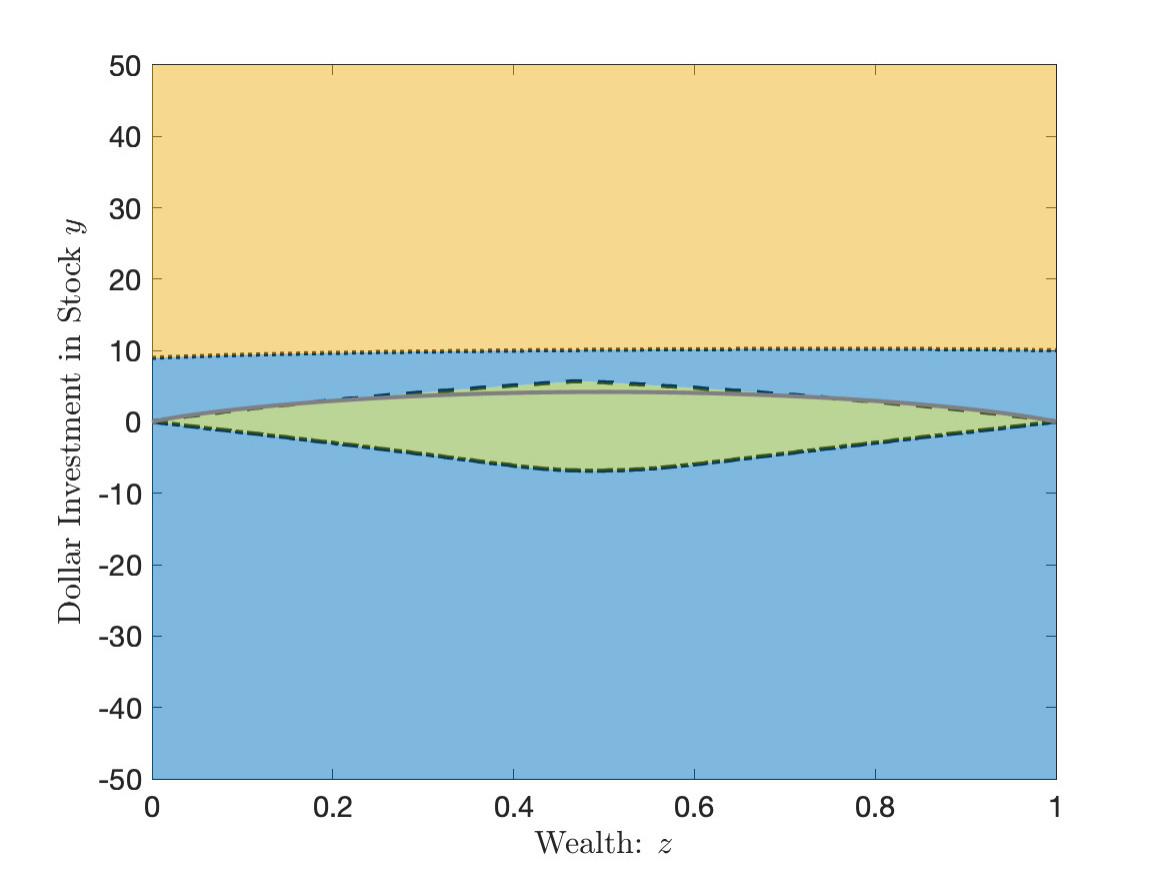}
\caption{The action regions for the goal-reaching problem with positive risk premium and larger transaction costs. Left panel: $T-t=0.01$; Right panel: $T-t=0.1$. Yellow: sell region; Blue: no-trading region; Green: buy region. Solid line: target position without transaction costs. Parameters: $\theta_1=\theta_2=2\times10^{-3},\sigma=0.3,\eta=0.04.$}\label{fig goal reaching theta}
\end{figure}

Next, we illustrate the effect of the transaction cost rates. By comparing the left and right panels of Figure \ref{fig goal reaching theta} to the top right and bottom right panels of Figure \ref{fig goal reaching etas}, respectively, we see that a larger transaction cost shrinks the buy region in $y<0$, i.e., it reduces the investor's incentive to switch from a short position to a long position to chase for the positive risk premium.

For the goal-reaching problem, \cite{dai2022nonconcave} also document investors' risk-seeking behavior. In their case, the trigger for such behavior is not the transaction costs, but rather the imposed two-sided portfolio bounds, which are the limit on the level of permitted leverage. In contrast, transaction costs do not put a direct limit on leverage; rather, the intrinsic limit that prevents taking infinite leverage is the large potential transaction cost that needs to be paid upon liquidation.  Consequently, Figure \ref{fig goal reaching etas} indicates that it is possible that the investor takes and holds on to much higher leverage compared to \cite{dai2022nonconcave}; for example, starting from the buy region, the investor will buy to the upper boundary of this region and keep the position if it subsequently moves into the no-trading region. Furthermore, unlike their model, our model does not produce a sudden switch between large long and short positions as the wealth changes, since this would trigger a very large amount of transaction costs.

\subsection{Aspiration Utility}\label{subsect:asp}
As a third example, we discuss the strategy under the aspiration utility \eqref{eqn:aspUtility} with $p=0.5$, $c_1=0$, $c_2=1.5$, and $ \bar{z}=1$. Note that due to the upward jump of the utility function at $z=\bar{z}$, the aspiration utility maximization problem resembles a goal-reaching problem locally for $z$ close to but below $\bar{z}$, so as to reach the ``goal'' $\bar{z}$ and get a boost in the utility. However, as the investor is still risk-averse and gains utility from the actual wealth level away from $\bar{z}$, the investor will act more risk-averse compared to the goal-reaching problem so as to preserve portfolio value. Therefore, we expect that the investor's optimal strategy should be a mixture of a goal-reaching problem and a concave utility maximization problem.

We first illustrate the numerical results for the case with $\eta=0$ in Figure \ref{fig aspi strategy eta0} for $T-t=0.01$ and 0.1. We see that, for $z$ below $\bar{z}=1$, the action regions locally resemble those of the goal-reaching problem (e.g., compare the regions I to IV to the four action regions from top to bottom in the top left panel of Figure \ref{fig goal reaching etas}). 
However, since the investor will also gain utility from the wealth level, the wealth-preserving motivation interacts with the goal-reaching motivation in two ways. First, when the wealth level $z$ is sufficiently close to 0 or above 1, the optimal strategy is to voluntarily liquidate the risky asset position to preserve value, rather than maintaining a long or short position as in the goal-reaching problem. 
This resembles the classic Merton strategy with proportional transaction costs under $\eta=0$ and a power utility with risk aversion $1-p$.
Second, when $z$ is below $\bar{z}$ but away from 0, the region I (resp. region IV) is upper bounded by a sell region (resp. lower bounded by a buy region), suggesting that the investor should not allow arbitrarily large long or short positions in the risky asset. Intuitively, the investor will still gain utility from the terminal wealth even if $\bar{z}=1$ cannot be eventually reached, and therefore, the investor should not take arbitrarily large leverage and risk.

\begin{figure}[hptb!]
\includegraphics[width=0.48\textwidth]{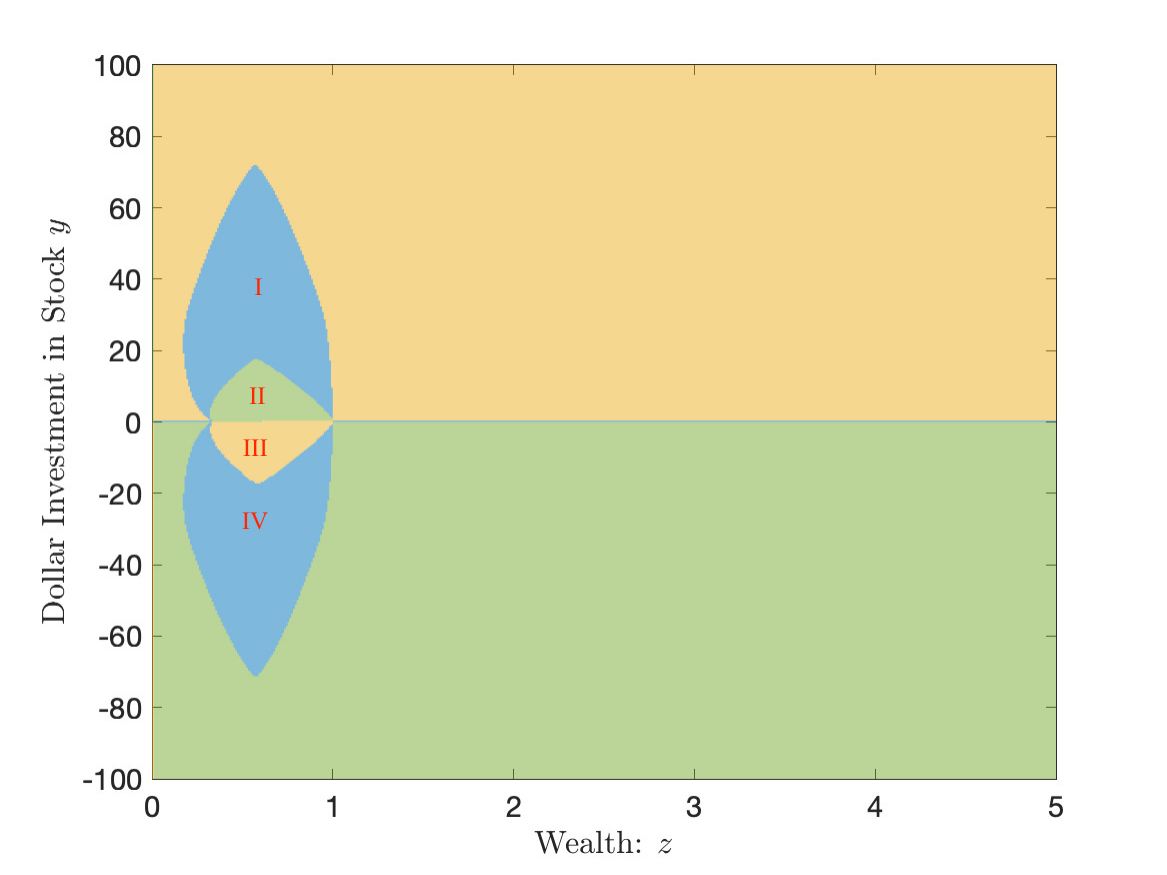}
\includegraphics[width=0.48\textwidth]{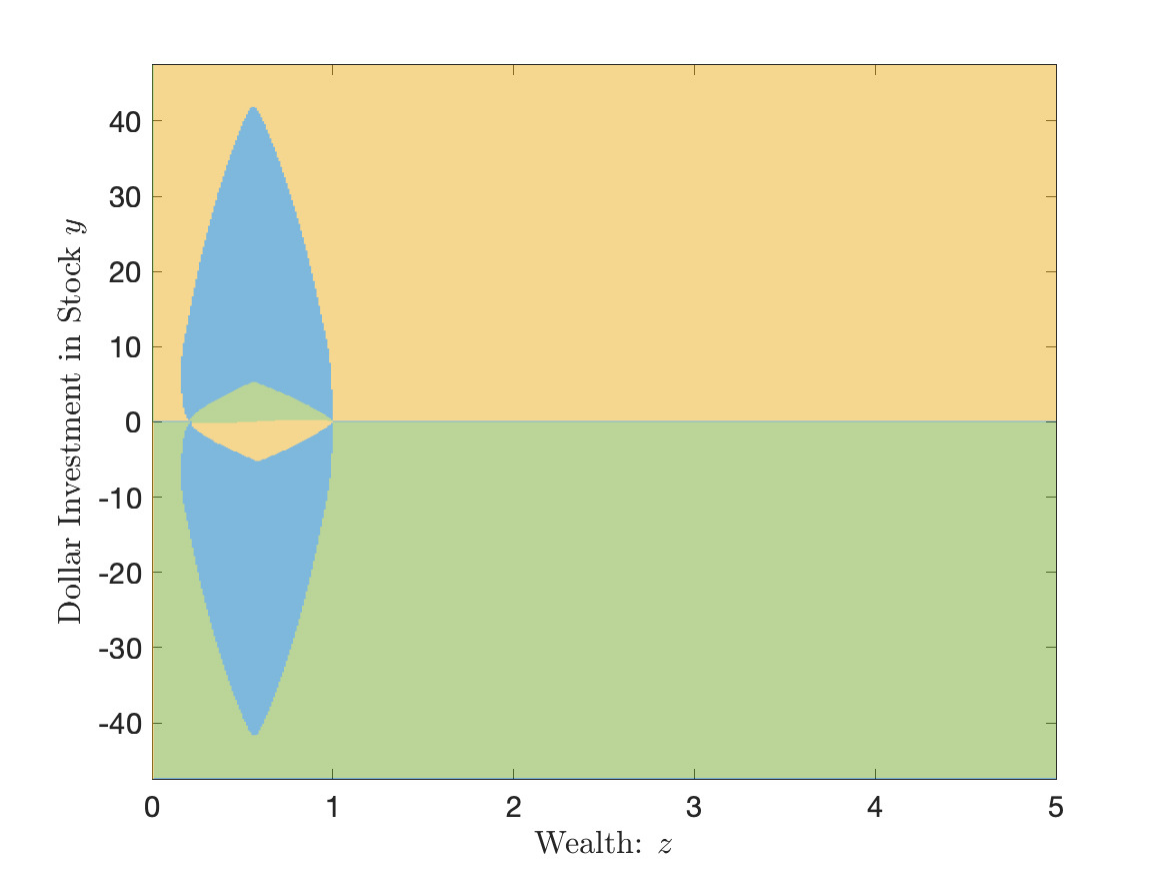}
\caption{The action regions for the aspiration utility problem with zero risk premium. Yellow: sell region; Green: buy region; Blue: no-trading region. 
Left panel: $T-t=0.01$; Right panel: $T-t=0.1$. Parameters: $\theta_1=\theta_2=10^{-3},\sigma=0.3,\eta=0$.
}\label{fig aspi strategy eta0}
\end{figure}

\begin{figure}[hptb!]
\includegraphics[width=0.48\textwidth]{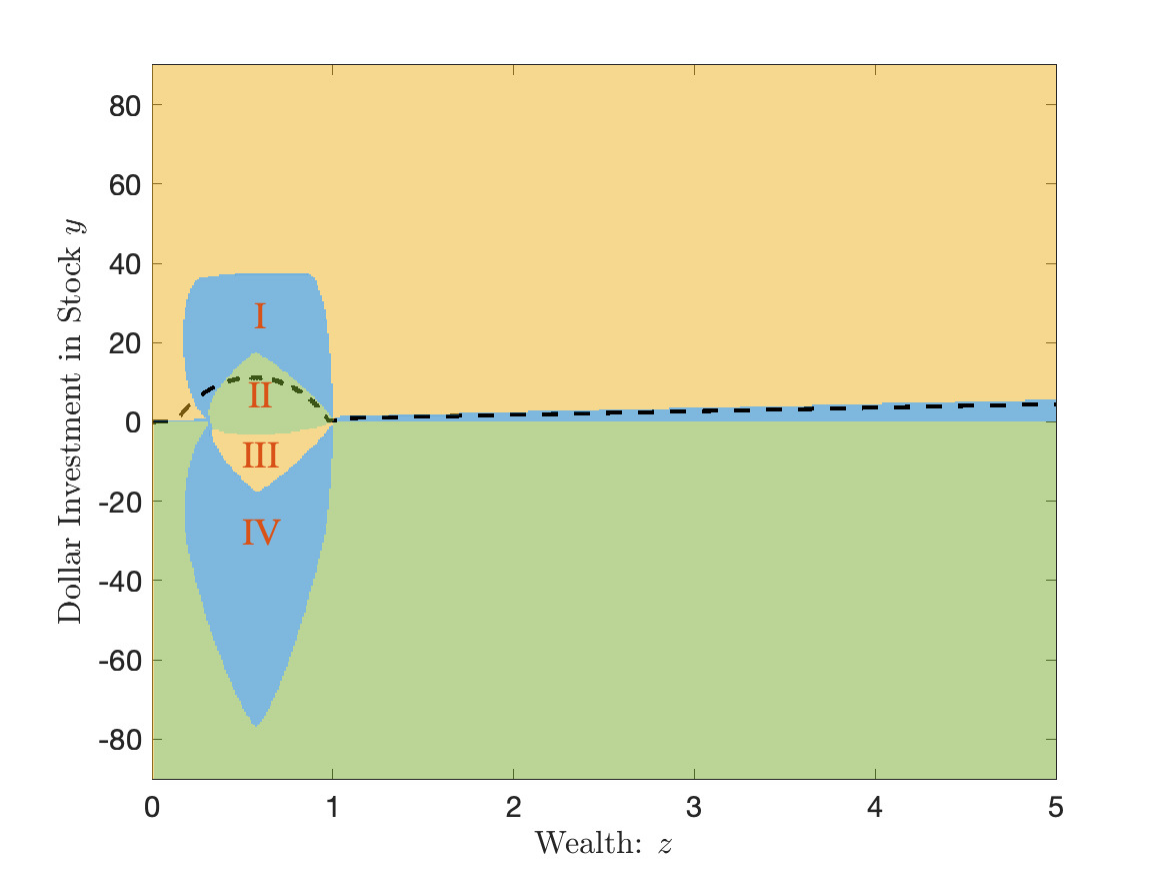}
\includegraphics[width=0.48\textwidth]{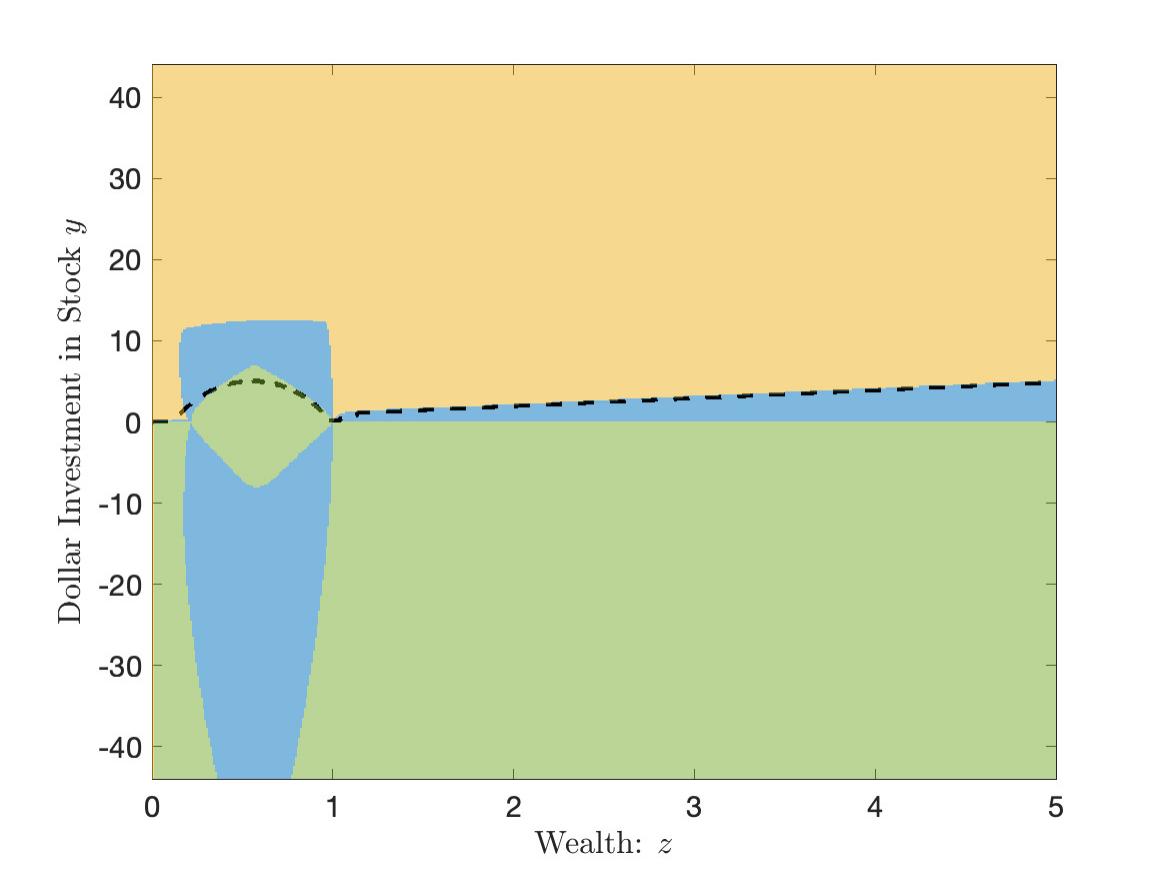}\\
\includegraphics[width=0.48\textwidth]{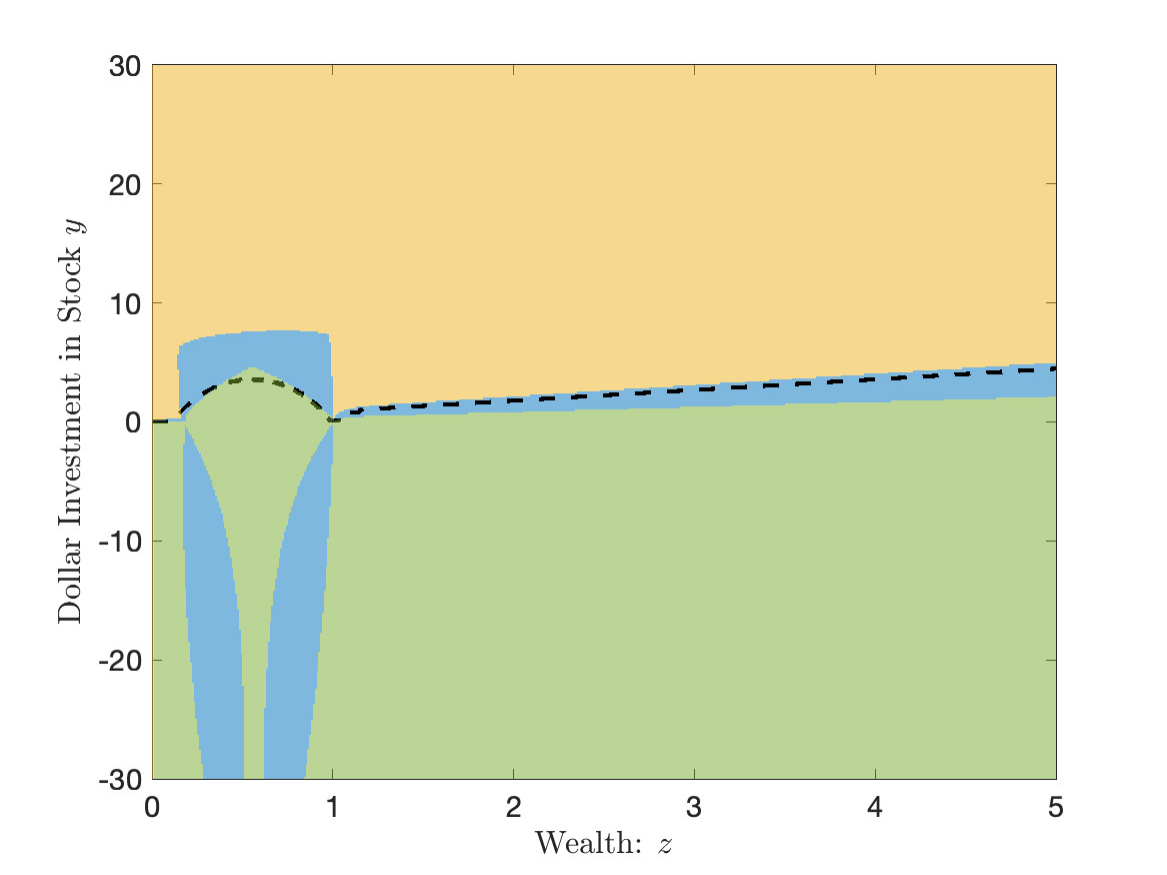}
\includegraphics[width=0.48\textwidth]{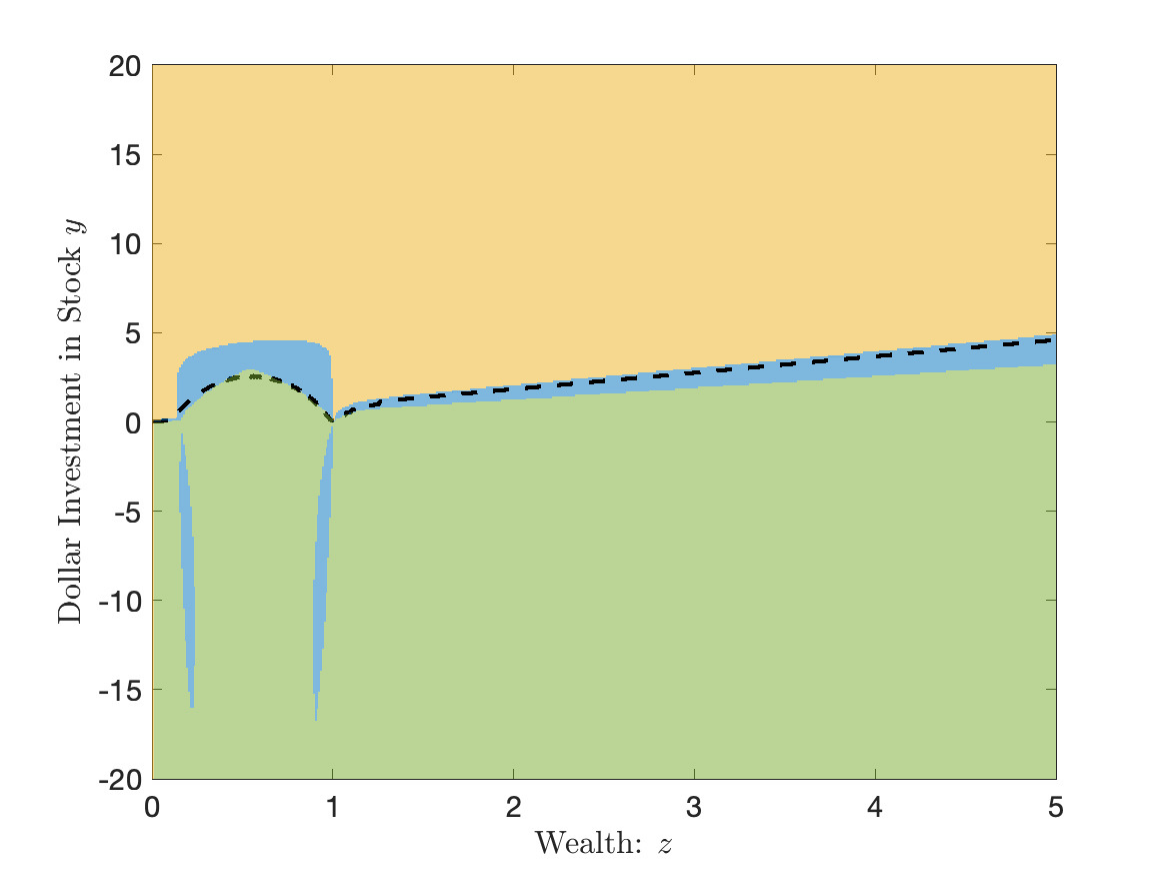}
\caption{The action regions for the aspiration utility problem with positive risk premium. Yellow: sell region; Green: buy region; Blue: no-trading region. Dashed line: target position without transaction cost. Upper left: $T-t=0.01$, upper right: $T-t=0.05$, lower left $T-t=0.1$, lower right $T-t=0.2$. Parameters: $\theta_1=\theta_2=10^{-3},\sigma=0.3,\eta=0.04.$
}\label{fig aspi strategy}
\end{figure}

Next, we switch to the case with a positive risk premium $\eta=0.04$. Figure \ref{fig aspi strategy} shows two main effects of the positive risk premium on the optimal strategy. First, when the wealth $z$ is sufficiently close to 0, the strategy is to keep a small fraction of wealth in the risky asset to avoid bankruptcy; when $z$ is sufficiently large, the optimal strategy again resembles the classic Merton strategy with transaction costs (e.g., \cite{shreve1994optimal} and \cite{dai2009finite}), that is, performing minimum trading to keep the position sufficiently close to the Merton line. 

Second, when $z$ is below 1 but away from 0, the effect of the risk premium is stronger when $T-t$ is larger, similar to the goal-reaching problem. Indeed, recall from the right panels of Figure \ref{fig goal reaching etas} that for the goal-reaching problem, as $T-t$ increases, the effects of the positive risk premium include (a) the downward enlargement of the sell region for $y>0$ sufficiently large, and (b) the downward enlargement of the buy region for $|y|$ small into the region $y<0$. By comparing Figure \ref{fig aspi strategy} with Figure \ref{fig aspi strategy eta0}, we see that similar effects of positive risk premium are also present in the aspiration utility case when $z$ is below 1 but away from 0: (a) the downward enlargement of the sell region above region I in Figure \ref{fig aspi strategy} when $y>0$ is sufficiently large; and (b) the downward enlargement of the buy region II, which first takes over the sell region III, then further protrudes downwards. The interpretation of these two effects is similar to the goal-reaching problem, namely, to reduce the stock position to avoid bankruptcy and to take advantage of the positive risk premium, respectively. Note that the lower bounded shape of the no-trading region IV and the gradual downward protrusion of the buy region II jointly create two disjoint no-trading regions in $y<0$ (c.f. the blue regions within $y<0$ in the bottom left and bottom right panels of Figure \ref{fig aspi strategy}). Summarizing the above discussions, in these two regions, the investor neither has sufficient motivation to reduce the short position to preserve wealth nor deems it cost-effective to switch to a long position for the positive risk premium.

\subsection{Goal-Reaching Problem under Gaussian Mean Return Model}\label{subset:GMR}
Finally, we plot the action regions for the goal-reaching problem with the Gaussian mean return model. To this end, we set $\bar{\nu} = 0.1333$ so that the risk premium corresponding to the long-term $\nu$ level is 0.04. Furthermore, we take $\kappa=0.27$, $\zeta=0.065$, and $\rho=-0.93$.

Overall, Figure \ref{fig nu strategy} shows that the action regions for the goal-reaching problem with the Gaussian mean return model are similar to those under the geometric Brownian motion model. Specifically, when the current risk premium is 0.04, the left panel of Figure \ref{fig nu strategy} resembles the bottom right panel of Figure \ref{fig goal reaching etas}, with a smaller buy region for $y<0$. Indeed, since the future risk premium $\mu(t)-r$ is uncertain, the investor will be more conservative in switching from a short position to a long position in the presence of transaction costs. When the current risk premium is $-0.04$, the action regions are flipped upside down as shown in the right panel of Figure \ref{fig nu strategy}. This can be interpreted similarly to the positive risk premium case, since the investor now benefits from the risk premium with a short position instead.

\begin{figure}[hptb!]
\includegraphics[width=0.48\textwidth]{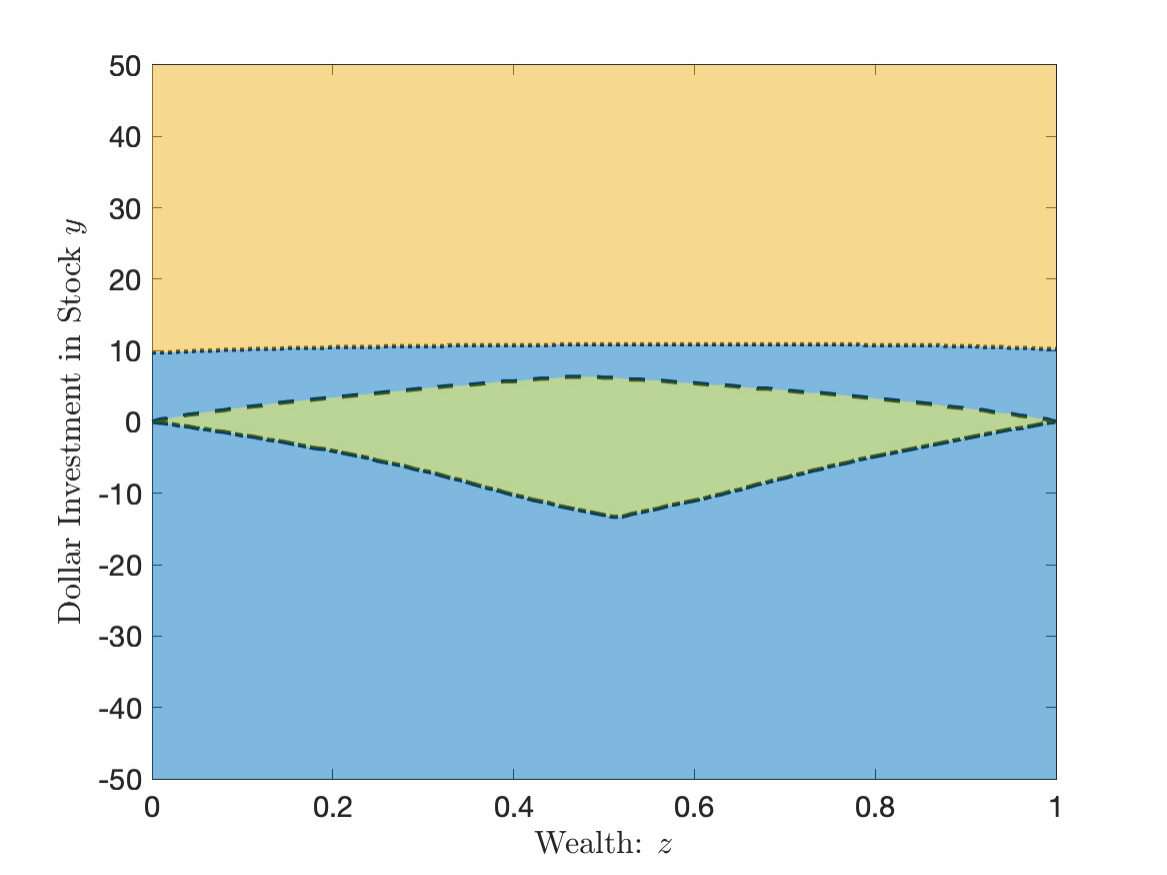}
\includegraphics[width=0.48\textwidth]{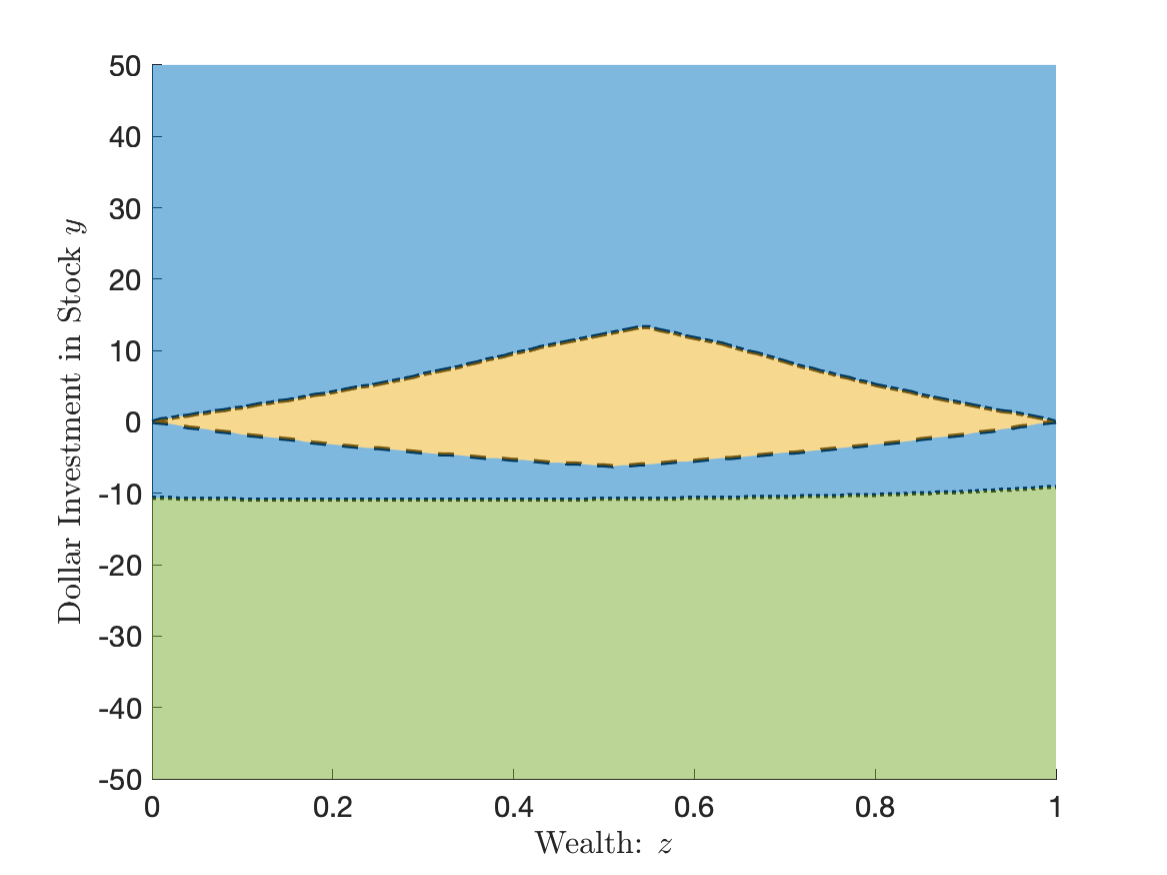}
\caption{The action regions of the goal-reaching problem under the Gaussian mean return model ($T-t=0.1$). Left panel: ${\nu}=0.1333$ (current $\mu-r=0.04$); Right panel ${\nu}=-0.1333$ (current $\mu-r=-0.04$). Yellow: sell region; Green: buy region; Blue: no-trading region. Dashed line: target position without transaction cost. Parameters: $\kappa=0.27$, $\zeta=0.065$, and $\rho=-0.93$ are from   \cite{dai2022nonconcave}. We also set $\sigma=0.3$, $\bar{\nu}=0.1333$ and $\theta_1=\theta_2=10^{-3}$. Under these parameters, $t_{\hat{p}} = -\infty$ for any $\hat{p}\in (0,1)$ and sufficiently close to $0$. 
}\label{fig nu strategy}
\end{figure}

\section{Conclusion}\label{sect:conclusion}
In this paper, we study the non-concave utility maximization problem under proportional transaction costs. Since the concavification principle is no longer applicable, we derive a rigorous theoretical characterization of the value function in terms of the discontinuous viscosity solution. Most importantly, we establish the asymptotic behavior of the value function as time approaches maturity. 

As numerical illustrations, we study the optimal strategies for the goal-reaching problem, as well as the aspiration utility maximization problem. We found that when the remaining time is short, the investor may hold on to a larger position in the risky asset compared to the frictionless case, and the investor can also hold on to a large short position of the risky asset despite its positive risk premium. 

From a theoretical perspective, this paper is among the strands of work on the discontinuous viscosity solutions arising from some mathematical finance problems. It will be of interest to further build a unified theoretical framework incorporating general frictions, such as capital gains tax and fixed costs. 
The joint impact of frictions and risk-seeking incentives predicted by our numerical results can also inspire future empirical work for real-world analyses.

\bibliographystyle{apalike}
\bibliography{reference.bib}

\section*{Appendices}
\addcontentsline{toc}{section}{Appendices}
\renewcommand{\thesubsection}{Appendix \Alph{subsection}}

\subsection{Proof of Theorem \ref{prop:term}}\label{sect:propterm}
We decompose the proof into three steps. We first show that the theorem holds in the special case of the goal-reaching problem in \ref{App A GBM 1}. Second, we prove a result bridging the goal-reaching problem to the general case in \ref{App A GBM 2}. Third, we finally prove for general utilities in \ref{App A GBM 3}.    

\subsubsection{The Special Case of Goal-Reaching Problem}\label{App A GBM 1}
\begin{proposition}\label{prop:goalreaching}
For the goal-reaching problem with $U(z) = \mathbf{1}_{z\geq w}$, $z\geq K$,  we have
\begin{align}\label{equ: termi}
\lim \limits_{(t, x, y, \nu) \to (T^-, \hat{x}, \hat{y}, \hat{\nu})}  V(t, x, y, \nu) - 2\Phi \bigg( \frac{\min\{z-w, 0\}}{ |w-x| \sigma(\hat{\nu})\sqrt{T-t} } \bigg) = 0,
\end{align}
where $z := x+(1-\theta_1)y^+- (1+\theta_2) y^-\geq K$ and $\hat{z} := \hat{x}+(1-\theta_1)\hat{y}^+- (1+\theta_2) \hat{y}^-\geq K$.  This equality can be achieved when choosing the strategy $\pi^*$ such that no transaction is made before reaching the liquidation boundary $K$ or the target wealth level $w$, and as soon as reaching either of these two boundaries, the investor liquidates all the stocks and waits until terminal time $T$. 
\end{proposition}

\begin{proof}[Proof of Proposition \ref{prop:goalreaching}]

The result is straightforward when $z \geq w$. Therefore, we focus on $z<w$ in the following. 
 
We only consider the case $y>0$; the case $y \leq 0$ can be proved similarly. 

1. We first show that for this terminal condition, we only need to consider strategies without buying or shorting stock in $[t, T]$.

Assume $\Delta>0$ dollars of stock is bought at time $t$, i.e., $d L_t = \Delta>0$, and the investor sells those stock at time $t'\leq T$, then for any $s \in [t, t']$, we have the value of these $\Delta$ amount of stocks is $(1-\theta_1)\Delta \frac{S_s}{S_t}$, while the total purchase cost paid at time $t$ is $(1+ \theta_2) \Delta$. Thus, since the interest rate is $0$, we have the excess return of this transaction is
\begin{align} 
\Delta \bigg((1-\theta_1) \frac{S_s}{S_t} - (1+ \theta_2) \bigg), \ \forall s\in [t, t'], 
\end{align} 
which will be always non-positive if  $\max \limits_{t\leq s_1, s_2 \leq T}  \frac{S_{s_2}}{S_{s_1}} \leq \frac{ 1+ \theta_2}{1-\theta_1}$. 
Similarly, short-selling will also  generate negative excess return if $\max \limits_{t\leq s_1, s_2 \leq T}  \frac{S_{s_2}}{S_{s_1}} \leq \frac{ 1+ \theta_2}{1-\theta_1}$. 
We have the following limit 
\begin{align}\label{equ lim S ratio}
\lim \limits_{(t, x, y, \nu) \to (T^-, \hat{x}, \hat{y}, \hat{\nu})} \mathbb{P} \left(\max \limits_{t\leq s_1, s_2 \leq T} \frac{S_{s_2}}{S_{s_1}} \geq \frac{1+\theta_2}{1-\theta_1}  \right) = 0, 
\end{align}
because to ensure $S_{s_2}$ fluctuate enough amount, either $|\nu|$ is large enough, or the Brownian motion is sufficiently fluctuating, while the first event's probability is bounded by Assumption \ref{ass:2}, which converges to $0$, and the second event's probability also converges to 0 when $T-t\to 0$.  

2. In what follows, we only consider the strategies without purchase and short-sale in $[t, T]$. In this case, since the entire stock position should be liquidated no later than the terminal time, we must have   
\begin{align}\notag
\mathbb{P}(Z_T \geq w) &\leq \mathbb{P}\left((1-\theta_1)y \left(\max\limits_{t\leq s \leq T}\frac{S_{s}}{S_{t}} -1\right) \geq w-z \right)\\ \label{equ:asym}
 &= \mathbb{P} \left(\max\limits_{t\leq s \leq T}\frac{S_{s}}{S_{t}}  \geq \frac{w-x}{(1-\theta_1)y} \right). 
\end{align}
For any $\delta \leq 1$, consider the subset $\Omega_\delta: = [\hat{\nu}-\delta, \hat{\nu}+\delta ]$, and denote $a = \max \limits_{\nu \in \Omega_1}|\eta( \nu) -\frac{1}{2}\sigma^2(\nu)| $,  $\sigma_\delta = \max \limits_{\nu \in \Omega_\delta}\sigma(\nu)$,  we have that  for any constant $C>1$,
\begin{align}
	& \mathbb{P} \left(\max\limits_{t\leq s \leq T}\frac{S_{s}}{S_{t}}  \geq C\right)\\
	 = & \mathbb{P} \left(\max\limits_{t\leq s \leq T}\ln\frac{S_{s}}{S_{t}}  \geq \ln C\right) \notag \\
	\leq  & \mathbb{P} \left(a(T-t) +  \max\limits_{t\leq s \leq T}  \int_{t}^s  \sigma(\nu_v) d \mathcal{B}_v \geq \ln C, \max \limits_{t\leq s \leq T} |\nu_s -\hat{\nu}| \leq \delta  \right) + \mathbb{P} \left( \max \limits_{t\leq s \leq T} |\nu_s -\hat{\nu}| \geq \delta \right) \\
	\leq & \mathbb{P} \left( \max\limits_{t\leq s \leq T}  \int_{t}^s  \sigma(\nu_v) d \mathcal{B}_v \geq \ln C - a(T-t) , \max \limits_{t\leq s \leq T} \sigma(\nu_s) \leq \sigma_\delta  \right) + \mathbb{P} \left( \max \limits_{t\leq s \leq T} |\nu_s -\hat{\nu}| \geq \delta \right). \label{equ max S esti 1}
\end{align}
We introduce a new time scale $\xi = t+ \int_t^s \frac{\sigma^2(\nu_v)}{\sigma_\delta^2} dv$, and have 
\begin{align}
\int_t^s \sigma(\nu_v) d \mathcal{B}_v \text{ is identical in distribution with } \int_t^\xi \sigma_\delta d \mathcal{B}_v,\ \forall\ t\leq s\leq T .
\end{align}
Noticing that when $\max \limits_{t\leq s \leq T} \sigma(\nu_s) \leq \sigma_\delta$, we always have $\xi \leq s \leq T$, then 
\begin{align}
 	& \mathbb{P} \left( \max\limits_{t\leq s \leq T}  \int_{t}^s  \sigma(\nu_v) d \mathcal{B}_v \geq \ln C - a(T-t) , \max \limits_{t\leq s \leq T} \sigma(\nu_s) \leq \sigma_\delta  \right)\\
 	\leq   &  \mathbb{P} \left( \max\limits_{t\leq \xi \leq T}  \int_{t}^\xi  \sigma_\delta d \mathcal{B}_s \geq \ln C - a(T-t) \right)\\
 	=  & \mathbb{P} \left(\max\limits_{t\leq \xi \leq T} (\mathcal{B}_\xi - \mathcal{B}_t) \geq \frac{\ln C-a(T-t)}{\sigma_\delta} \right)\\
 	= &2 \Phi \left(\frac{\min\{-\ln C+a(T-t), 0\}}{\sigma_\delta\sqrt{T-t}}\right). \label{equ max S esti 2}
\end{align}
Therefore, from \eqref{equ max S esti 1}, \eqref{equ max S esti 2}, according to the  inequality $  \ln w\geq \frac{w-1}{w}$, $\forall w > 0$, 
\begin{align}
	\eqref{equ:asym} &\leq 2 \Phi \left(\frac{\min \{-\ln (\frac{w-x}{(1-\theta_1)y}) +a(T-t), 0\}} {\sigma_\delta\sqrt{T-t}}\right) + \mathbb{P} \left( \max \limits_{t\leq s \leq T} |\nu_s -\hat{\nu}| \geq \delta \right)\notag\\
	&\leq  2 \Phi  \left(\frac{\min \{-\frac{w-z}{w-x} +a(T-t), 0\}} {\sigma_\delta\sqrt{T-t}}\right)+ \mathbb{P} \left( \max \limits_{t\leq s \leq T} |\nu_s -\hat{\nu}| \geq \delta \right) \notag \\
    & = 2 \Phi  \left(\frac{\min \{z-w+ a(w-x)(T-t), 0\}} {\sigma_\delta(w-x)\sqrt{T-t}}\right)+ \mathbb{P} \left( \max \limits_{t\leq s \leq T} |\nu_s -\hat{\nu}| \geq \delta \right).\label{equ max S}
\end{align}
As a result, on the one hand, from \eqref{equ lim S ratio}, \eqref{equ max S} and Assumption \ref{ass:2}, we have 
\begin{align}
	& \limsup  \limits_{(t, x, y, \nu) \to ( T^-, \hat{x}, \hat{y}, \hat{\nu})} \bigg[V(t, x, y, \nu) -2 \Phi \left(\frac{z-w}{ (w-x) \sigma(\hat{\nu})\sqrt{T-t} } \right) \bigg]\\
	\leq 
& \limsup \limits_{(t, x, y, \nu) \to ( T^-, \hat{x}, \hat{y}, \hat{\nu})} \bigg[ \mathbb{P}\left(\max \limits_{t\leq s_1, s_2 \leq T} \frac{S_{s_2}}{S_{s_1}} \geq \frac{1+\theta_2}{1-\theta_1}  \right)  \\
	& \qquad  +2 \Phi  \left(\frac{\min\{ z-w+ a(w-x)(T-t), 0\}} {\sigma_\delta(w-x)\sqrt{T-t}}\right) -2 \Phi \left(\frac{z-w}{ (w-x) \sigma(\hat{\nu})\sqrt{T-t} } \right) \notag \\
	& \qquad + \mathbb{P} \left( \max \limits_{t\leq s \leq T} |\nu_s -\hat{\nu}| \geq \delta \right)\bigg]\\
	= &\limsup \limits_{(t, x, y, \nu) \to ( T^-, \hat{x}, \hat{y}, \hat{\nu})} \bigg[ 2 \Phi  \left(\frac{\min\{ z-w+ a(w-x)(T-t), 0\}} {\sigma_\delta(w-x)\sqrt{T-t}}\right) -2 \Phi \left(\frac{z-w}{ (w-x) \sigma(\hat{\nu})\sqrt{T-t} } \right) \notag \\
	& \qquad + \mathbb{P} \left( \max \limits_{t\leq s \leq T} |\nu_s -\hat{\nu}| \geq \delta \right)\bigg]\\
	= & \limsup \limits_{(t, x, y, \nu) \to ( T^-, \hat{x}, \hat{y}, \hat{\nu})} \bigg[  2 \Phi  \left(\frac{\min\{ z-w+ a(w-x)(T-t), 0\}} {\sigma_\delta(w-x)\sqrt{T-t}}\right) -2 \Phi \left(\frac{z-w}{ (w-x) \sigma(\hat{\nu})\sqrt{T-t} } \right)\bigg] \notag \\
	= & \limsup \limits_{(t, x, y, \nu) \to ( T^-, \hat{x}, \hat{y}, \hat{\nu})} \bigg[ 2 \Phi  \left(\frac{z-w} {\sigma_\delta(w-x)\sqrt{T-t}}\right) -2 \Phi \left(\frac{z-w}{ (w-x) \sigma(\hat{\nu})\sqrt{T-t} } \right)\bigg] \label{equ uni Phi} \\
 	\leq  & 2 \max \limits_{v\leq 0} \bigg[\Phi\bigg(\frac{v}{\sigma_\delta}\bigg) - \Phi\bigg(\frac{v}{\sigma(\hat{\nu})}\bigg)\bigg], 
\end{align}
where \eqref{equ uni Phi} is from the uniformly continuity of  function $\Phi$ on $\mathbb{R}$.
Since $\delta$ can be arbitrarily small and $\lim\limits_{\epsilon\to 0} \sup\limits_{w\in \mathbb{R}}|\Phi((1+\epsilon) w) -\Phi(w)| = 0$,  
we have actually 
\begin{align}
	 \limsup  \limits_{(t, x, y, \nu) \to ( T^-, \hat{x}, \hat{y}, \hat{\nu})} \bigg[V(t, x, y, \nu) -2 \Phi \left(\frac{z-w}{ (w-x) \sigma(\hat{\nu})\sqrt{T-t} } \right) \bigg] \leq 0. 
\end{align}
On the other hand, consider the strategy $\pi^*$ such that no transaction will be made before reaching the target or the liquidation boundary, then  
\begin{align}\label{equ V lowerbd}
V(t, x, y, \nu)\geq & \mathbb{P}(Z^{\pi^*}_T = w) 
 	=  \mathbb{P} \left(\max\limits_{t\leq s \leq T}\frac{S_{s}}{S_{t}}  \geq \frac{w-x}{(1-\theta_1)y} \right) - \mathbb{P} \left(\min\limits_{t\leq s \leq T}\frac{S_{s}}{S_{t}}  \leq \frac{K-x}{(1-\theta_1)y} \right). 
\end{align}
Analogous to \eqref{equ max S}, we have from $\ln w \leq w-1$, $\forall w >0$ that 
\begin{align}
	& \mathbb{P} \left(\max\limits_{t\leq s \leq T}\frac{S_{s}}{S_{t}}  \geq \frac{w-x}{(1-\theta_1)y} \right) \geq 2 \Phi  \left(\frac{\min \{z-w- a(1-\theta_1)y (T-t), 0\}} {(1-\theta_1)y\sqrt{T-t}\min \limits_{\nu \in \Omega_\delta}\sigma(\nu)}\right)- \mathbb{P} \left( \max \limits_{t\leq s \leq T} |\nu_s -\hat{\nu}| \geq \delta \right),
\end{align}
and 
\begin{align}
	& \mathbb{P} \left(\min\limits_{t\leq s \leq T}\frac{S_{s}}{S_{t}}  \leq \frac{K-x}{(1-\theta_1)y} \right) \leq     2 \Phi  \left(\frac{\min \{K-z+ a(1-\theta_1)y (T-t), 0\}} {(1-\theta_1)y\sqrt{T-t}\min \limits_{\nu \in \Omega_\delta}\sigma(\nu)}\right)+ \mathbb{P} \left( \max \limits_{t\leq s \leq T} |\nu_s -\hat{\nu}| \geq \delta \right), \label{equ prob low}
\end{align}
which converges to $0$ when $(t, x, y, \nu) \to ( T^-, \hat{x}, \hat{y}, \hat{\nu})$ for any $\hat{z} >K$.

{Therefore,  from Assumption \ref{ass:2}} and \eqref{equ V lowerbd},  \eqref{equ prob low}, when $\hat{z} = w$, 
\begin{align}
	& \liminf  \limits_{(t, x, y, \nu) \to ( T^-, \hat{x}, \hat{y}, \hat{\nu})} \bigg[V(t, x, y, \nu) - 2 \Phi \bigg(\frac{z-w}{ (w-x) \sigma(\hat{\nu})\sqrt{T-t} } \bigg)\bigg]\\
	\geq &  \liminf  \limits_{(t, x, y, \nu) \to ( T^-, \hat{x}, \hat{y}, \hat{\nu})} \bigg[2 \Phi  \left(\frac{\min \{z-w- a(1-\theta_1)y (T-t), 0\}} {(1-\theta_1)y\sqrt{T-t}\min \limits_{\nu \in \Omega_\delta}\sigma(\nu)}\right)- 2\mathbb{P} \left( \max \limits_{t\leq s \leq T} |\nu_s -\hat{\nu}| \geq \delta \right) \\
	&-  2 \Phi  \left(\frac{\min \{K-z+ a(1-\theta_1)y (T-t), 0\}} {(1-\theta_1)y\sqrt{T-t}\min \limits_{\nu \in \Omega_\delta}\sigma(\nu)}\right) - 2 \Phi \bigg(\frac{z-w}{ (w-x) \sigma(\hat{\nu})\sqrt{T-t} } \bigg) \bigg]\\
  = & \liminf  \limits_{(t, x, y, \nu) \to ( T^-, \hat{x}, \hat{y}, \hat{\nu})} \bigg[2 \Phi  \left(\frac{\min \{z-w- a(1-\theta_1)y (T-t), 0\}} {(1-\theta_1)y\sqrt{T-t}\min \limits_{\nu \in \Omega_\delta}\sigma(\nu)}\right) - 2 \Phi \big(\frac{z-w}{ (w-x) \sigma(\hat{\nu})\sqrt{T-t} } \big)\bigg]\\
  = &  \liminf  \limits_{(t, x, y, \nu) \to ( T^-, \hat{x}, \hat{y}, \hat{\nu})} \bigg[2 \Phi  \left(\frac{z-w} {(1-\theta_1)y\sqrt{T-t}\min \limits_{\nu \in \Omega_\delta}\sigma(\nu)}\right) - 2 \Phi \big(\frac{z-w}{ (w-x) \sigma(\hat{\nu})\sqrt{T-t} } \big)\bigg]\label{equ uni Phi2}\\
  \geq   & 2 \max \limits_{v\leq 0} \bigg[\Phi\bigg(\frac{v}{\min \limits_{\nu \in \Omega_\delta}\sigma(\nu)}\bigg) - \Phi\bigg(\frac{v}{\sigma(\hat{\nu})}\bigg)\bigg],
\end{align}
where \eqref{equ uni Phi2} is from that the function $\Phi$ is uniformly continuous on $\mathbb{R}$.
Since $\delta$ can be arbitrarily small and $\lim\limits_{\epsilon\to 0} \sup\limits_{w\in \mathbb{R}}|\Phi((1+\epsilon) w) -\Phi(w)| = 0$, 
we have actually 
\begin{align}
\liminf  \limits_{(t, x, y, \nu) \to ( T^-, \hat{x}, \hat{y}, \hat{\nu})} V(t, x, y, \nu) - 2 \Phi (\frac{z-w}{ (w-x) \sigma(\hat{\nu})\sqrt{T-t} } ) \geq 0, \text{ when } \hat{z}= w. 
\end{align}
When $\hat{z} < w$, $\limsup  \limits_{(t, x, y, \nu) \to ( T^-, \hat{x}, \hat{y}, \hat{\nu})} \Phi (\frac{z-w}{ (w-x) \sigma(\hat{\nu})\sqrt{T-t} } ) = 0$, and since $V(t, x, y, \nu)\geq 0$, we can easily derive the inequality
\begin{align}
\liminf  \limits_{(t, x, y, \nu) \to ( T^-, \hat{x}, \hat{y}, \hat{\nu})} V(t, x, y, \nu) - 2 \Phi (\frac{z-w}{ (w-x) \sigma(\hat{\nu})\sqrt{T-t} } ) \geq 0, \text{ when } \hat{z}< w. 
\end{align} 
Consequently, we have proved the proposition. 
\end{proof}

\subsubsection{Bridging the Goal-Reaching Problem with the General Case}\label{App A GBM 2}

Before we prove the terminal condition \eqref{equ: tercon} for general utility functions, we need the following proposition. 
\begin{proposition}\label{pro:estipro}
For any constants $0<q<1$, $\alpha>0$, $w>z$, $n \in \mathbb{N}^+$, and  $(x, y, \nu)\in \mathscr{S}$ , there exists constant {$C_{\nu, n}, d_{\nu} >0$}, which are locally uniformly bounded with respect to $\nu$ such that when
\begin{align}
T-t \leq &\min \bigg\{1, T-{t_{\nu, 2|\nu|}}, \bigg(\frac{w-z}{16 C_{\nu} (1-\theta_1)w^\alpha } \bigg)^4, \bigg(\frac{\ln 2}{4C_{\nu}}\bigg)^2,   \frac{z^{4\alpha}}{|{y}|^4}, (\frac{d_{\nu}}{2n})^2 \bigg\},  
\end{align} 
where $t_{\nu, 2|\nu|}$ is from Assumption \ref{ass:2} by setting $\bar{N} = 2|\nu|$ and  $C_{\nu}: = 2L |\nu| + |\eta(0)|+ \frac{1}{2} \bigg(\sigma(0) +2L|\nu|\bigg)^2$,
we have for any admissible strategy $(L_s, M_s)_{t\leq s \leq T} \in \mathcal{A}_t(x, y, \nu)$, the following inequality always holds
\begin{align}
\mathbb{P}\bigg(Z_T \geq w\big| (X_t, Y_t, \nu_t) = (x, y, \nu)\bigg) \leq &C_{\nu, n} \bigg(z^q (T-t)^{q/4} w^{-\alpha q} + \bigg(\frac{w^\alpha} {w-z}\bigg)^n  (T-t)^{n/4}\bigg).
\end{align}
\end{proposition}

\begin{proof}[Proof of Proposition \ref{pro:estipro}]

For any strategy $\Pi = (L_s, M_s)$, $t\leq s\leq T$, consider the process with no transaction cost under strategy $\Pi$. Denote the corresponding cash, stock, and wealth processes by $X^{(0)}_s$,  $Y^{(0)}_s$,  $Z^{(0)}_s$. Then 
\begin{align}
L_T- L_t \leq \frac{1}{\theta_1}\max\limits_{t\leq s\leq T} Z^{(0)}_s,\notag
\end{align}
 {since the stock position is subject to transaction cost $\theta_1$ upon liquidation}, and cumulative transaction cost cannot exceed $\max\limits_{t\leq s\leq T} Z^{(0)}_s$.  Similarly, we have $$M_T- M_t \leq \frac{1}{\theta_2}\max\limits_{t\leq s\leq T} Z^{(0)}_s.$$
Given  $(x_t, y_t, z_t, \nu_t)$, to make $Z_T >w$, {we need either that the long or short leverage is sufficiently high, or the stock price is sufficiently fluctuant in the remaining time.} Therefore,   
for any constant $B>0$, we have 
\begin{align}
&\mathbb{P}\bigg(Z_T \geq w\big| (X_t, Y_t, \nu_t) = (x, y, \nu)\bigg)\notag \\
 \leq & \mathbb{P}(L_T- L_t\geq B) + \mathbb{P}\left((1-\theta_1)(B+|y|)\left(\max \limits_{t\leq s_1\leq s_2 \leq T}\frac{S_{s_2}}{S_{s_1}}-1\right)  \geq w-z\right)  \notag \\
& \quad +  \mathbb{P}(M_T- M_t \geq B) +  \mathbb{P}\left((1+\theta_2)(B+|y|)\left(1- \min\limits_{t\leq s_1\leq s_2 \leq T}\frac{S_{s_2}}{S_{s_1}}\right)  \geq w-z\right) \notag \\
\leq & \mathbb{P}\left(\max\limits_{t\leq s\leq T} Z^{(0)}_s \geq \theta_1 B\right)  + \mathbb{P}\left((1-\theta_1)(B+|y|)\left(\max \limits_{t\leq s_1\leq s_2 \leq T}\frac{S_{s_2}}{S_{s_1}}-1\right)  \geq w-z\right)  \label{equ distriesti1} \\
& \quad +  \mathbb{P}\left(\max\limits_{t\leq s\leq T} Z^{(0)}_s \geq \theta_2 B\right) +  \mathbb{P}\left((1+\theta_2)(B+|y|)\left(1- \min\limits_{t\leq s_1\leq s_2 \leq T}\frac{S_{s_2}}{S_{s_1}}\right)  \geq w-z\right). \label{equ: distriesti}
\end{align}
We need the following lemmas for the proof of the proposition. 
\begin{lemma}\label{lem:hitting prob}
For any constant $w>z>0$ and $0<q<1$, we have
\begin{align}
\max\limits_{\Pi} \mathbb{P}\left(\max\limits_{t\leq s\leq T}Z^{(0)}_s \geq w|z_t = z, \nu_t = \nu\right) = \max\limits_{\Pi} \mathbb{P}(Z^{(0)}_T \geq w|z_t = z, \nu_t = \nu) \leq F^{(q)}(t, \nu) \frac{z^q}{w^q}. \label{equ: L_t}
\end{align} 
\end{lemma}
\begin{proof}[Proof of Lemma \ref{lem:hitting prob}]
Since $\widetilde{U}(\zeta) : = \frac{\zeta^q}{w^q} \geq \mathbf{1}_{\zeta\geq w}$, we have 
\begin{align}
\max\limits_{\Pi} \mathbb{P}\left(Z^{(0)}_T\geq w|z_t = z, \nu_t = \nu\right) \leq \max\limits_{\Pi}  \mathbb{E}\left[\widetilde{U}(Z^{(0)}_T)| z_t = z, \nu_t = \nu\right] = F^{(q)}(t, \nu)  \frac{z^q}{w^q} \notag
\end{align}
from the closed-form solution for CRRA utilities.
\end{proof}

\begin{lemma}\label{lem:hitting prob GBM}
For any constant $\epsilon>0$, when 
\begin{align}\label{equ con lem GBM}
{T-t <\min \{1, T- {t_{\nu, 2|\nu|}}\} \ \text{and}\ \frac{\epsilon}{4 \sqrt{T-t}} > C_{\nu}}, 
\end{align}
we have 
\begin{align}
\max\bigg\{\mathbb{P} \left(\max \limits_{t\leq s_1\leq s_2 \leq T} \ln\frac{S_{s_2}}{S_{s_1}} \geq \epsilon\right), \mathbb{P} \left(\min \limits_{t\leq s_1\leq s_2 \leq T} \ln\frac{S_{s_2}}{S_{s_1}} \leq -\epsilon\right)  \bigg\}\leq \frac{1}{d_{\nu}} e^{-d_{\nu} \frac{\epsilon}{ \sqrt{T-t}}}, \label{equ: Osc}
\end{align}
for some constant $d_{\nu}>0$, which is locally uniformly positive with respect to $\nu$.
\end{lemma}
\begin{proof}[Proof of Lemma \ref{lem:hitting prob GBM}]
From the Assumption \ref{ass:2}, we have 
\begin{align}
	a_N := &  \max \limits_{ |\nu|< N}|\eta(\nu) -\frac{1}{2}\sigma^2(\nu)|  
	\leq  L N + |\eta(0)|+ \frac{1}{2} \bigg(\sigma(0) +LN\bigg)^2 .\label{equ aN}
\end{align}
Given the condition that $\frac{\epsilon}{4 \sqrt{T-t}} > C_{\nu}$, we  can find $N>0$ such that the right hand equals $\frac{\epsilon}{4 \sqrt{T-t}}$. By definition of $N$ and $C_\nu$, we have 
\begin{align}
N \geq \max \bigg\{C \frac{\sqrt{\epsilon}}{(T-t)^{1/4}}, 2|\nu|\bigg\}>0. \label{equ N lb} 
\end{align}
for some constant $C$ since the right hand side of \eqref{equ aN} grows in speed of $N^2$. 

When {$T-t\leq1$}, analogous to \eqref{equ max S esti 1}, \eqref{equ max S esti 2}, since $a_N(T-t) \leq \frac{\epsilon}{4 \sqrt{T-t}} (T-t) \leq \frac{\epsilon}{4}$,  we have 
\begin{align}
	& \mathbb{P} \left(\max \limits_{t\leq s_1\leq s_2 \leq T} \ln\frac{S_{s_2}}		{S_{s_1}} \geq \epsilon\right)  \\
	\leq & \mathbb{P} \left(\max \limits_{t\leq s \leq T} \ln\frac{S_{s}}		{S_{t}} \geq \frac{\epsilon}{2}\right) +  \mathbb{P} \left(\min \limits_{t\leq s \leq T} \ln\frac{S_{s}}		{S_{t}} \leq -\frac{\epsilon}{2}\right)\\
	\leq & 2\bigg[ 2 \Phi \left(\frac{\min\{-\frac{\epsilon}{2} +a_N(T-t), 0\}}{\max \limits_{|\nu|\leq N} \sigma(\nu) \sqrt{T-t}}\right) + \mathbb{P} \left( \max \limits_{t\leq s \leq T} |\nu_s | \geq N \right) \bigg]\\
	\leq & 2\bigg[ 2 \Phi \left(\frac{\min\{-\frac{\epsilon}{4}, 0\}}{\max \limits_{|\nu|\leq N} \sigma(\nu)\sqrt{T-t}}\right) + \mathbb{P} \left( \max \limits_{t\leq s \leq T} |\nu_s| \geq N \right) \bigg]\\
	\leq &4\Phi \left(\frac{-{\epsilon}}{4\max \limits_{|\nu|\leq N} \sigma(\nu)\sqrt{T-t}}\right) + 2\mathbb{P} \left( \max \limits_{t\leq s \leq T} |\nu_s | \geq N \right).  \label{equ max S esti}
\end{align}
Noticing when $\xi<0$, 
\begin{align}
\Phi(\xi) = \frac{1}{\sqrt{2\pi}} \int^{\infty}_{|\xi|} e^{-\frac{1}{2} t^2} dt 
\leq  \frac{1}{\sqrt{2\pi}}   \int^{\infty}_{|\xi|} \frac{t}{|\xi|}e^{-\frac{1}{2} t^2} dt
=  \frac{1}{\sqrt{2\pi} |\xi|} e^{-\frac{1}{2} \xi^2}, \label{equ esti cdfn}
\end{align}
and from the definition of $N$ and \eqref{equ con lem GBM}, we have
\begin{align}
\max \limits_{|\nu|\leq N} \sigma^2(\nu) (T-t) \leq  \bigg(\sigma(0) +LN\bigg)^2 (T-t) \leq 2 C_{\nu}(T-t) \leq  \frac{\epsilon}{2} \sqrt{T-t}.
\end{align} 
Therefore,  we can prove  
\begin{align}
\Phi \left(\frac{-{\epsilon}}{4\max \limits_{|\nu|\leq N} \sigma(\nu)\sqrt{T-t}}\right)\leq \frac{4\max \limits_{|\nu|\leq N} \sigma(\nu)\sqrt{T-t}}{\sqrt{2\pi}\epsilon}  e^{-\frac{1}{2} \frac{\epsilon^2}{16\max \limits_{|\nu|\leq N} \sigma^2(\nu) (T-t)}} \leq \frac{2 (T-t)^{1/4}}{\sqrt{\pi\epsilon}}e^{-\frac{\epsilon}{16 \sqrt{T-t}}} . \label{equ max S esti2}
\end{align}
Since $\xi: = \frac{\epsilon}{4 \sqrt{T-t}} > C_{\nu}>0$, we have the right hand side of \eqref{equ max S esti2} equals
\begin{align}\label{equ est Da}
\frac{1}{\sqrt{\pi} \sqrt{\xi}}e^{-\frac{\xi}{4 }}  \leq \frac{1}{d_\nu} e^{-\frac{\xi}{4 }}
\end{align} 
for some constant $d_\nu>0$. 

From Assumption \ref{ass:2} and \eqref{equ N lb}, 
\begin{align}
	& \mathbb{P} \left( \max \limits_{t\leq s \leq T} |\nu_s| \geq N \right) 
	\leq  \mathbb{P} \left( \max \limits_{t\leq s \leq T} |\nu_s -{\nu}| \geq N - |\nu| \right) \\
	\leq &  \widetilde{C}_{\nu, 2 |\nu|} e^{- \delta_{\nu} \frac{(N - |\nu|)^2}{T-t}}
	\leq     \widetilde{C}_{\nu, 2 |\nu|} e^{- \delta_{\nu} \frac{N ^2}{4(T-t)}}
	\leq    \widetilde{C}_{\nu, 2 |\nu|} e^{- \delta_{\nu} \frac{C^2\epsilon }{4(T-t)^{3/2}}}.  \label{equ max S esti3}
\end{align}
Therefore, from \eqref{equ max S esti}, \eqref{equ max S esti2}, \eqref{equ est Da}, and \eqref{equ max S esti3}, we have 
\begin{align}
	     \mathbb{P} \left(\max \limits_{t\leq s_1\leq s_2 \leq T} \ln\frac{S_{s_2}}		{S_{s_1}} \geq \epsilon\right)  \leq & \frac{4}{d_{\nu}} e^{-\frac{\epsilon}{16 \sqrt{T-t}}}+ \widetilde{C}_{\nu, 2 |\nu|} e^{- \delta_{\nu} \frac{C^2\epsilon }{4(T-t)^{3/2}}} 
	\leq \frac{4}{d_{\nu}} e^{-d_{\nu} \frac{\epsilon}{ \sqrt{T-t}}}
\end{align}
for some constants $d_{\nu}>0$.
\end{proof}
\begin{lemma}\label{lem:esti delta}
For any integer $n>0$, there is a constant $\delta_n>0$, such that
\begin{align}
(1+A)^\zeta \geq \delta_n (\zeta A)^n, \quad \forall \ \zeta\geq 2n,\ A \geq 0. \label{equ: expoesti}
\end{align}
\end{lemma}
\begin{proof}[Proof of Lemma \ref{lem:esti delta}]
Denote $\lfloor \zeta \rfloor$ the maximum integer no larger than $\zeta$, we have
\begin{align}
(1+A)^\zeta \geq (1+A)^{\lfloor \zeta \rfloor} \geq C_{\lfloor \zeta \rfloor}^n  A^n. \notag
\end{align} 
By noticing
\begin{align}
\inf \limits_{\zeta\geq 2n} \frac{C_{\lfloor \zeta \rfloor}^n}{\zeta^n} = \delta_n >0, \notag
\end{align}
we finish the proof.
\end{proof}

Now we turn back to the proof of Proposition \ref{pro:estipro}. We first focus on \eqref{equ distriesti1}.
According to Lemma \ref{lem:hitting prob}, we have 
\begin{align}
\mathbb{P}\left(\max\limits_{t\leq s\leq T} Z^{(0)}_s \geq \theta_1 B\right) \leq F^{(q)}(t, \nu)  \frac{z^q}{\theta_1^q} B^{-q}. \label{equ CRRA max}
\end{align}
Denote $\epsilon = \ln \left( \frac{w-z}{(1-\theta_1) (B+|y|)}+1  \right)$.
According to Lemma \ref{lem:hitting prob GBM}, when $B = w^\alpha (T-t)^{-1/4}$  
and  $${T -t \leq   \frac{z^{4\alpha}}{|{y}|^4} \leq \frac{w^{4\alpha}}{|{y}|^4}, }$$
we have $B \geq |y|$, and from $\ln (1+\zeta) \geq  \frac{1}{2} \zeta$, $\forall\ 0\leq \zeta \leq 1$, we know that 
\begin{align}
\epsilon \geq  \ln \left( \frac{w-z}{2 (1-\theta_1) B}+1  \right) \geq \min \bigg\{ \ln 2,   \frac{w-z}{4 (1-\theta_1) B} \bigg\}.
\end{align}
Therefore, if we further have 
\begin{align}
{T-t \leq \min \bigg\{1,  \bigg(\frac{w-z}{16C_{\nu}(1-\theta_1)w^\alpha } \bigg)^4, \bigg(\frac{\ln 2}{4C_{\nu}}\bigg)^2 \bigg\},  }
\end{align} 
we have \eqref{equ con lem GBM} is satisfied and thus from Lemma \ref{lem:hitting prob GBM}, 
\begin{align}
 \mathbb{P}\left((1-\theta_1)(B+|y|)\left(\max \limits_{t\leq s_1\leq s_2 \leq T}\frac{S_{s_2}}{S_{s_1}}-1\right)  \geq w-z\right)  \leq \frac{1}{d_{\nu}} e^{-d_{\nu} \frac{\epsilon}{ \sqrt{T-t}}}. \label{equ stock max}
\end{align}
In summary, from \eqref{equ CRRA max} and \eqref{equ stock max}
we have under this selection of $B = C^\alpha (T-t)^{-1/4}$,
\begin{align}
&  \mathbb{P}(\max\limits_{t\leq s\leq T} Z^{(0)}_s \geq \theta_1 B)  + \mathbb{P}\left((1-\theta_1)(B+|y|)\left(\max \limits_{t\leq s_1\leq s_2 \leq T}\frac{S_{s_2}}{S_{s_1}}-1\right)  \geq w-z\right) \notag \\
 \leq & F^{(q)}(t, \nu)  \frac{z^q}{(\theta_1)^q} (T-t)^{q/4} w^{-\alpha q} \label{equ term esti 1}\\
 & +\frac{1}{d_{\nu}}  e^{-d_{\nu} \frac{\epsilon}{ \sqrt{T-t}}} \label{equ term esti 2}. 
\end{align}

We also have  
\begin{align}
\eqref{equ term esti 1} 
\leq  F^{(q)}(t, \nu)  \frac{z^q}{(\min \{\theta_1,\theta_2\})^q} (T-t)^{q/4} w^{-\alpha q}.  \notag
\end{align}
Noticing $\epsilon = { \ln \left( \frac{w-z}{(1-\theta_1) (B+|y|)}+1  \right)} \geq {\ln \left( \frac{w-z}{2 (1-\theta_1) B}+1  \right)}$, we have 
\begin{align}
e^{-d_\nu \frac{\epsilon}{ \sqrt{T-t}}} \leq  \left( \frac{w-z}{2 (1-\theta_1) w^\alpha (T-			t)^{-1/4}}+1  \right)^{- \frac{d_\nu}{\sqrt{T-t}} }.
\end{align}
For any integer $n>0$, according to Lemma \ref{lem:esti delta},  when {$T-t \leq (\frac{d_{\nu}}{2n})^2$},  
\begin{align*}
&    \left( \frac{w-z}{2 (1-\theta_1) w^\alpha (T-t)^{-1/4}}+1  \right)^{\frac{-d_{\nu}}{\sqrt{T-t}} } \\
	\leq & \frac{1}{ \delta_n} \bigg[\frac{d_{\nu}(w-z)}{2(1-\theta_1) w^\alpha (T-		t)^{-1/4} \sqrt{T-t}} \bigg]^{-n}      \\
	= & \frac{1}{ \delta_n}  \bigg[\frac{2(1-\theta_1) w^\alpha (T-t)^{1/4}} {d_{\nu}(w-z)}\bigg]^n  \\
	= &  \frac{1}{\delta_n}  \bigg[\frac{2 (1-\theta_1) w^\alpha} {d_{\nu}(w-z)}\bigg]^n  (T-t)^{n/4}\\
	\leq &  \frac{1}{\delta_n}  \bigg[\frac{2 (1+\theta_2) w^\alpha} {d_{\nu}(w-z)}\bigg]^n  (T-t)^{n/4}
\end{align*}
In summary, we have 
\begin{align}
\eqref{equ term esti 1} + \eqref{equ term esti 2}   \leq C_{\nu, n} \bigg(z^q (T-t)^{q/4} w^{-\alpha q} + \bigg(\frac{w^\alpha} {w-z}\bigg)^n  (T-t)^{n/4}\bigg).
\end{align}
Similarly, we can handle \eqref{equ: distriesti}. 
This finishes the proof.
\end{proof}

\subsubsection{Proof of Theorem \ref{prop:term} }\label{App A GBM 3}

We prove the equality of \eqref{equ: tercon} by showing that inequalities hold in both directions.

{\flushleft\textbf{We first show the left side of \eqref{equ: tercon} is no greater than $0$.}}

For any constant $\epsilon>0$, there is a $0< \delta <\hat{z}$, s.t. $U(w)\leq U(\hat{z})+\epsilon$, when $w\leq \hat{z} +\delta$, and $U(w) \geq U(\hat{z}-)-\epsilon$ when $w\geq \hat{z}-\delta$. We have
\begin{align}
&\mathbb{E}[U(Z_T)\mathbf{1}_{{Z_T}\geq \hat{z}+ \delta}] \notag\\ 
=  & \int_{\hat{z}+\delta}^{+\infty} U(w) d \mathbb{P}(Z_T\leq w)\notag\\
\leq &  \int_{\hat{z}+\delta}^{+\infty} (C_1+C_2 w^p) d \mathbb{P}(Z_T\leq w)\notag\\
= & C_1\mathbb{P}(Z_T\geq \hat{z}+\delta) + C_2 (\hat{z}+\delta)^p \mathbb{P}(Z_T\geq \hat{z}+\delta) +C_2 \int_{\hat{z}+\delta}^{+\infty} p w^{p-1} \mathbb{P} (Z_T\geq w) dw. \label{equ:expext >z0}
\end{align} 
We let  $p<q<1$, $\frac{p}{q}<\alpha<1$, $n\in\mathbb{N}^+\cap (1/(1-\alpha), +\infty)$, then according to Proposition \ref{pro:estipro},
\begin{align*}
& \limsup \limits_{(t, x, y, \nu) \to ( T^-, \hat{x}, \hat{y}, \hat{\nu})} \int_{\hat{z}+\delta}^{+\infty} p w^{p-1}  \mathbb{P} (Z_T\geq w) dw \\
	 \leq & \limsup \limits_{(t, x, y, \nu) \to ( T^-, \hat{x}, \hat{y}, \hat{\nu})} C_{\nu, n} \bigg[\int_{\hat{z}+\delta}^{+\infty} p w^{p-1} {z^q} (T-t)^{q/4} w^{-\alpha q} dw\\
		 & \qquad \qquad + \int_{\hat{z}+\delta}^{+\infty} p w^{p-1} \bigg(\frac{w^\alpha} {w-z}\bigg)^n  (T-t)^{n/4} dw\bigg] \notag\\
	\leq &  \limsup \limits_{(t, x, y, \nu) \to ( T^-, \hat{x}, \hat{y}, \hat{\nu})} C_{\nu, n}\int_{\hat{z}+\delta}^{+\infty} p  {z^q} (T-t)^{q/4} w^{-\alpha q+p-1} dw\\
	& \qquad  + \limsup \limits_{(t, x, y, \nu) \to ( T^-, \hat{x}, \hat{y}, \hat{\nu})} C_{\nu, n} \int_{\hat{z}+\delta}^{+\infty} p w^{p-1} \bigg(\frac{w^\alpha} {w-z}\bigg)^n  (T-t)^{n/4} dw.\notag
\end{align*}
Since $\alpha q >p$ and $p-1+(\alpha-1) n \leq p-1-1 <-1$, when $T-t\to0$,   
\begin{align}
{\limsup \limits_{(t, x, y, \nu) \to ( T^-, \hat{x}, \hat{y}, \hat{\nu})} \int_{\hat{z}+\delta}^{+\infty} p w^{p-1}  \mathbb{P} (Z_T\geq w)dw } = 0.\notag
\end{align}
Analogously, 
\begin{align}
&\limsup \limits_{(t, x, y, \nu) \to ( T^-, \hat{x}, \hat{y}, \hat{\nu})}  \mathbb{P} (Z_T\geq \hat{z}+\delta) = 0. \notag
\end{align}
Consequently, we have from \eqref{equ:expext >z0}, 
\begin{align}
 \limsup \limits_{(t, x, y, \nu) \to ( T^-, \hat{x}, \hat{y}, \hat{\nu})} \mathbb{E}[U(Z_T)\mathbf{1}_{{Z_T}\geq \hat{z}+ \delta}] = \limsup \limits_{(t, x, y, \nu) \to ( T^-, \hat{x}, \hat{y}, \hat{\nu})}  \int_{\hat{z}+\delta}^{+\infty} U(w) d \mathbb{P}(Z_T\leq w)=0.\label{equ:esti expe >z0}
\end{align}
Set 
\begin{numcases}{\bar{U}(\zeta) := }
	U(\hat{z}-) & $K\leq \zeta< \hat{z}$\notag\\
	U(\hat{z}) & $\zeta\geq \hat{z}$,\notag
\end{numcases}
from \eqref{equ:esti expe >z0} we have 
\begin{align}
	& \limsup \limits_{(t, x, y, \nu) \to ( T^-, \hat{x}, \hat{y}, \hat{\nu})} E[U(Z_T)] -U(\hat{z}-) - 2 \big(U(\hat{z}) - U(\hat{z}-)\big)  \Phi\bigg(\frac{\min\{z-\hat{z},0\} }{ |\hat{z}-x| \sigma(\hat{\nu})\sqrt{T-t} } \bigg) \notag\\
	= & \limsup \limits_{(t, x, y, \nu) \to ( T^-, \hat{x}, \hat{y}, \hat{\nu})}  \int_{K}^{+\infty} U(w) d \mathbb{P}(Z_T\leq w) - U(\hat{z}-) - 2 \big(U(\hat{z}) - U(\hat{z}-)\big)  \Phi\bigg(\frac{\min\{z-\hat{z},0\} }{ |\hat{z}-x| \sigma(\hat{\nu})\sqrt{T-t} } \bigg)  \notag\\
	= & \limsup \limits_{(t, x, y, \nu) \to ( T^-, \hat{x}, \hat{y}, \hat{\nu})}   \int_{\hat{z}+\delta}^{+\infty} U(w) d \mathbb{P}(Z_T\leq w)+  \int_K^{\hat{z}+\delta}U(w) d \mathbb{P}(Z_T\leq w)\notag\\
		& \qquad \qquad - U(\hat{z}-) - 2 \big(U(\hat{z}) - U(\hat{z}-)\big)  \Phi\bigg(\frac{\min\{z-\hat{z},0\} }{ |\hat{z}-x|  \sigma(\hat{\nu})\sqrt{T-t} } \bigg) \notag\\
	= &  \limsup \limits_{(t, x, y, \nu) \to ( T^-, \hat{x}, \hat{y}, \hat{\nu})}   \int_K^{\hat{z}+\delta} U(w) d \mathbb{P}(Z_T\leq w)- U(\hat{z}-) - 2 \big(U(\hat{z}) - U(\hat{z}-)\big)  \Phi\bigg(\frac{\min\{z-\hat{z},0\} }{ |\hat{z}-x| \sigma(\hat{\nu})\sqrt{T-t} } \bigg)\notag  \\
	\leq &  \epsilon+\limsup \limits_{(t, x, y, \nu) \to ( T^-, \hat{x}, \hat{y}, \hat{\nu})}   \int_K^{\hat{z}+\delta} \bar{U}(w) d \mathbb{P}(Z_T\leq w)- U(\hat{z}-) - 2 \big(U(\hat{z}) - U(\hat{z}-)\big)  \Phi\bigg(\frac{\min\{z-\hat{z},0\} }{ |\hat{z}-x|\sigma(\hat{\nu})\sqrt{T-t} } \bigg). \notag
\end{align}
From Proposition \ref{prop:goalreaching}, we have 
\begin{align}
\limsup \limits_{(t, x, y, \nu) \to ( T^-, \hat{x}, \hat{y}, \hat{\nu})} \int_K^{\hat{z}+\delta} \bar{U}(w) d \mathbb{P}(Z_T\leq w) - U(\hat{z}-) - 2 \big(U(\hat{z}) - U(\hat{z}-)\big)  \Phi\bigg(\frac{\min\{z-\hat{z},0\} }{ |\hat{z}-x| \sigma(\hat{\nu})\sqrt{T-t} } \bigg) \leq 0. \notag 
\end{align}
Since $\epsilon$ is arbitrary, 
\begin{align}
\limsup \limits_{(t, x, y, \nu) \to ( T^-, \hat{x}, \hat{y}, \hat{\nu})}   \mathbb{E}[U(Z_T)] - 2 \big(U(\hat{z}) - U(\hat{z}-)\big)  \Phi\bigg(\frac{\min\{z-\hat{z},0\} }{ |\hat{z}-x| \sigma(\hat{\nu})\sqrt{T-t} } \bigg) \leq 0. \notag 
\end{align}

{\flushleft\textbf{We next show the left side of \eqref{equ: tercon} is no less than $0$.}}

When $\hat{y} >0$, consider the strategy $\pi^*$ which does not sell or buy in $[t, \tau_{\hat{z}}\wedge \tau_{\hat{z}-\delta})$, and sell all stock at $\tau_{\hat{z}}\wedge \tau_{\hat{z}-\delta}$, where $\tau_{\hat{z}}: = \inf \{ s\in [t, T]| Z^{\pi^*}_s\geq \hat{z}\}$, $\tau_{\hat{z}-\delta}: = \inf \{ s\in [t, T]| Z^{\pi^*}_s\leq \hat{z}-\delta\}$. 
For any $\hat{z}>0$, we have from \eqref{equ prob low}
\begin{align}
	      &\limsup \limits_{(t, x, y, \nu) \to ( T^-, \hat{x}, \hat{y}, \hat{\nu})}\mathbb{P}\big(Z^{\pi^*}_T \leq \hat{z}-\delta\big) \\
	\leq & \limsup \limits_{(t, x, y, \nu) \to ( T^-, \hat{x}, \hat{y}, \hat{\nu})} \mathbb{P} \left(\min\limits_{t\leq s \leq T}\frac{S_{s}}{S_{t}}  \leq \frac{\hat{z}-\delta-x}{(1-\theta_1)y} \right) \\
	\leq &  \limsup \limits_{(t, x, y, \nu) \to ( T^-, \hat{x}, \hat{y}, \hat{\nu})}  2 \Phi  \left(\frac{\min \{\hat{z}-\delta-z+ a(1-\theta_1)y(T-t), 0\}} {(1-\theta_1)y\sqrt{T-t}\min \limits_{\nu \in \Omega_\delta}\sigma(\nu)}\right)+ \mathbb{P} \left( \max \limits_{t\leq s \leq T} |\nu_s -\hat{\nu}| \geq \delta \right)\\
	= & 0. 
\end{align}
Therefore, when $(t, x, y, \nu)$ is sufficiently close to $( T^-, \hat{x}, \hat{y}, \hat{\nu})$, we have $\mathbb{P}\big(Z^{\pi^*}_T \leq \hat{z}-\delta\big)\leq \epsilon$, and 
\begin{align}
       &V(t, x, y, \nu) \\
	\geq & U(\hat{z}-) \mathbb{P}(\hat{z}-\delta \leq Z^{\pi^*}_T \leq \hat{z})+ U(\hat{z})  \mathbb{P}( Z^{\pi^*}_T \geq \hat{z}) \\
	\geq & U(\hat{z}-) \big(1-\mathbb{P}( Z^{\pi^*}_T \geq \hat{z}) \big)+ U(\hat{z})  \mathbb{P}( Z^{\pi^*}_T \geq \hat{z}) -\epsilon | U(\hat{z}-) |, 
\end{align}
while the first two terms is actually the goal-reaching problem's utility. Then from Proposition  \ref{equ: termi},  
\begin{align}
	&\liminf \limits_{(t, x, y, \nu) \to ( T^-, \hat{x}, \hat{y}, \hat{\nu})}  V(t, x, y, \nu) - U(\hat{z}-) - 2 \big(U(\hat{z}) - U(\hat{z}-)\big)  \Phi\bigg(\frac{\min\{z-\hat{z},0\} }{ |\hat{z}-x| \sigma(\hat{\nu})\sqrt{T-t} } \bigg) \notag  \\
	\geq & \liminf \limits_{(t, x, y, \nu) \to ( T^-, \hat{x}, \hat{y}, \hat{\nu})} U(\hat{z}-) \big(1-\mathbb{P}( Z^{\pi^*}_T \geq \hat{z}) \big)+ U(\hat{z})  \mathbb{P}( Z^{\pi^*}_T \geq \hat{z}) -\epsilon | U(\hat{z}-) | \notag \\
		& \qquad - U(\hat{z}-) - 2 \big(U(\hat{z}) - U(\hat{z}-)\big)  \Phi\bigg(\frac{\min\{z-\hat{z},0\} }{ |\hat{z}-x| \sigma(\hat{\nu})\sqrt{T-t} } \bigg)\notag\\
\geq &  -\epsilon| U(\hat{z}-) |.\label{equ:liminf w terminal}
\end{align}
The result is derived by noticing that $\epsilon$ can be arbitrarily small. 

For the case $\hat{y}\leq0$, the proof is similar, and that finishes our proof.\qed

\subsection{Proof of Theorem \ref{thm compa principle}}\label{sect:comparison}

We prove by contradiction. Consider $\psi(t, x, y, \nu) = e^{\beta(t-T)} u(t, x, y, \nu)$ and  $\phi(t, x, y, \nu) = e^{\beta(t-T)} v(t, x, y, \nu)$, where ${\beta>0}$, then $\psi$ (resp. $\phi$) is a viscosity subsolution (resp. supersolution) to 
\begin{align}
&\min\{-F_t -\frac{1}{2}\bigg(\sigma^2(\nu) y^2 F_{yy} + \zeta^2(\nu) F_{\zeta\zeta} + 2 \rho \sigma(\nu) y \zeta(\nu)F_{y\zeta} \bigg) - \eta(\nu) y F_y-m(\nu) F_\nu + \beta F ,\\
&\qquad \qquad  F_y - (1-\theta_1)F_x, \quad (1+\theta_2)  F_x - F_y   \} = 0, \notag
\end{align}
with the boundary condition  for $x_0+ (1-\theta_1) y_0^+ - (1+\theta_2) y_0^- = K$ 
\begin{align}
\begin{cases}
\lim \limits_{(s, x, y, \nu) \rightarrow(t, x_0, y_0, \nu_0)} F(s, x, y, \nu)=e^{\beta(t-T)} U(K) & \text { if } U(K)>-\infty \\ 
\lim \limits_{(s, x, y, \nu) \rightarrow(t, x_0, y_0, \nu_0)}  F(s, x, y, \nu)-e^{\beta(t-T)}V_{CRRA}(s, x, y, \nu)=0 & \text { if } K=0 \text { and } U(0)=-\infty,
\end{cases}
\end{align}
and the terminal condition \eqref{equ: tercon}.

{We only consider the case with $|U(K)|< \infty$, the other case can be handled in a similar approach by considering the function $F(t, x, y, \nu) - e^{\beta(t-T)}V_{CRRA}(t, x, y, \nu)$. }

Assume on the contrary there is some point $(\bar{t}, \bar{x}, \bar{y}, \bar{\nu}) \in (t_p, T)\times \mathscr{S}$ such that 
\begin{align}
\psi(\bar{t}, \bar{x}, \bar{y}, \bar{\nu}) -\phi(\bar{t}, \bar{x}, \bar{y}, \bar{\nu}) = 2\delta >0. \notag
\end{align}
Define function
\begin{align}
G(t, x, y, \nu, s, u, v, \xi) :=& \big((x+y)^q+C\big) F^{(q)}(t-(q-\hat{p}), \nu)+\big((u+v)^q+C\big) F^{(q)}(s-(q-\hat{p}), {\xi})\\
: =& H^{(1)} (t, x, y, \nu) + H^{(2)} (s, u, v, \xi)
\end{align}
where $q\in (p, 1)$ such that $t_{q} < \bar{t} - (q-{p})$ and $C> \max \limits_{w\geq 0} (C_1 w^{{p}}+C_2- w^{q})$.
 Then consider 
\begin{align}
M_\alpha (t, x, y, \nu, s, u, v, \xi): = &\psi(t, x, y, \nu) - \phi(s, u, v, \xi) -\varphi(t, x, y, \nu, s, u, v, \xi), \notag
\end{align}
where 
\begin{align}
&\varphi(t, x, y, \nu, s, u, v, \xi): \\
&\qquad =  \frac{\epsilon_1}{t-(t_q+q-{p})} + \frac{\epsilon_2}{T-t} + \epsilon_3 G(t, x, y, \nu, s, u, v, \xi) + \frac{\alpha}{2} \left(  (t-s)^2+ (y-v)^2  +(x-u)^2 + (\nu-\xi)^2   \right), 
\end{align}
and $\epsilon_1$, $\epsilon_2$, $\epsilon_3$ are three positive constants which are sufficiently small, s.t. $$M_\alpha (\bar{t}, \bar{x}, \bar{y}, \bar{\nu}, \bar{t}, \bar{x}, \bar{y}, \bar{\nu}) > \delta>0. $$ 

First we show that for any $\alpha>0$, we can find an interior point $(t, x, y, \nu, s, u, v, \xi) \in (t_q+q-{p}, T) \times \mathscr{S} \times  (t_q+q-{p}, T) \times \mathscr{S}$, such that $M_\alpha$ attains global maximum. 
On the one hand, when $x+y$ is large, 
\begin{align}
	\psi(t, x, y, \nu)\leq & C_1 (x+y)^{{p}} F^{({p})}(t, \nu)+C_2 \\
	\leq &\big(C_1 (x+y)^{{p}}+C_2\big) F^{({p})}(t, \nu)\\
	\leq & \epsilon_3(x+y)^qF^{(q)}(t, \nu)
\end{align}
On the other hand, we have that 
\begin{align}
\epsilon_d : = \min \limits_{t_{q} + q-{p} < t \leq T } \{A_q(t - (q-{p})) - A_{{p}}(t) \} >0. 
\end{align}
For the remaining set, when $x+y$ is bounded,  if $|\nu|$ is large enough, 
\begin{align}
	\psi(t, x, y, \nu) \leq & C_1 (x+y)^{{p}} F^{({p})}(t, \nu)+C_2\\
	\leq &\big(C_1 (x+y)^{{p}}+C_2\big) F^{({p})}(t, \nu)\\
	= & \big(C_1 (x+y)^{{p}}+C_2\big)e^{A_{{p}}(t) \nu^2 +B_{{p}}(t) \nu +C_{{p}}(t)}\\
	\leq &\big(C_1 (x+y)^{{p}}+C_2\big) e^{-\frac{\epsilon_d}{2}\nu^2} F^{(q)}(t-(q-{p}), \nu)\\
	\leq &  \epsilon_3\bigg((x+y)^q+C\bigg) F^{(q)}(t-(q-{p}), \nu) .
\end{align}
Given the lower bound over $\phi$ from the definition of set $\mathcal{C}$, when $x+y$ or $|\nu|$ is sufficiently larger, $M_\alpha$ will be negative. Moreover, given bounded $x+y$ and $|\nu|$, if $u+v$ or $|\xi|$ is sufficiently large, $M_\alpha$ is also negative.
Therefore, $M_\alpha$ takes maximum only on a bounded subset of $(t_q+q-{p}, T) \times \mathscr{S} \times  (t_q+q-{p}, T) \times \mathscr{S}$ independent of $\alpha$, $\epsilon_1$, and $\epsilon_2$.  
Moreover, this maximum cannot be taken on $z = K$ for the boundary condition we imposed. 
We denote one of the maximizers by $(t_\alpha, x_\alpha, y_\alpha, \nu_\alpha, s_\alpha, u_\alpha, v_\alpha, \xi_\alpha)$. 

Second, we show that we can find a sufficiently small $\epsilon_2 >0$, such that 
\begin{align}
\frac{\epsilon_1}{(t_\alpha-(t_q+q-{p}))^2} \geq \frac{\epsilon_2}{(T-t_\alpha)^2},\ \text{for sufficiently large $\alpha$}.
\label{eq:cp3}
\end{align}
Notice that for any small fixed $\epsilon_1$, $\epsilon_2$, and $\epsilon_3$ there exists a subsequence such that 
\begin{align}
\lim \limits_{\alpha \to +\infty}\alpha \big( (t_\alpha-s_\alpha)^2 +(x_\alpha-u_\alpha)^2+  (y_\alpha-v_\alpha)^2 + (\nu-\xi)^2 \big) = 0,\label{equ:alpha lim}
\end{align}
and both $(t_\alpha, x_\alpha, y_\alpha, \nu_\alpha)$
and $ (s_\alpha,u_\alpha, v_\alpha, \xi_\alpha)$ converge to some interior point $(\hat{t}, \hat{x}, \hat{y}, \hat{\nu})$ (see, e.g. \cite{crandall1992user}). 
Let $(\hat{t}_0, \hat{x}_0, \hat{y}_0, \hat{\nu}_0)$ be a limit of $(\hat{t}, \hat{x}, \hat{y}, \hat{\nu})$ as $\epsilon_2\to 0$.

We next show that $\hat{t}_0<T$, and \eqref{eq:cp3} is proved accordingly.
We prove by contradiction.
If $\hat{t}_0=T$, according to the terminal condition \eqref{equ: tercon}, 
we have
\begin{align*}
	& \limsup \limits_{(\hat{t}, \hat{x}, \hat{y},\hat{\nu}) \to (T-, \hat{x}_0, \hat{y}_0, \hat{\nu}_0)} \bigg(\psi(\hat{t}, \hat{x}, \hat{y},\hat{\nu}) - \phi(\hat{t}, \hat{x}, \hat{y},\hat{\nu})\bigg)  \\
	\leq & \limsup \limits_{(\hat{t}, \hat{x}, \hat{y},\hat{\nu}) \to (T-, \hat{x}_0, \hat{y}_0, \hat{\nu}_0)}\psi(\hat{t}, \hat{x}, \hat{y},\hat{\nu})  - U(\hat{z}_0-) -  2 \Phi\left(\frac{\min\{z-\hat{z}, 0\}}{ |\hat{z}-x| \sigma(\nu_0)\sqrt{T-t} } \right) (U(\hat{z}_0) - U(\hat{z}_0-) ) \\
	\quad &- \liminf \limits_{(\hat{t}, \hat{x}, \hat{y},\hat{\nu}) \to (T-, \hat{x}_0, \hat{y}_0, \hat{\nu}_0)} \phi(\hat{t}, \hat{x}, \hat{y},\hat{\nu}) - U(\hat{z}_0-) -  2 \Phi\left(\frac{\min\{z-\hat{z}, 0\}}{ |\hat{z}-x| \sigma(\nu_0)\sqrt{T-t} } \right) (U(\hat{z}_0) - U(\hat{z}_0-) ) \\
 \leq & 0,
\end{align*}
which contradicts the fact that, for each $\epsilon_2$ and $\alpha$, 
\begin{align}
       & \psi(t_\alpha, x_\alpha, y_\alpha, \nu_\alpha) - \phi(t_\alpha, x_\alpha, y_\alpha, \nu_\alpha)\\
  & \quad    \geq  M_\alpha(t_\alpha, x_\alpha, y_\alpha, \nu_\alpha, s_\alpha, u_\alpha, v_\alpha, \xi_\alpha)\geq M_\alpha(\bar{t}, \bar{x}, \bar{y}, \bar{\nu},\bar{t}, \bar{x}, \bar{y}, \bar{\nu})>\delta>0.\label{eq:cp5}
\end{align}
Third, we apply Ishii's lemma to prove the theorem. For notational simplicity, we use $(t, x, y, \nu, s, u, v, \xi)$ for $(t_\alpha, x_\alpha, y_\alpha, \nu_\alpha, s_\alpha, u_\alpha, v_\alpha, \xi_\alpha)$. By Ishii's lemma, for any $\gamma>0$, there are constants $M$ and $N$ such that
\begin{align}
& \min\{- H(t, x, y, \nu,  \nabla \varphi, M) + \beta \psi, \varphi_y - (1-\theta_1)\varphi_x, (1+\theta_2)  \varphi_x - \varphi_y   \} \leq 0,\label{equ: cpsub} \\
&  \min\{-H(s, u, v, \xi, -\nabla \varphi, N)+ \beta \phi, -\varphi_v + (1-\theta_1)\varphi_u, -(1+\theta_2)  \varphi_u + \varphi_v   \} \geq 0,\label{equ: cpsuper} 
\end{align}
where 
\begin{align}
& H(t, x, y, \nu, \nabla \varphi, A) \\
& \qquad : = \varphi_t +\frac{1}{2}\bigg(\sigma^2(\nu) y^2 A_{11} + \zeta^2(\nu) A_{22} + 2 \rho \sigma(\nu) y \zeta(\nu) A_{12} \bigg) + \eta( \nu) y \varphi_y+m(\nu) \varphi_\nu \\
& \qquad =  \varphi_t +\frac{1}{2} Tr (A \Sigma_{y, \nu}) + \eta(\nu) y \varphi_y+m(\nu) \varphi_\nu,
\end{align}
with $ \Sigma_{y, \nu}: = 
 \begin{pmatrix}
   \sigma^2(\nu) y^2 &\rho \sigma(\nu) y \zeta(\nu)\\
   \rho \sigma(\nu) y \zeta(\nu) & \zeta^2(\nu)
  \end{pmatrix}$
and 
\begin{align}
\left (
 \begin{matrix}
   M &0\\
   0 &-N
  \end{matrix}
  \right)               \leq
    \nabla^2_{y, \nu, v, \xi} \ \varphi+ \gamma
  \left (\nabla^2_{y, \nu, v, \xi} \ \varphi
  \right) ^2 \label{eq:Ishii second}
\end{align}
with
\[
\nabla^2_{y, \nu, v, \xi} \ \varphi= \left (
 \begin{matrix}
   \nabla^2_{y, \nu} \varphi & -\alpha I_2\\
    -\alpha I_2 & \nabla^2_{v, \xi} \varphi 
  \end{matrix}
  \right).
\]
From the definition of $\varphi$, we have 
\begin{align}
&\varphi_t =  -\frac{\epsilon_1}{(t-(t_q+q-{p}))^2} + \frac{\epsilon_2}{(T-t)^2}+\epsilon_3 H^{(1)}_t+ \alpha (t-s),\\ 
& \varphi_s = \epsilon_3H^{(2)}_s-\alpha (t-s),\notag\\
& \varphi_x = \epsilon_3 H^{(1)}_x +\alpha (x-u), \qquad  \varphi_y = \epsilon_3 H^{(1)}_y + \alpha (y-v),\notag\\
&  \varphi_u =\epsilon_3 H^{(2)}_u  - \alpha (x-u), \qquad   \varphi_v = \epsilon_3 H^{(2)}_v - \alpha (y-v), \notag\\
& \varphi_\nu =  \epsilon_3 H^{(1)}_\nu +\alpha (\nu-\xi), \qquad \varphi_\xi =  \epsilon_3H^{(2)}_\xi -\alpha (\nu-\xi)
\end{align}
and 
\begin{align}
&\nabla^2_{y, \nu} \ \varphi = \nabla^2_{y, \nu} H^{(1)}+ \alpha I_2, \quad 
\nabla^2_{v, \xi} \ \varphi = \nabla^2_{v, \xi} H^{(2)}+ \alpha I_2
\end{align}

According to  \eqref{equ: cpsub}, one of the following three inequalities should be satisfied. 
  
{\flushleft (i) $\varphi_y - (1-\theta_1) \varphi_x \leq 0$.}

We have from \eqref{equ: cpsuper} that 
\begin{align}
0\geq & [\varphi_y - (1-\theta_1) \varphi_x]  - [ -\varphi_v + (1-\theta_1) \varphi_u  ]  =\theta_1\epsilon_3 \bigg[ H^{(1)}_y + H^{(2)}_v \bigg]>0, \notag
\end{align}
which cannot hold, and hence leads to a contraction.

{\flushleft (ii) $(1+\theta_2) \varphi_x -\varphi_y \leq 0$.}

We have from \eqref{equ: cpsuper} that 
\begin{align}
0\geq & [(1+\theta_2)\varphi_x - \varphi_y]  -  [-(1+\theta_2)\varphi_u + \varphi_v]  = \theta_2\epsilon_3 \bigg[ H^{(1)}_\nu + H^{(2)}_\xi \bigg] >0, \notag
\end{align}
which again cannot hold, and hence leads to a contraction.

{\flushleft (iii) $- H(t, x, y, \nu,  \nabla \varphi, M) + \beta \psi \leq 0$}

We have from \eqref{equ: cpsuper} that 
\begin{align*}
	0\geq& [-\varphi_t -\frac{1}{2} Tr (M \Sigma_{y, \nu}) - \eta(\nu) y \varphi_y-m(\nu) \varphi_\nu+\beta\psi]\\
	&  + [-\varphi_s +\frac{1}{2} Tr (N \Sigma_{v, \xi}) - \eta(\xi) v \varphi_v - m(\xi) \varphi_\xi-\beta\phi  ] \\
	= & \bigg[\frac{\epsilon_1}{(t-(t_q+q-{p}))^2} - \frac{\epsilon_2}{(T-t)^2}-\epsilon_3  H^{(1)}_t-\alpha (t-s) \\
	& \qquad -\frac{1}{2}Tr (M \Sigma_{y, \nu})  - \eta(\nu) y \bigg(\epsilon_3 H^{(1)}_y + \alpha (y-v)\bigg) -m(\nu) \bigg(\epsilon_3 H^{(1)}_\nu +\alpha (\nu-\xi) \bigg) +\beta \psi\bigg] \\
	&\quad + \bigg[- \epsilon_3 H^{(2)}_s+ \alpha(t-s)+\frac{1}{2} Tr (N \Sigma_{v, \xi}) \\
	& \qquad  - \eta(\xi) v \bigg( \epsilon_3 H^{(2)}_v - \alpha (y-v) \bigg)- m( \xi) \bigg(\epsilon_3 H^{(2)}_\xi-\alpha (\nu-\xi)  \bigg) -\beta \phi   \bigg] \\
	\geq & -\epsilon_3\bigg( H^{(1)}_t + \eta(\nu) y H^{(1)}_y  + m(\nu) H^{(1)}_\nu \bigg) -\epsilon_3\bigg(H^{(2)}_s + \eta(\xi) v H^{(2)}_v  +  m( \xi) H^{(2)}_\xi \bigg)\\
	& \qquad -\frac{1}{2} \big( Tr (M \Sigma_{y, \nu})-Tr (N \Sigma_{v, \xi})\big) + \alpha (y-v) \big( \eta(\xi) v -\eta(\nu) y \big) +\alpha (\nu-\xi) \big(m( \xi) -m(\nu) \big)\\
	& \qquad + \beta (\psi-\phi) \label{equ cp final}
\end{align*}
Since $\Sigma_{y, \nu}$ is symmetric and positive definite, there is $2\times 2$ matrix $\sigma_{y, \nu}: = 
 \begin{pmatrix}
  \sigma(\nu) y, &0\\
  \rho \zeta(\nu), & \sqrt{1-\rho^2} \zeta(\nu)
  \end{pmatrix}$ such that 
\begin{align}
 \Sigma_{y, \nu} =  \sigma_{y, \nu}\sigma^\top_{y, \nu}.
\end{align} 
According to \eqref{eq:Ishii second}, from the Lipschitz continuity of $\sigma( \nu)$ and $\zeta(\nu)$, if we choose sufficiently small $\gamma$ as $\alpha \to \infty$, 
\begin{align}
	&Tr (M \Sigma_{y, \nu})-Tr (N \Sigma_{v, \xi})\\
	= & Tr ( \sigma^\top_{y, \nu} M \sigma_{y, \nu})-Tr (\sigma^\top_{v, \xi} N \sigma_{v, \xi})\\
	= & Tr \bigg(  (\sigma^\top_{y, \nu},  \sigma^\top_{v, \xi}) 
\begin{pmatrix}
M,&0\\
0, &-N
\end{pmatrix}
\begin{pmatrix}
\sigma_{y, \nu}\\
\sigma_{v, \xi}
\end{pmatrix}
\bigg)\\
    \leq & Tr \bigg(  (\sigma^\top_{y, \nu},  \sigma^\top_{v, \xi}) 
\left(\nabla^2_{y, \nu, v, \xi} \ \varphi + \gamma ( \nabla^2_{y, \nu, v, \xi} \ \varphi)^2\right)
\begin{pmatrix}
\sigma_{y, \nu}\\
\sigma_{v, \xi}
\end{pmatrix}
\bigg)\\ 
     \leq & Tr \bigg(  (\sigma^\top_{y, \nu},  \sigma^\top_{v, \xi}) 
\nabla^2_{y, \nu, v, \xi} \ \varphi 
\begin{pmatrix}
\sigma_{y, \nu}\\
\sigma_{v, \xi}
\end{pmatrix}
\bigg) +\gamma Tr \bigg(  (\sigma^\top_{y, \nu},  \sigma^\top_{v, \xi}) 
(\nabla^2_{y, \nu, v, \xi} \ \varphi)^2
\begin{pmatrix}
\sigma_{y, \nu}\\
\sigma_{v, \xi}
\end{pmatrix}
\bigg)\\
	= & \alpha  Tr\left((\sigma_{y, \nu}-\sigma_{v, \xi})^\top(\sigma_{y, \nu}-\sigma_{v, \xi})\right) + \epsilon_3  Tr\bigg(\sigma_{y, \nu}^\top (\nabla^2_{y,\nu} H^{(1)}) \sigma_{y, \nu} + \sigma_{v, \xi}^T (\nabla^2_{v,\xi} H^{(2)}) \sigma_{v, \xi}\bigg)\\
	& \quad +\gamma Tr \bigg(  (\sigma^\top_{y, \nu},  \sigma^\top_{v, \xi}) 
(\nabla^2_{y, \nu, v, \xi} \ \varphi)^2
\begin{pmatrix}
\sigma_{y, \nu}\\
\sigma_{v, \xi}
\end{pmatrix}
\bigg) \\
	= & \epsilon_3  Tr\bigg( (\nabla^2_{y,\nu} H^{(1)}) \Sigma_{t, y, \nu} +  (\nabla^2_{v,\xi} H^{(2)}) \Sigma_{s, v, \xi}\bigg) +o(1),
\end{align}
and
\begin{align}
 \alpha (y-v) \big( \eta(\xi) v -\eta(\nu) y \big) +\alpha (\nu-\xi) \big(m(\xi) -m( \nu) \big) = o(1).
\end{align} 
 Therefore,
 \begingroup
 \allowdisplaybreaks
 \begin{align}
 \eqref{equ cp final} \geq& -\epsilon_3 \bigg( H^{(1)}_t + \frac{1}{2}Tr\bigg( (\nabla^2_{y,\nu} H^{(1)}) \Sigma_{y, \nu}\bigg) + \eta(\nu) y H^{(1)}_y  + m(\nu) H^{(1)}_\nu \bigg) \\
 & \quad -\epsilon_3\bigg(H^{(2)}_s + \frac{1}{2} Tr\bigg( (\nabla^2_{v,\xi} H^{(2)}) \Sigma_{v, \xi}\bigg)+ \eta(\xi) v H^{(2)}_v  +  m(\xi) H^{(2)}_\xi \bigg)+ \beta (\psi-\phi) +o(1). \label{equ cp final2}
 \end{align} 
 Since $h^{(1)}(t, x, y, \nu): = (x+y)^q F^{(q)} (t, \nu)$ corresponds to the closed form solution when no transaction cost is imposed, we have 
\begin{align} 
 & h^{(1)}_t + \frac{1}{2}Tr\bigg( (\nabla^2_{y,\nu} h^{(1)}) \Sigma_{y, \nu}\bigg) + \eta(t, \nu) y h^{(1)}_y  + m(t, \nu) h^{(1)}_\nu \\
	 = & (x+y)^q F^{(q)}_t + \frac{1}{2}\bigg(\sigma^2(\nu) y^2 q(q-1) (x+y)^{q-2}F^{(q)} + \zeta^2(\nu) (x+y)^q F^{(q)}_{\nu\nu} \\
 & \quad + 2 \rho \sigma(\nu) y q(x+y)^{q-1} F^{(q)}_\nu \bigg)+ \eta(\nu) y q(x+y)^{q-1} F^{(q)}+ m(\nu) (x+y)^{q} F^{(q)}_\nu\\
 	\leq & (x+y)^q \max \limits_{\kappa := \frac{y}{x+y}} \bigg[F^{(q)}_t + \frac{1}{2}\bigg( \sigma^2(\nu) \kappa^2 q(q-1) F^{(q)} + \zeta^2(\nu)  F^{(q)}_{\nu\nu} + 2 \rho \sigma(\nu)  q \kappa F^{(q)}_\nu \bigg)  \\
 & \qquad + \eta(\nu)  q \kappa F^{(q)}+ m(\nu) F^{(q)}_\nu  \bigg]\\
 	 =  & (x+y)^q \bigg[ F^{(q)}_t+ \frac{1}{2} \zeta^2(\nu) F^{(q)}_{\nu\nu} + m(\nu) F^{(q)}_\nu \\ 
 	 &\qquad +\max \limits_{\kappa := \frac{y}{x+y}} \bigg( \frac{1}{2}  \sigma^2(\nu) \kappa^2 q(q-1) F^{(q)}  +  \rho \sigma(\nu)  q \kappa F^{(q)}_\nu   +  \eta(\nu)  q \kappa F^{(q)} \bigg) \bigg]   \\
	 = & 0. \label{equ h esti}
\end{align}
\endgroup
and for $J^{(1)}(y, \nu): = F^{(q)} (t, \nu) $, 
\begin{align} 
 & J^{(1)}_t + \frac{1}{2}Tr\bigg( (\nabla^2_{y,\nu} J^{(1)}) \Sigma_{y, \nu}\bigg) + \eta(\nu) y J^{(1)}_y  + m(\nu) J^{(1)}_\nu \\
 = &  F^{(q)}_t +\frac{1}{2} \zeta^2(\nu) F^{(q)}_{\nu\nu} + m(\nu) F^{(q)}_\nu, \label{equ J esti}
\end{align}
since $\max \limits_{\kappa := \frac{y}{x+y}} \bigg( \frac{1}{2}  \sigma^2( \nu) \kappa^2 q(q-1) F^{(q)}  +  \rho \sigma(\nu)  q \kappa F^{(q)}_\nu   +  \eta( \nu)  q \kappa F^{(q)} \bigg) \geq 0$, we have from \eqref{equ h esti}
that 
$\eqref{equ J esti} \leq 0. $

In summary, since 
\begin{align}
H^{(1)}(t, x, y, \nu) = ((x+y)^q+C) F^{(q)}(t, \nu) = h^{(1)}(t, x, y, \nu) + C J^{(1)}(y, \nu),\\ 
H^{(2)}(s, u, v, \xi) = ((u+v)^q+C) F^{(q)}(s, \xi)= h^{(2)}(s, u, v, \xi) + C J^{(1)}(v, \xi). 
\end{align}
we have 
\begin{align}
\eqref{equ cp final2} \geq  \beta (\psi-\phi) +o(1) = 2\beta\delta +o(1) >0. 
\end{align}
This leads to a contradiction and concludes the proof.\qed

\subsection{Proof of Theorem \ref{thm:viscosity solution}}\label{sect:viscosity}
As indicated in the main body, we only need to show Theorem \ref{thm:viscosity solution}(ii), i.e., $V$ is a viscosity solution. Also, condition c) in Definition \ref{def visco} is a direct result of Theorem \ref{prop:term} and has been proved in \ref{sect:propterm}. Therefore, in the following, we focus on verifying conditions a) and b) in Definition \ref{def visco}.

\subsubsection{Verifying Condition a)}\label{sec con a}
Condition a) is from the following weak dynamic programming principle. 
\begin{proposition}[Weak Dynamic Programming]\label{pro:DP}
Denote $(\hat{X}_{s}, \hat{Y}_s, \hat{\nu}_s)$ as the state processes $(X_s, Y_s, \nu_s)$ starting from $X_t=x$, ${Y}_t = {y}$, ${\nu}_t = {\nu}$ under the portfolio ${\pi}: = (L_s, M_s)_{t\leq s\leq T}$.
For any stopping time $\tau$ taking values within $(t_{p}, T]$, and $(t, {x}, y, \nu) \in (t_p, T)  \times \mathscr{S}$, we have
\begin{align}
V(t, {x}, y, \nu)\leq \sup_{{\pi} \in \mathcal{A}_t(x, y, \nu)} {E}[V^*(\tau, \hat{X}_\tau, \hat{Y}_\tau, \hat{\nu}_\tau)] \notag
\end{align}
and
\begin{align}
V(t, {x}, y, \nu)\geq \sup_{{\pi} \in \mathcal{A}_t(x, y, \nu)} {E}[V_*(\tau, \hat{X}_\tau, \hat{Y}_\tau, \hat{\nu}_\tau)]. \notag
\end{align}
\end{proposition}
The proof of this proposition is identical to \cite{dai2022nonconcave}, then Condition a) is verified by Corollary 5.6 of \cite{bouchard2011weak}. 

It only remains to verify that $V$ is locally uniformly bounded, which we will prove here. Given Assumption \ref{assmpt:1}, on the one hand, $V(t, {x}, y, \nu)$ cannot exceed $C_1+ C_2 {z^p} F^{(p)}(t, \nu)$, which is the solution without transaction costs and terminal utility $C_1+ C_2 w^p$. On the other hand, for any $(t, x, y, \nu) \in (t_p, T)  \times \mathscr{S}$, the investor can immediately liquidate all stock and hold cash until $T$, that generates a natural locally uniformly lower bound for $V$.  

\subsubsection{Verifying Condition b)}
In this part, we want to prove the boundary condition on $z = K$. More precisely, we have the following result. 
\begin{proposition}\label{prop bound con}
For any $t_{p} \leq t, t_0 \leq T$ and $z = x + (1-\theta_1) y^+ - (1+\theta_2) y^-\geq K$, We have 
\begin{numcases}{}
\lim\limits_{(t, x, y, \nu) \to(t_0, \hat{x}, \hat{y}, \hat{\nu})} V(t, x, y, \nu) = U(K), & \text{when} $U(K)> -\infty$, \text{and} $\hat{z} = K$. \notag\\
\lim\limits_{(t, x, y, \nu) \to(t_0, \hat{x}, \hat{y}, \hat{\nu})} V(t, x, y, \nu) - V_{CRRA}(t, x, y, \nu) = 0, & \text{when} $K= 0$, $U(0) = -\infty$ and $\hat{z} = K$. \notag
\end{numcases}
\end{proposition}
\begin{proof}[Proof of Proposition \ref{prop bound con}]

Case 1: $U(K) > -\infty$. 

On the one hand, it is easy to find that 
\begin{align}
\liminf\limits_{(t, x, y, \nu) \to(t_0, \hat{x}, \hat{y}, \hat{\nu})} V(t, x, y, \nu) \geq U(K).\notag
\end{align}
On the other hand, denote by $\hat{V}(t, x, y, \nu)$ the value function for given wealth $z$ at time $t$ and without transaction costs, we then have
\begin{align}
\hat{V}(t, x+(1-\theta_1)y^+- (1+\theta_2)y^-, \nu) \geq V(t, x, y, \nu). \notag
\end{align}
According to the result for non-concave utility maximization without transaction costs (\cite{dai2022nonconcave}), we have
\begin{align}
\limsup\limits_{(t, x, y, \nu) \to(t_0, \hat{x}, \hat{y}, \hat{\nu})} V(t, x, y, \nu) \leq \limsup\limits_{(t, x, y, \nu) \to(t_0, \hat{x}, \hat{y}, \hat{\nu})} \hat{V}(t, x+(1-\theta_1)y^+- (1+\theta_2)y^-, \nu) = K. \notag
\end{align}
Then we have proved this proposition for Case 1.

Case 2: $K = 0$ and $U(K) = -\infty$.

On the one hand, from Assumption \ref{assmpt:1}, we have $U(w) \geq A_1 \frac{w^{\widetilde{p}}-1}{\widetilde{p}}+$ $A_2$ for all $w>0$. Therefore, 
$$ V(t, x, y, \nu) \geq V_{CRRA}(t, x, y, \nu).  $$
On the other hand, still from the result for non-concave utility maximization without transaction costs (\cite{dai2022nonconcave}), we have
\begin{align}
\limsup \limits_{(t, x, y, \nu) \to(t_0, \hat{x}, \hat{y}, \hat{\nu})} V(t, x, y, \nu) - V_{CRRA}(t, x, y, \nu) \leq  0. 
\end{align}
Then we have proved this proposition for Case 2.
\end{proof}

\subsection{Gaussian Mean Return Model}\label{sec GMR ass2.2}
We verify here that the Gaussian mean return model satisfies Assumption \ref{ass:2}.2. 

We observe from the dynamics of $\nu$ that 
$$d \nu_s = \kappa (\bar{\nu}-\nu_s) dt + \zeta d \mathcal{B}^x_s \leq \kappa (|\bar{\nu}|+ |\nu| + N) dt + \zeta d \mathcal{B}^x_s, \ \text{when $s\leq \tau_{N}$},  $$
where $\tau_N$ is  the first hitting time of $\nu_s  = \nu +N$. That is, when $s\leq \tau_{N}$, 
\begin{align}
\nu_s -\nu\leq \kappa (|\bar{\nu}|+ |\nu| + N) (s-t) + \zeta(\mathcal{B}^x_s -\mathcal{B}^x_t ) \leq \kappa (|\bar{\nu}|+ |\nu| + N) (T-t) + \zeta(\mathcal{B}^x_s -\mathcal{B}^x_t ).  
\end{align}
Without loss of generality, we can set $T-t_{\nu, \bar{N}}$ sufficiently small, such that $\kappa (|\bar{\nu}|+ |\nu| + N) (T-t) \leq \frac{N}{2}$, $\forall t \in [t_{\nu, \bar{N}}, T]$. 
In this case, 
\begin{align}
\mathbb{P} \left( \max \limits_{t\leq s \leq T} \nu_s -\nu \geq N \bigg| \nu_t = \nu \right) \leq \mathbb{P} \left( \max \limits_{t\leq s \leq T} \zeta(\mathcal{B}^x_s -\mathcal{B}^x_t ) \geq \frac{N}{2} \bigg| \nu_t = \nu \right) = 2 \Phi\bigg( -\frac{N/2}{\zeta \sqrt{T-t}} \bigg).  
\end{align}
From \eqref{equ esti cdfn}, 
\begin{align}
\Phi\bigg( - \frac{N/2}{\zeta \sqrt{T-t}} \bigg) \leq C \frac{\sqrt{T-t}}{N} e^{-d \frac{N^2}{T-t}} \leq C \frac{\sqrt{T}}{N} e^{-d \frac{N^2}{T-t}}
\end{align}
for some constants $C, d >0$.  
Then we have our estimate for $\max \limits_{t\leq s \leq T} \nu_s -\nu \geq N$. The estimate for  $\min \limits_{t\leq s \leq T} \nu_s -\nu \leq -N$ is similar.

\subsection{Numerical Procedure}\label{appendix:numScheme}
Define the change of variable $z=x+(1-\theta_1)y^+-(1+\theta_2)y^-$ as in \eqref{eqn:z}, $v=\sqrt{T-t}\cdot y$, $W(t,z,v,\nu)=V(t,x,y,\nu)$. Under this transformation, the HJB equation \eqref{equ PDE main} becomes 
\begin{align*}
	\min\Bigg\{-W_t-\tilde{\mathcal{L}}_{-\theta_1}W,~W_v,~(\theta_1+\theta_2)W_z-\sqrt{T-t}\cdot W_v\Bigg\}&=0,~~\text{for }v\ge 0,\\
	\min\Bigg\{-W_t-\tilde{\mathcal{L}}_{\theta_2}W,~(\theta_1+\theta_2)W_z+\sqrt{T-t}\cdot W_v,~-W_v\Bigg\}&=0,~~\text{for }v<0,
\end{align*}
where
\begin{align*}
\tilde{\mathcal{L}}_{\theta} W &= \frac{1}{2}\sigma^2v^2\left(W_{vv}+\frac{2(1+\theta)}{\sqrt{T-t}}W_{zv}+\frac{(1+\theta)^2}{T-t}W_{zz}\right)+ \frac{1}{2}\zeta^2 W_{\nu\nu}+\rho\sigma\zeta v\frac{1+\theta}{\sqrt{T-t}}W_{z\nu}\\
&\quad +\rho\sigma\zeta v W_{v\nu} + (r+\sigma\nu)v\frac{1+\theta}{\sqrt{T-t}}W_z+\left(r+\sigma \nu-\frac{1}{2(T-t)}\right)vW_v+\kappa (\bar{\nu}-\nu)W_{\nu}.
\end{align*}
Note that in the case of the geometric Brownian motion model, all derivatives with respect to $\nu$ vanish in the above equation, the operator $\tilde{\mathcal{L}}_\theta$ reduces to
\begin{align*}
\tilde{\mathcal{L}}_{\theta}W = \frac{1}{2}\sigma^2 v^2 \left( W_{vv}+\frac{2(1+\theta)}{\sqrt{T-t}} W_{vz}+\frac{(1+\theta)^2}{T-t} W_{zz}\right)-\left(\frac{1}{2(T-t)}-\eta\right) v W_v +\eta(1+\theta) \frac{v}{\sqrt{T-t}} W_z,
\end{align*}
and the problem dimension reduces to two.

We then solve the above variational inequalities numerically via the penalty method (c.f. \cite{dai2010penalty}). The corresponding penalty formulation is
\begin{align*}
	W_t+\tilde{\mathcal{L}}_{-\theta_1}W+\lambda (-W_v)^++\lambda \left(\sqrt{T-t}\cdot W_v-(\theta_1+\theta_2)W_z\right)^+&=0,~~\text{for }v\ge 0,\\
	W_t+\tilde{\mathcal{L}}_{\theta_2}W+\lambda (W_v)^++\lambda \left(-\sqrt{T-t}\cdot W_v-(\theta_1+\theta_2)W_z\right)^+&=0,~~\text{for }v< 0,
\end{align*}
where the penalty constant $\lambda>0$ is a large number. The nonlinear terms $\left(\sqrt{T-t}\cdot W_v-(\theta_1+\theta_2)W_z\right)^+$ and $\left(-\sqrt{T-t}\cdot W_v-(\theta_1+\theta_2)W_z\right)^+$ are linearized using the non-smooth Newton iteration (c.f. \cite{forsyth2002quadratic}), and the linearized equations are solved using the implicit finite-difference scheme.
The convergence of the penalty scheme for the above HJB quasi variational inequalities follows from a similar argument as \cite{dai2010penalty}, which also tackles a utility maximization problem with proportional transaction costs. 
Following \cite{dai2010penalty}, we first assume the finite difference scheme generates 
an $M$-matrix. Then, we can prove the stability of the numerical scheme as Proposition 5.1 of \cite{dai2010penalty}, with which the convergence of the value function as $\lambda \to \infty$ can be guaranteed. 
The convergence of the nonlinear iteration from the penalty term can also be proved as Proposition 5.3 of \cite{dai2010penalty} due to the similar monotonicity in the iteration process.

For the boundary conditions, in the case of the goal-reaching problem, we set $W(t,0,v) = 0$ due to bankruptcy, and $W(t,1,v) = 1$ since the goal is reached by liquidating the whole risky asset position. As $|v|\to \infty$, we impose the boundary condition $W(t,z,v) = z$.
If short-selling of risky assets is prohibited, then only the equation in the region $v\ge 0$ remains, and the buy strategy is imposed on $v=0$. In the case of aspiration utility and S-shaped utility problems, when $z$ is very large, the problem is asymptotically a classic Merton optimal investment problem with proportional transaction costs (up to a shifting and scaling in the utility function). Therefore, we set a Dirichlet boundary condition at a large value of $z$ so that $W$ equals the classic Merton problem with transaction costs up to the same shifting and scaling.

\subsection{Additional Numerical Results}
\subsubsection{Goal-Reaching Problem with Short-selling Prohibited}
We consider the goal-reaching problem with a positive risk premium, but with the additional constraint that short-selling the risky asset is prohibited.  The action regions illustrated in Figure \ref{fig: region} resemble the regions within $y>0$ in the right panels of Figure \ref{fig goal reaching etas}. Therefore, adding a short-selling constraint does not change our findings qualitatively.

\begin{figure}[htbp]
	\begin{center}
		\includegraphics[width=0.48\textwidth]{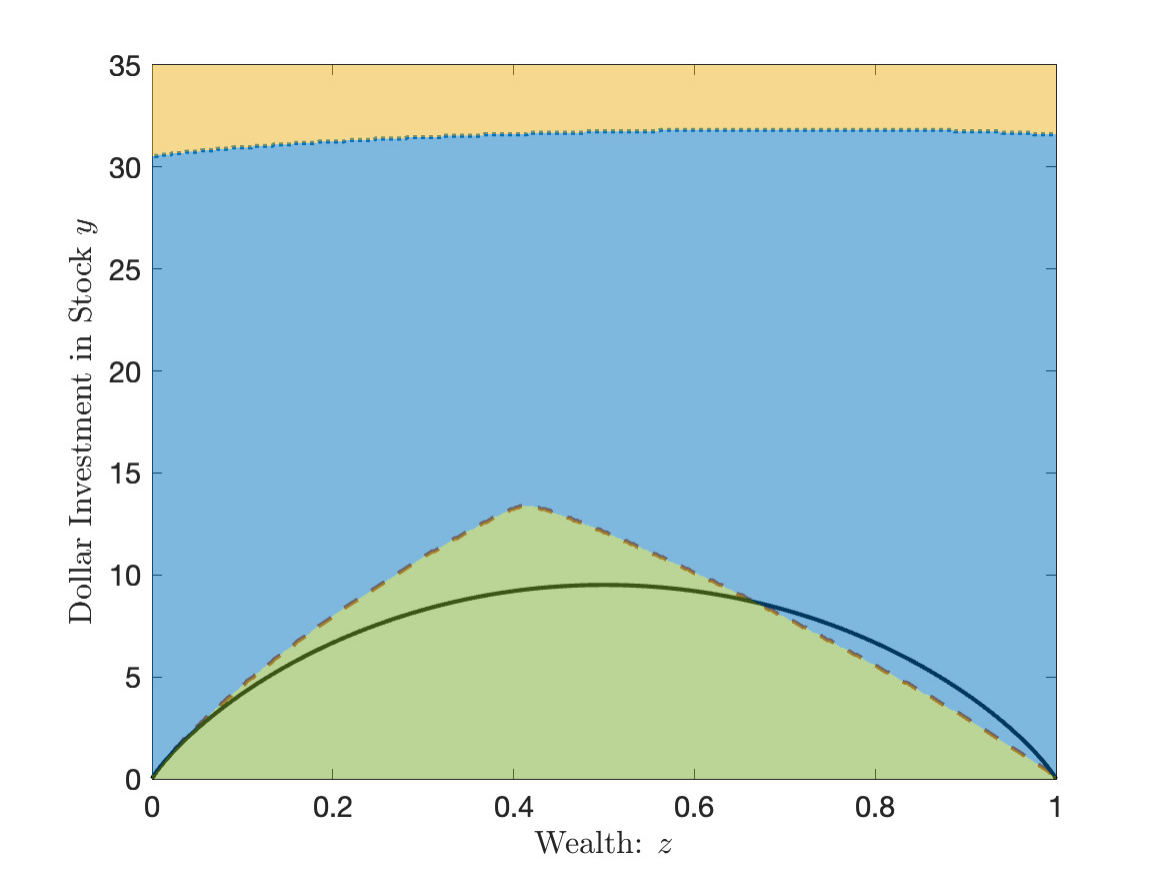}
		\caption{\footnotesize { Action regions for the goal-reaching problem with short-selling prohibited, with time to maturity $T-t=0.02$. Yellow: sell region; Blue: no-trading region; Green: buy region; Solid line: target position without transaction costs. Parameters: $\theta_1=\theta_2=10^{-2},\sigma=0.3,\eta=0.04.$
		}
		}
		\label{fig: region}
	\end{center}
\end{figure}

\subsubsection{The S-Shaped Utility}
We present the results for the S-shaped utility \eqref{eqn:Sutil}, with parameters $\lambda=2.25,p=0.5,z_0=1$, and with short-selling allowed. Due to the similarity with the above results, we only present the action region at $T-t=0.01$ in Figure \ref{fig:SS}. We see that the shape of the action region resembles that of the aspiration utility, and it can still be optimal to take large negative leverage despite the positive risk premium. The only major difference is that when $z$ is close to 0, the investor does not actively reduce the risky position to avoid bankruptcy, as in the aspiration utility. The reason is that, unlike the aspiration utility where the investor is risk-averse for $z<z_0$, here the investor is risk-seeking, and therefore she would gamble for a smaller loss rather than trying to reduce the stock position to avoid bankruptcy.

\begin{figure}[ht!]
\centering
\includegraphics[width=0.48\textwidth]{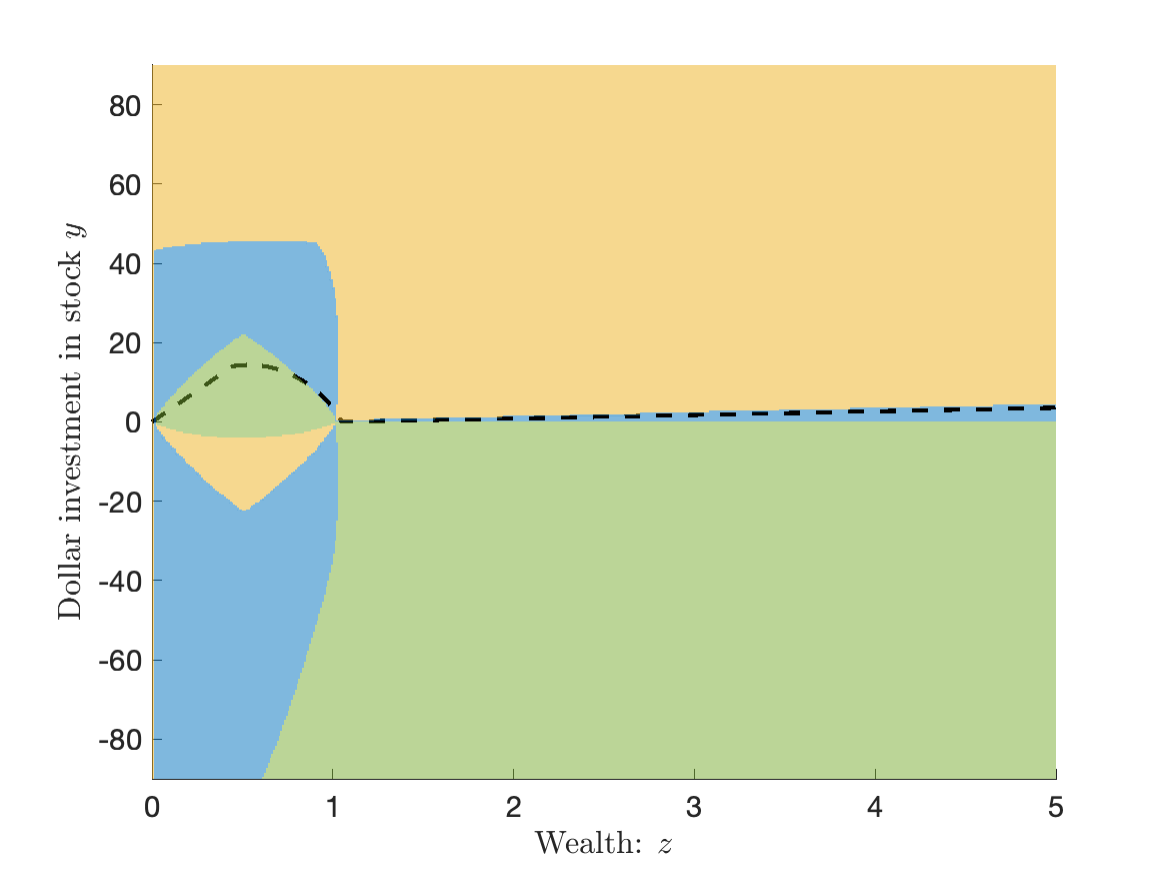}
\caption{The action regions for the S-shaped utility problem without short-selling constraint. Yellow: sell region; Green: buy region; Blue: no-trading region. Dashed line: target position without transaction cost. $T-t=0.01$. Parameters: $\theta_1=\theta_2=10^{-3},\sigma=0.3,\eta=0.04.$}\label{fig:SS}
\end{figure}

\subsubsection{Effect of Jump Size of the Aspiration Utility}\label{appendix:aspiration_jump}
As discussed in Section \ref{subsect:asp}, the action regions of aspiration utility show the goal-reaching incentive in the region where $z$ is below $\bar{z}=1$ but away from 0, so as to reach $\bar{z}$ for the sudden boost in the utility. On the other hand, the risk-averse incentive due to the CRRA-type utility in $z<\bar{z}$ reduces the leverage when $z$ is sufficiently close to zero to preserve wealth and avoid bankruptcy. Intuitively, for an aspiration utility with a larger jump size at $\bar{z}$, the goal-reaching incentive should be stronger, and therefore, the action regions should be closer to those in the goal-reaching case. In the following, we illustrate this effect. 

To this end, we take the same parameter for the aspiration utility as in Section \ref{subsect:asp} with $c_1=0$ (with the jump size of 1 at $\bar{z}$) and consider an alternative case with $c_1=1$ (with the jump size of 2 at $\bar{z}$). The results in Figure \ref{fig:aspiration_jump} confirm the above intuition. Indeed, compared to the smaller jump-size case (left panel), the no-trading region in the larger jump-size case (right panel) expands and is closer to the goal-reaching case (top right panel of Figure \ref{fig goal reaching etas}). 

\begin{figure}[ht!]
\centering
\includegraphics[width=0.48\textwidth]{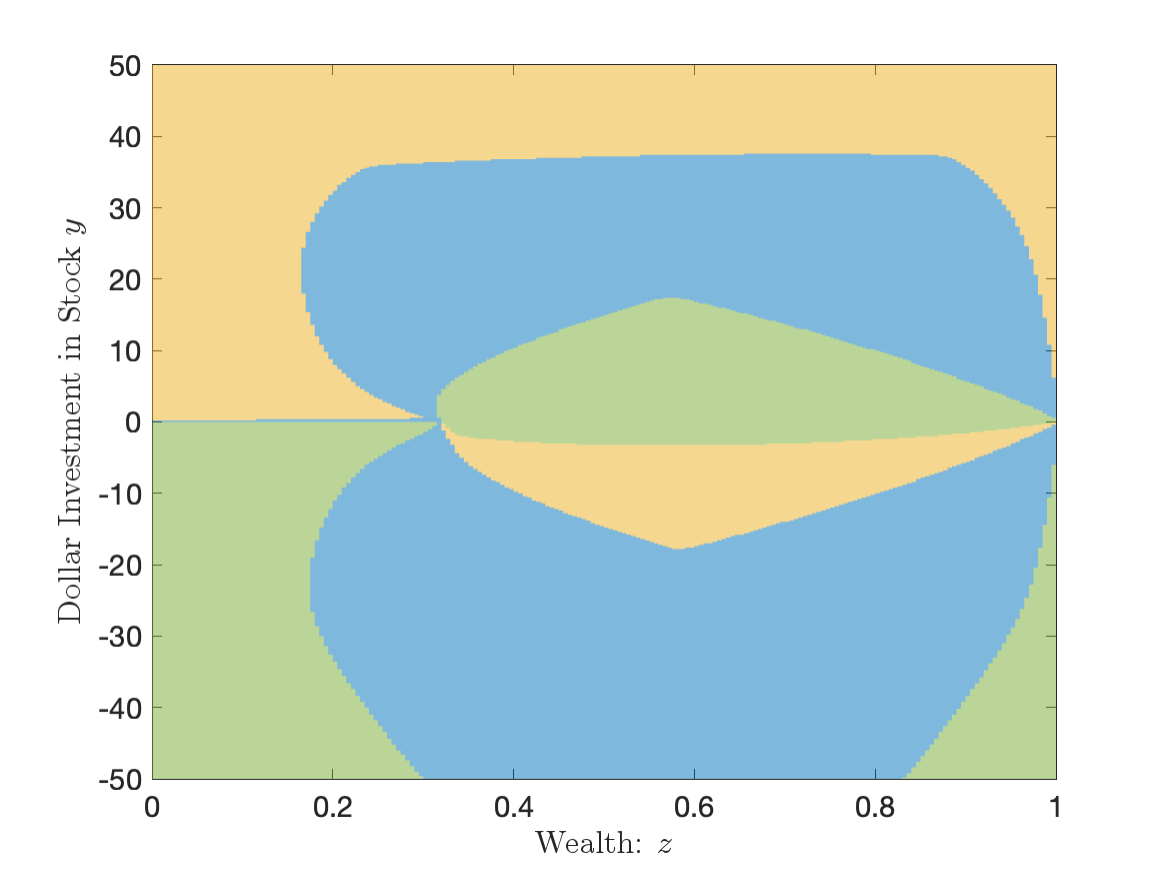}
\includegraphics[width=0.48\textwidth]{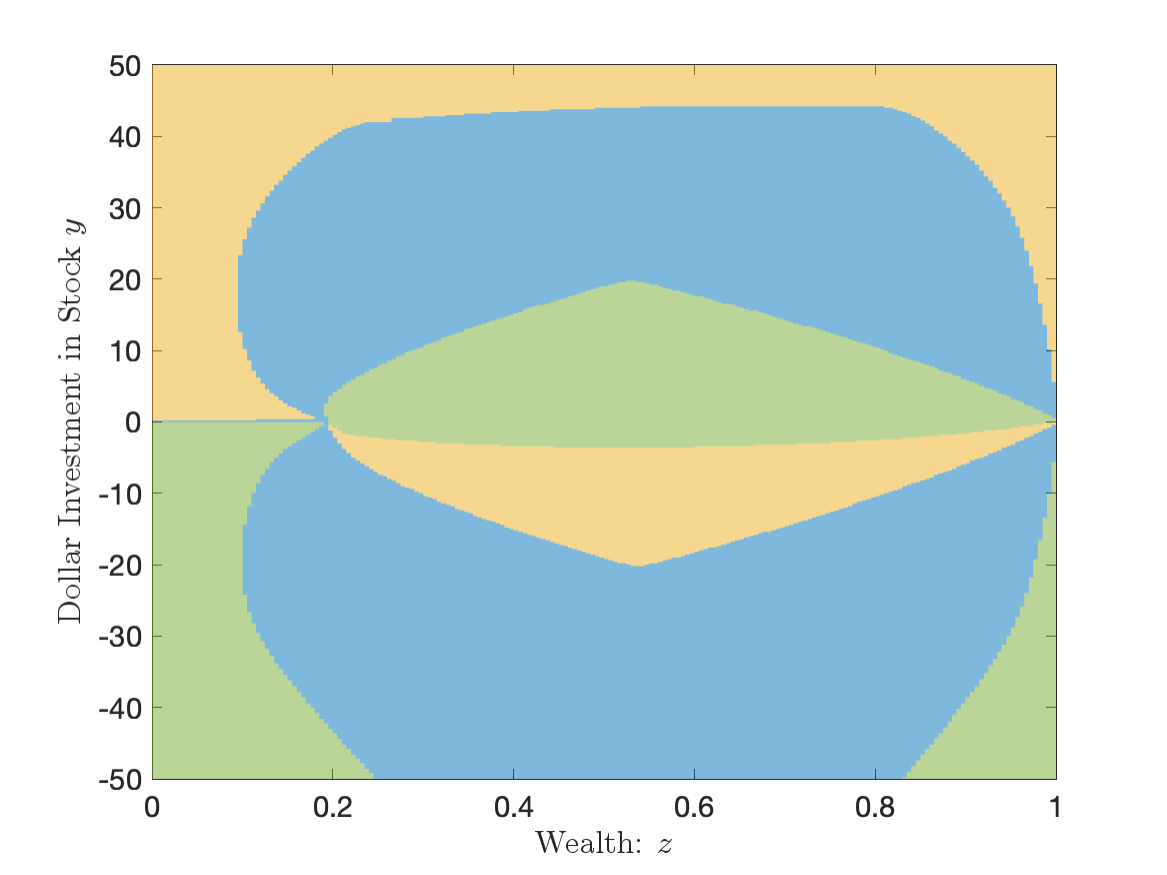}
\caption{The action regions of aspiration utility with $c_1=0$ (Left Panel) and $c_1=1$ (Right Panel). Yellow: sell region; Green: buy region; Blue: no-trading region. $T-t=0.01$. Parameters: $\theta_1=\theta_2=10^{-3},\sigma=0.3,\eta=0.04,p=0.5,c_2=1.5.$}\label{fig:aspiration_jump}
\end{figure}

\subsubsection{The Aspiration Utility without Transaction Costs}\label{appendix:kink}
As shown by the dashed lines in Figure \ref{fig aspi strategy}, even without transaction costs, the optimal target position under the aspiration utility exhibits an interesting shape, where the investment in stock peaks around $z=0.55$ and reduces to 0 as $z$ increases to $\bar{z}=1$. While this observation is consistent with that in Figure 1(f) in \cite{dai2022nonconcave}, we provide some intuitions in the following.

For illustration, we take an aspiration utility with $p=0.5$ and $c_2=0.5$, and take $c_1=1$, 1.2, and 2. We take $\theta_1=\theta_2=0$, $\bar{z}=1$, and $T-t=0.01$. The left panel in Figure \ref{fig:aspiC1} plots the utility functions, where for $c_1=1$, the utility function is continuous in $z$ but its first derivative jump at $z=1$, and for $c_1=1.2$ and 2, the utility function jumps at $z=1$, with a larger jump size under for $c_1=2$.

\begin{figure}[htbp!]
\includegraphics[width=0.45\textwidth]{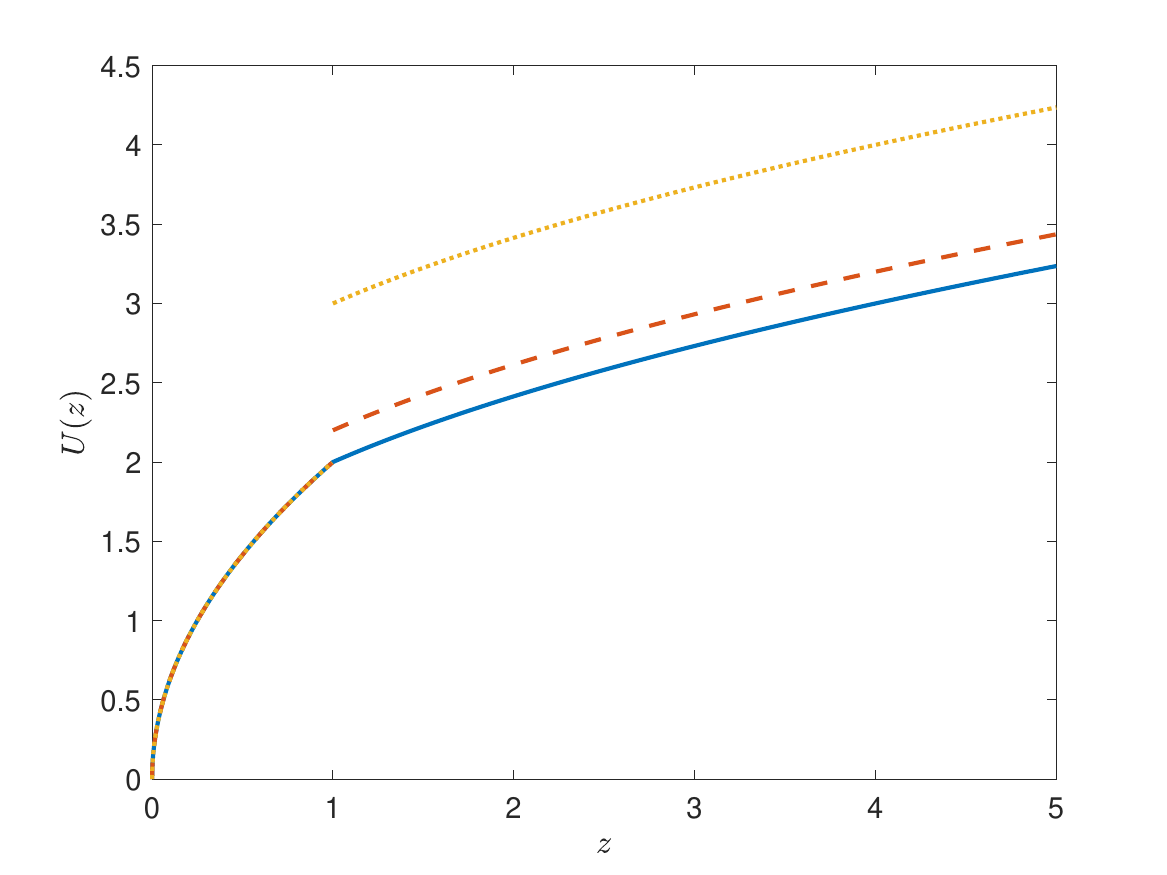}
\includegraphics[width=0.45\textwidth]{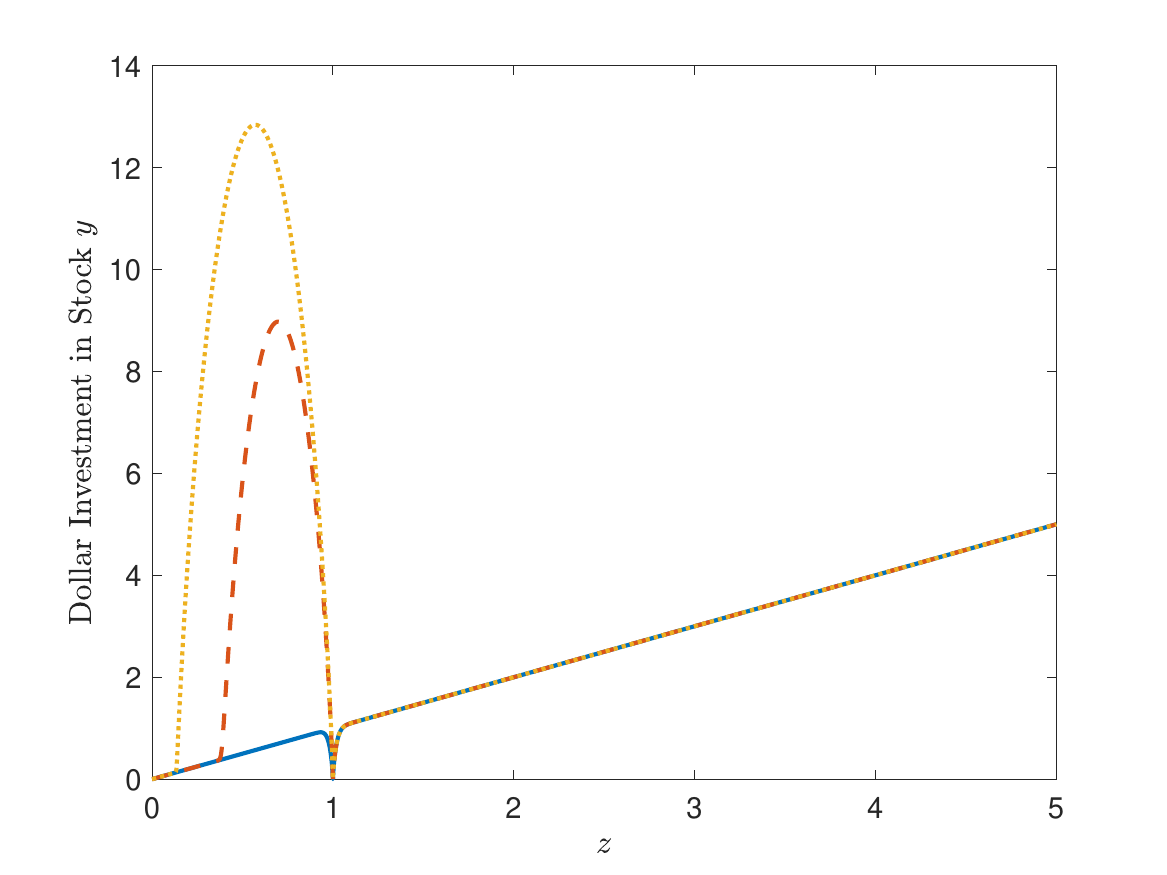}
\caption{\textbf{Left:} Aspiration utility functions; \textbf{Right:} the optimal risky asset $y$ under the aspiration utility without transaction costs. Parameters for the aspiration utility: $p=0.5$, $c_2=0.5$, $c_1=1$ (solid line), 1.2 (dashed line), and 2 (dotted line). Parameters: $\theta_1=\theta_2=0,T-t=0.01$.}\label{fig:aspiC1}
\end{figure}

The optimal investments in stock $y$, as shown in the right panel, share a drop in $y$ to 0 around $z=1$. Intuitively, this is due to the jump of the first derivative of the utility function. On the other hand, the local increase in the optimal $y$ below $z=1$ but away from $z=0$ is positively related to the jump size in the utility at $z=1$. Intuitively, this is due to betting for $z$ to reach 1 for the sudden boost in the utility. However, as $z$ approaches 0, the optimal $y$ reduces to 0. Indeed, in this case, the investor reduces the leverage to preserve wealth and avoid bankruptcy.

\subsubsection{Verification of the Terminal Condition}\label{appendix:terminal}
In this section, we numerically verify the terminal condition \eqref{equ: tercon} for various problems studied in Section \ref{sect:numerical}. To illustrate, we plot the left-hand side difference in \eqref{equ: tercon} against a range of the wealth $z$ while fixing the dollar investment in stock at a certain level $y$, for various time to maturity $T-t$. The terminal condition \eqref{equ: tercon} predicts that the difference should converge to 0 in a pointwise manner as $T-t$ goes to 0.

\subsubsection*{Goal-Reaching Problem}
We verify the terminal condition \eqref{equ: tercon} for the goal-reaching problem with short-selling constraint in Figure \ref{fig com asymp bd cond}, without short-selling constraint in Figure \ref{fig com asymp bd cond pm}, and under the Gaussian mean return model in Figure \ref{fig term veri nu}. Since the goal-reaching utility jumps at $z = 1$, the difference jumps from 1 to 0 at $z = 1-$. This figure confirms that the difference converges to 0 in a pointwise manner. However, such convergence is not uniform, as the maximum difference is always 1. 

\begin{figure}[hptb!]
\centering
\includegraphics[width=0.45\textwidth]{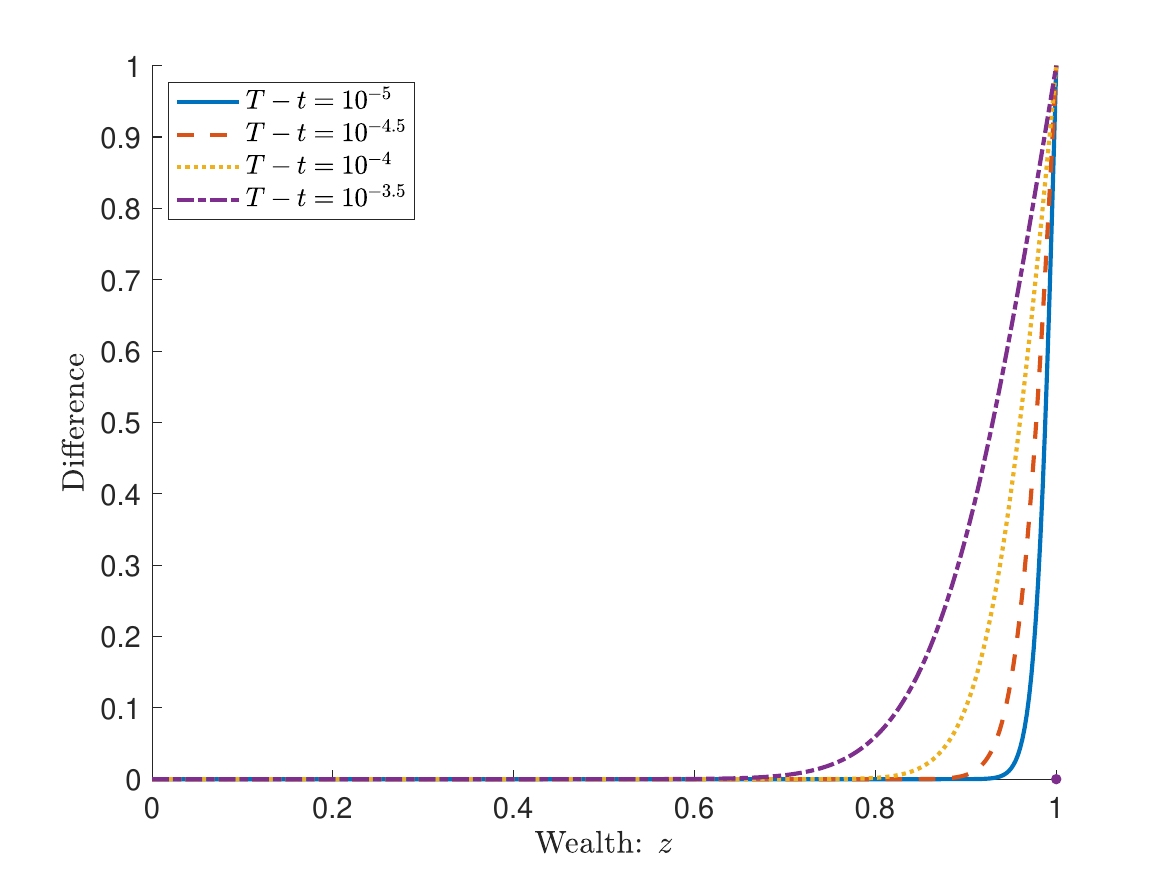}
\caption{Difference between value function and asymptotic expression for goal-reaching problem with short-selling constraint and $y=20$. Parameters: $\theta_1=\theta_2=10^{-2},\sigma=0.3,\eta=0.04.$
} \label{fig com asymp bd cond}
\includegraphics[width=0.45\textwidth]{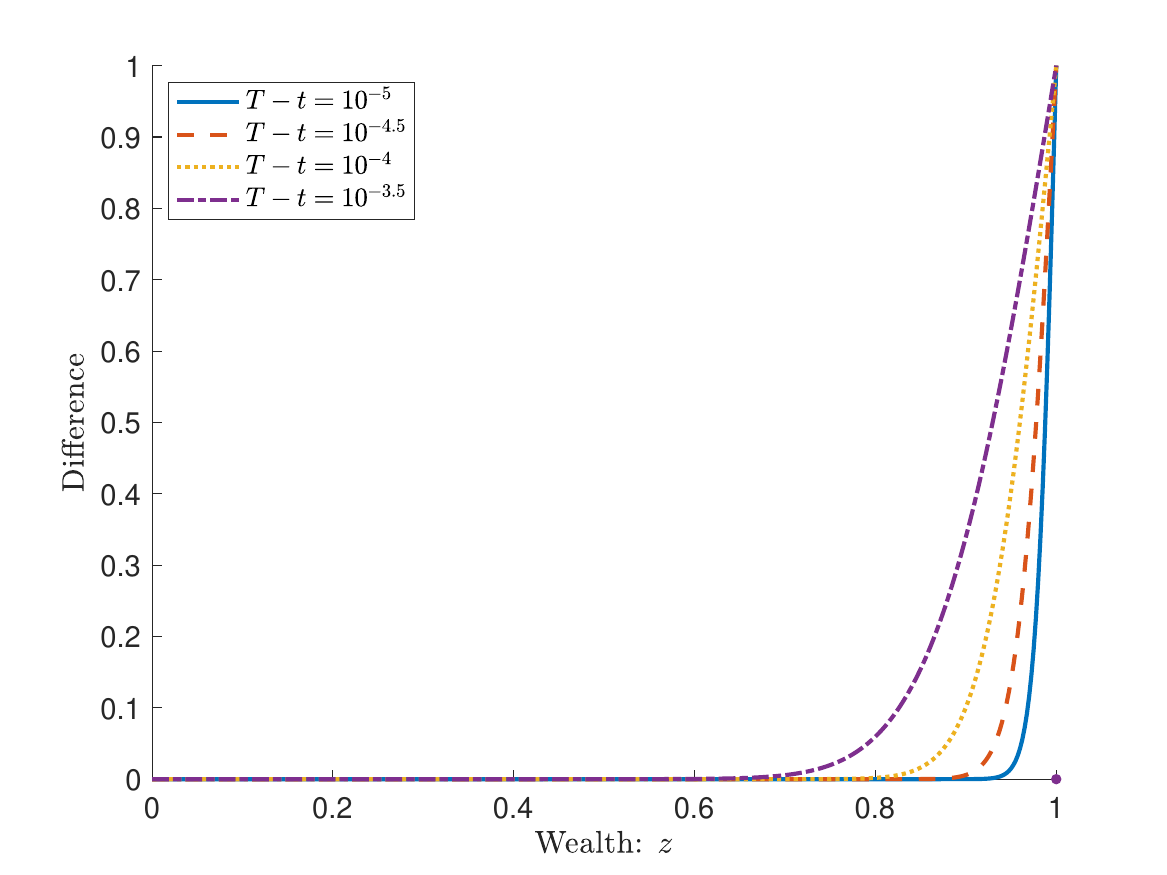}
\includegraphics[width=0.45\textwidth]{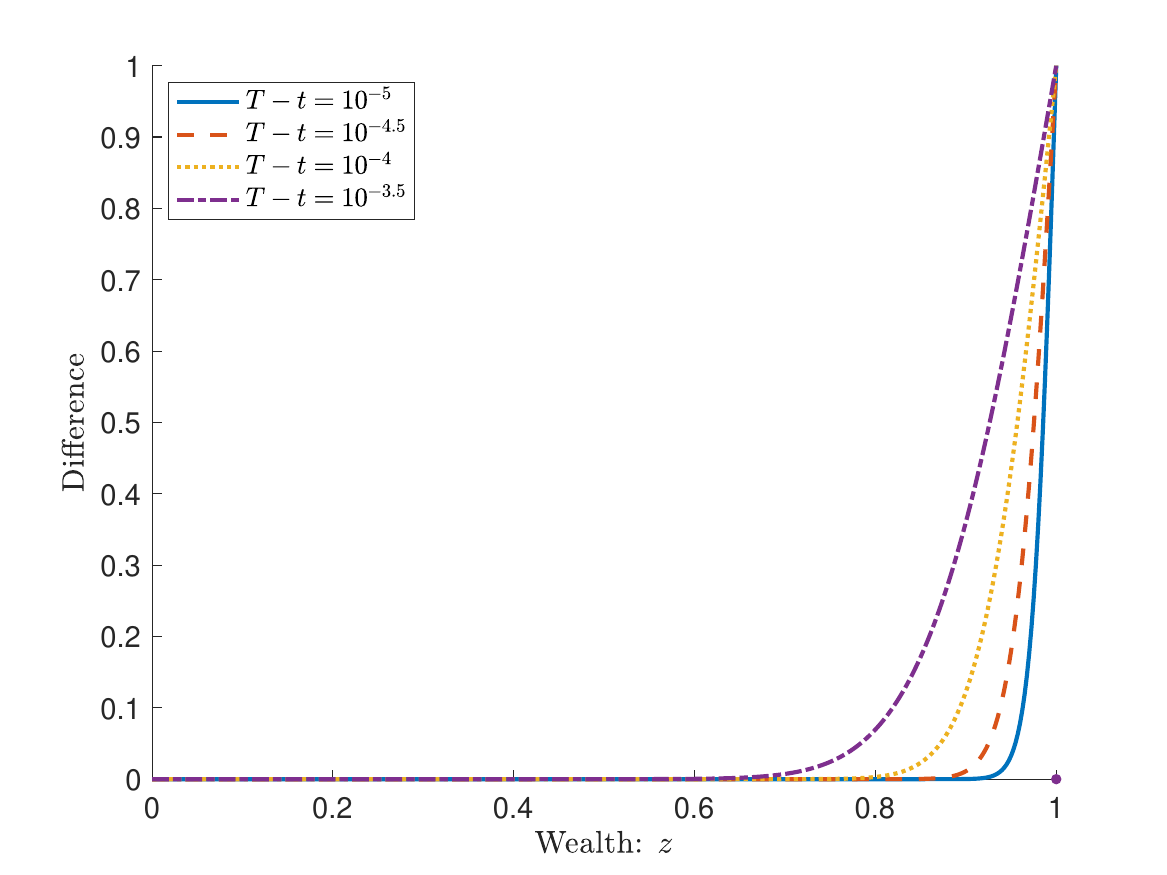}
\caption{Difference between the value function and asymptotic expression for goal-reaching problem without short-selling constraint at $y=20$ (left figure) and $y=-20$ (right figure). Parameters: $\theta_1=\theta_2=10^{-2},\sigma=0.3,\eta=0.04$.
} \label{fig com asymp bd cond pm}
\includegraphics[width=0.45\textwidth]{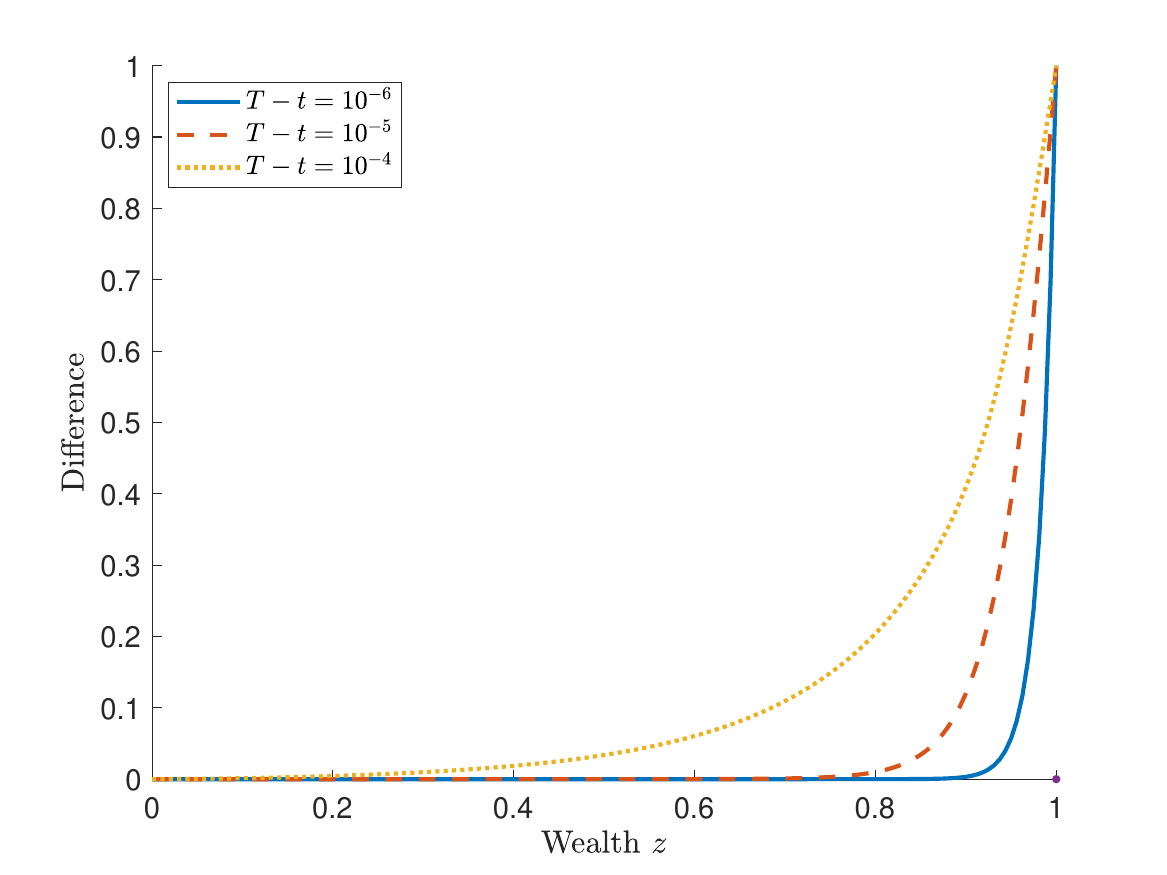}
\includegraphics[width=0.45\textwidth]{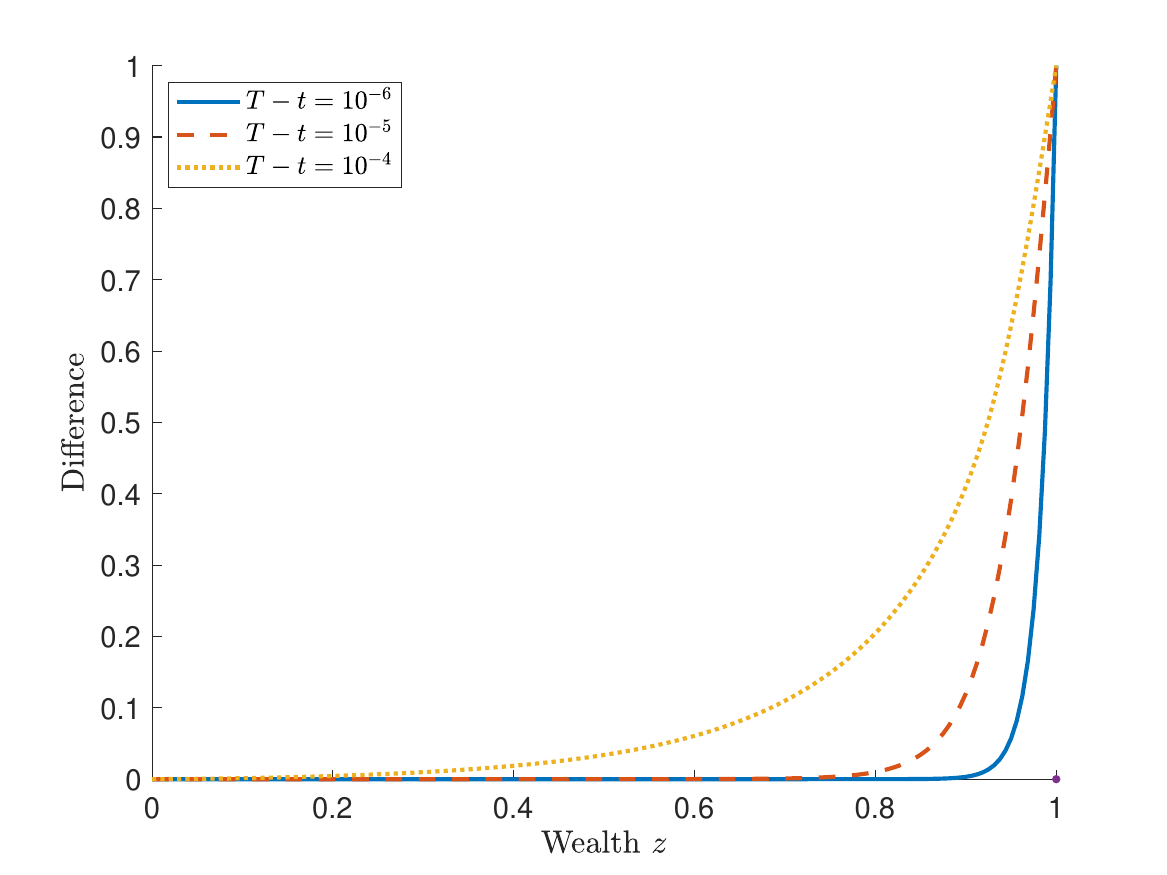}
\caption{Difference between the value function and asymptotic expression for the goal-reaching problem at $y=1$, for $\nu=-0.1333$ (left figure) and $\nu=0.1333$ (right figure). Parameters: $\theta_1=\theta_2=10^{-3},\sigma=0.3,\eta=0.04.$
} \label{fig term veri nu}  
\end{figure}

\subsubsection*{Other Utility}
We verify the terminal condition \eqref{equ: tercon} for the aspiration utility and S-shaped utility in Figure \ref{fig term veri} and \ref{fig term veri S-shaped}, respectively. 

\begin{figure}[hptb!]
\centering
\includegraphics[width=0.45\textwidth]{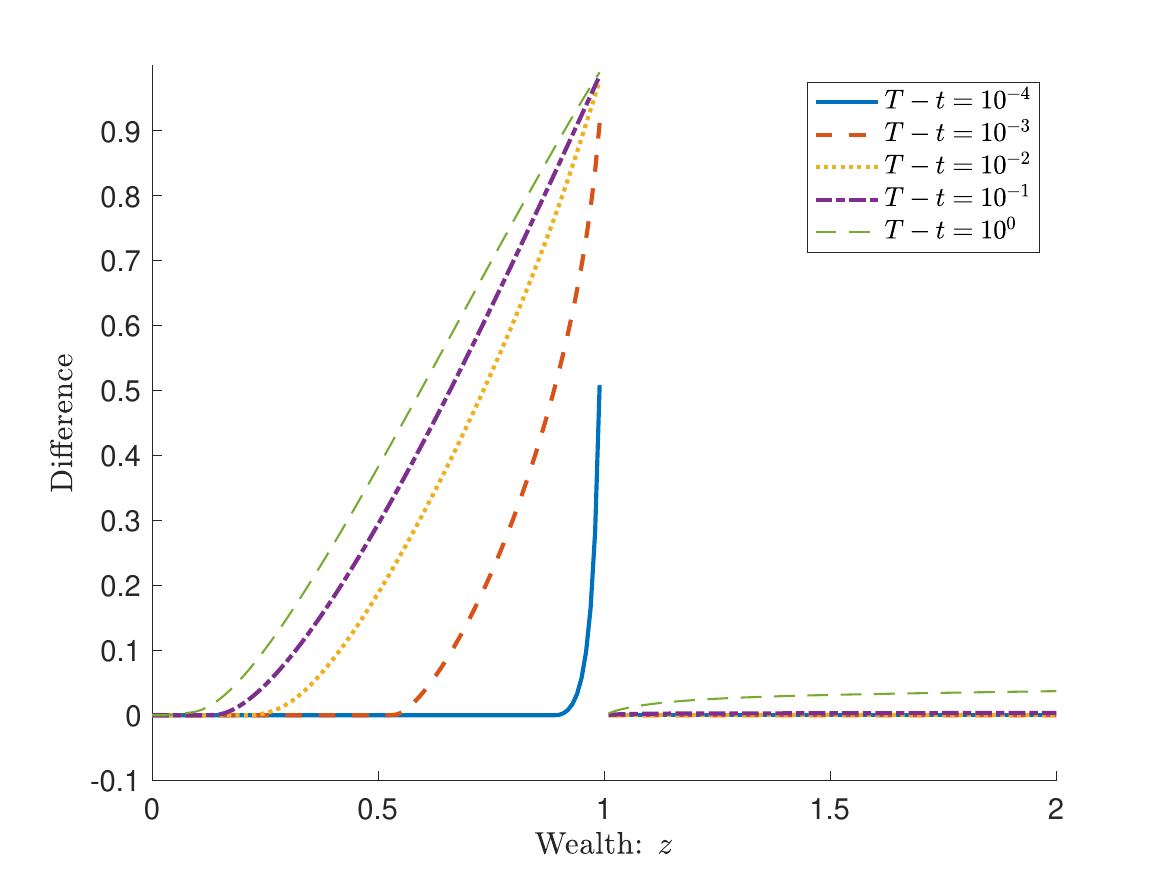}
\includegraphics[width=0.45\textwidth]{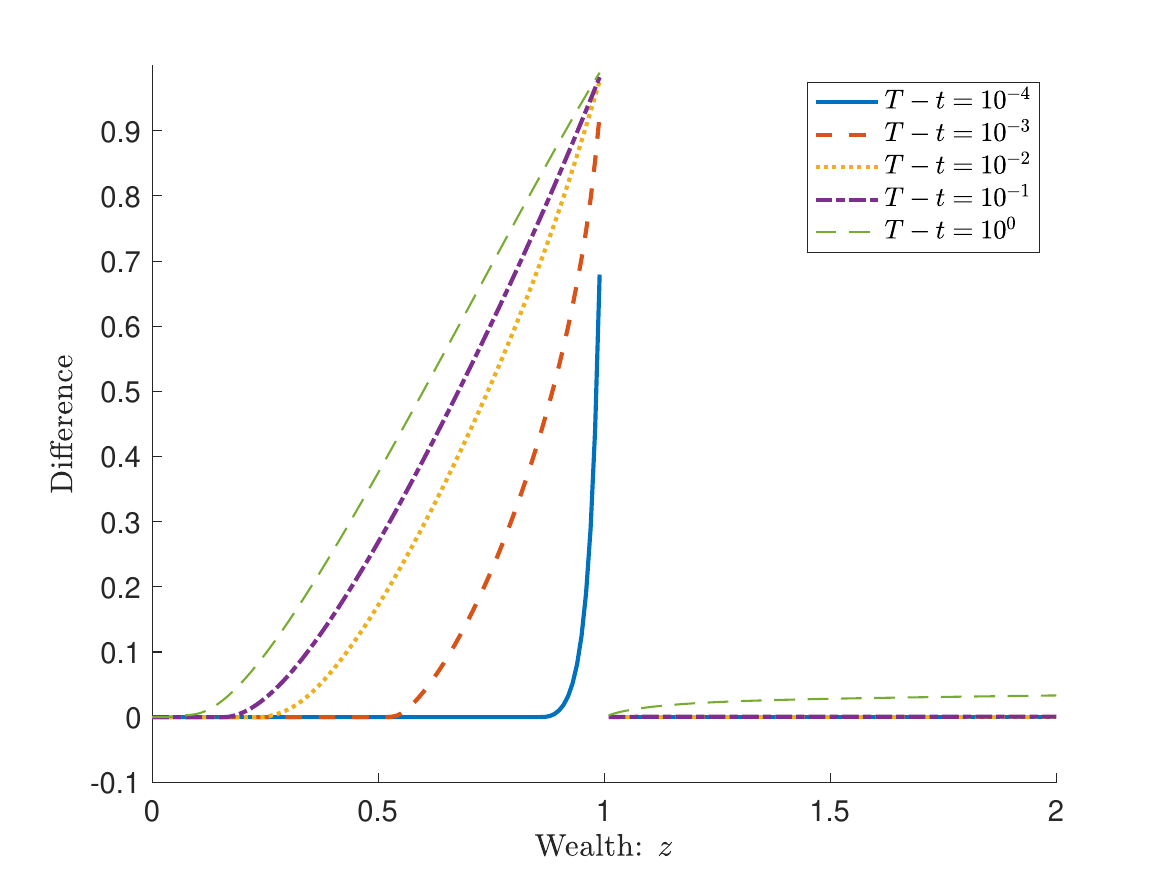}
\caption{Difference between the value function and asymptotic expression for aspiration utility at $y=5$ (left figure) and $y=-5$ (right figure). Parameters: $\theta_1=\theta_2=10^{-3},\sigma=0.3,\eta=0.04.$
} \label{fig term veri}  
\includegraphics[width=0.45\textwidth]{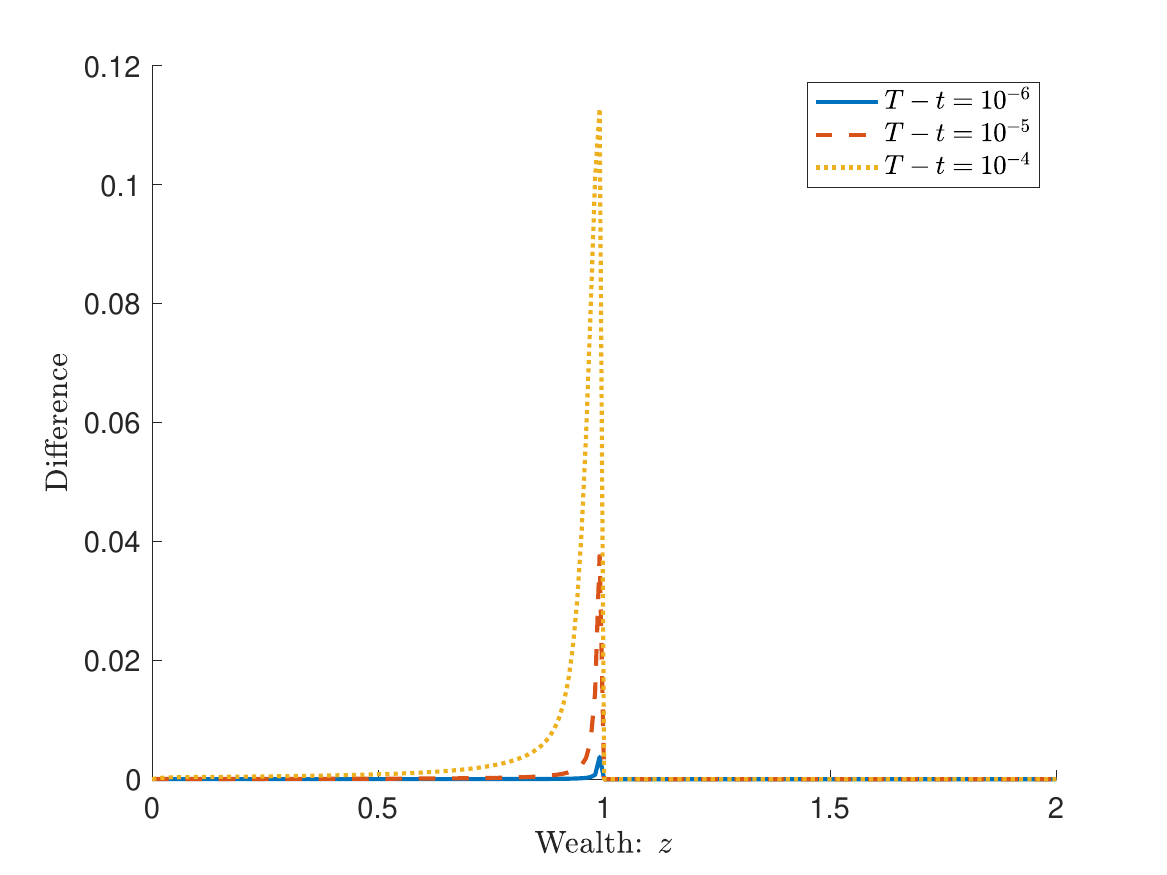}
\caption{Difference between the value function and asymptotic expression for the S-shaped utility at $y=10$. Parameters: $\theta_1=\theta_2=10^{-3},\sigma=0.3,\eta=0.04.$
}\label{fig term veri S-shaped}  
\end{figure}

\end{document}